\documentclass[12pt]{article}
\usepackage{amsmath}
\usepackage{graphicx}
\usepackage{enumerate}
\usepackage{natbib}
\usepackage{enumitem}
\usepackage{mathrsfs}
\usepackage{microtype}
\usepackage[sectionbib]{bibunits}
\defaultbibliographystyle{apalike} 
\defaultbibliography{bibliography2}
\usepackage[title]{appendix}
\input{commands.tex}
\usepackage{url} 
\newcommand*{\RRd}{\mathbb{R}^d}
\newcommand*{\Ll}{\mathscr{L}}

\newcommand{\blind}{1}

\begin{document}

\def\spacingset#1{\renewcommand{\baselinestretch}%
{#1}\small\normalsize} \spacingset{1}


\if1\blind
{
  \title{\bf On model-based clustering with entropic optimal transport}
  \author{Gonzalo Mena\hspace{.2cm}\\
    Department of Statistics and Data Science, Carnegie Mellon University\\
    }
  \maketitle
} \fi

\if0\blind
{
  \bigskip
  \bigskip
  \bigskip
  \begin{center}
    {\LARGE\bf On model-based clustering with entropic optimal transport}
\end{center}
  \medskip
} \fi

\bigskip
\begin{abstract}
We develop a new methodology for model-based clustering. Optimizing the log-likelihood provides a principled statistical framework for clustering, with solutions found via the EM algorithm. However, because the log-likelihood is nonconvex, convergence to stationary points is all that can be guaranteed, and practitioners often use multiple starting points in the hope that one will converge to the global solution. We consider a new loss function based on entropic optimal transport with the same global optimum as the log-likelihood, but a much better-behaved landscape so that spurious local optima configurations known to be pervasive for the log-likelihood are avoided. Similar to the EM algorithm for the log-likelihood, this new loss can be optimized by the Sinkhorn-EM algorithm, which we show converges at a rate comparable to that of EM. By analyzing extensive numerical experiments and two real-world applications in image segmentation in C. elegans microscopy and clustering in spatial transcriptomics, we show that this new loss outperforms log-likelihood optimization, indicating that it represents a valuable clustering methodology for practitioners. \end{abstract}

\noindent%
\vfill

\newpage
\spacingset{1.1} 
\begin{bibunit}
\section{Introduction}


Cluster analysis is one of the most ubiquitous tasks in statistics and data science, as the need to determine latent group structure arises in myriad applications, including bioinformatics \citep{cheng2000biclustering} and social network analysis \citep{handcock2007model}, among many others. The mainstream statistical approach to clustering is the so-called model-based approach \citep{mclachlan19829,bouveyron2019model}, which hinges on the inference of an underlying probabilistic generative model that typically takes the form of a finite \textit{mixture} of distributions, where each component represents an observation model for samples within a cluster.


The primary inferential approach for model-based clustering is to perform maximum likelihood estimation on the mixture model, and the Expectation Maximization (EM) algorithm  \citep{Dempster1977} is the most widely used computational tool for such inference, as it exploits the inherent latent structure in mixture models. However, the log-likelihood in mixture models is typically a non-convex function, and the EM algorithm often converges to spurious stationary points \citep{biernacki2003choosing}. To address this issue, practitioners resort to heuristics such as choosing sensible initializations (e.g., k-means++) or running the algorithm with several seeds and keeping the one with the highest likelihood \citep{biernacki2003choosing}. While practical, these heuristics are not guaranteed to work and may lead to the systematic selection of suboptimal solutions \citep{steinley2003local,franti2019much}.


We propose a new methodology for model-based clustering by incorporating tools from Optimal Transport  and show that they mitigate the risk of being trapped in suboptimal solutions. Optimal transport (OT) \citep{Villani2008} is a mathematical framework that provides a rich measure of the distance between distributions and has attracted considerable attention in statistics, computer vision, machine learning, and related fields (see, e.g., \citet{Peyre2019} for a survey). We focus on \textit{entropic optimal transport}, a variant of the original optimal transport problem that incorporates an entropic penalty term. This formulation possesses computational \citep{Cuturi2013} and statistical advantages over its unregularized version \citep{MenaNilesWeed2019}.


This paper is organized as follows: in Section \ref{sec:framework}, we develop a parameter-estimation framework for mixture models by optimizing an entropic optimal-transport-derived loss and describe some basic properties. In Section \ref{sec:em}, we describe an algorithm for optimization of this loss, paralleling the EM algorithm for the log-likelihood and establishing elementary convergence properties. In Section \ref{sec:local}, we show that our methodology can avoid bad stationary points in cases where the log-likelihood cannot. In Section \ref{sec:twogaussians}, we give an in-depth convergence and convergence analysis in the elementary mixture of two Gaussians case. In Section \ref{sec:experiments}, we illustrate the benefits of our method on simulated data. Then, in Sections \ref{sec:celegans} and \ref{sec:coclustering}, we describe two real-world examples. Finally, in Section \ref{sec:disc}, we outline future directions and discuss the limitations of our methodology.

\subsection{Related Work}


Our work contributes to scholarship on model-based clustering methodology \citep{mclachlan19829,bouveyron2019model}. Although not the main focus, our theoretical results connect to recent theoretical analysis of convergence for the EM algorithm seeking to obtain rigorous convergence guarantees in simplified settings~\citep{Xu2018,BalWaiYu17,Dwivedi2018,wu2021randomly}. Our local minima analysis is heavily inspired by the recent discussion of the structure of spurious local optima for the log-likelihood \citep{jin2016local,chen2020likelihood}. 


The use of optimal transport-based losses for parameter estimation has been recently advocated in \cite{Bernton2019}. Closer to our work is \cite{Rigollet2018}, who show that under some conditions, maximizing the log-likelihood on a deconvolution model (e.g., a mixture) yields \textit{the same} solution as solving an entropic optimal transport problem. This observation motivates the question of which scenarios we may expect to benefit from solving the entropic optimal transport problem rather than maximizing the likelihood, and our methodology provides a concrete example of such a scenario.


Several papers have advocated the use of optimal transport for clustering \citep{cuturi2014fast,kolouri2018sliced,Genevay2019,nejatbakhsh2020probabilistic}, demonstrating its benefits over likelihood-based approaches. In particular, \cite{nejatbakhsh2020probabilistic} first introduced Sinkhorn-EM, which we study in detail here. Some of these papers have suggested that the benefits of optimal transport stem from avoiding spurious local optima \citep{kolouri2018sliced}, and empirical results in related work \citep{yan2023learning} support this perspective, yet the question has not been explored in mathematical detail. Our work provides a coherent, unifying perspective on these phenomena, complementing and substantially extending the preliminary results of \cite{mena2020sinkhorn}, in which the Sinkhorn-EM algorithm was originally introduced.


\subsection{Preliminaries}
Suppose that an i.i.d. sequence $Y_1,\ldots, Y_n\in\mathbb{R}^d$ is generated from an arbitrary $Q_\ast$. To model these data we consider a parametric family on $\RRd$ given by a mixture of $K$ components defined as follows: for each component we specify a location parameter $\theta_k\in\mathbb{R}^d$ and additional parameters $\nu_k$ (such as a variance or scale), so that the density of each component writes as $q_{\theta_k,\nu_k}(y)=q^{\nu_k}(y-\theta_k)$, and $q^\nu(\cdot)$ is a template density. Then, for a vector of weights $\alpha$ in the probability simplex, we complete the specification of our model using the distribution $Q_\theta$ with density
\begin{equation}\label{eq:model} dQ_\theta(y) = \sum_{k=1}^K \alpha_k q_{\theta_k}(y)dy,\end{equation}
where $\alpha_k\in (0,1)$ are mixture weights assumed to add up to one. With slight abuse of notation, we suppressed $\nu$ above and encapsulated all parameters into a single $\theta_k$. We will denote $\Theta$ as the parameter space, and $\alpha$ may or may not be treated as a parameter. To indicate that $\alpha$ is fixed, we will write this set as $\Theta_\alpha$. Note that if $P_\theta$ is the measure concentrated on the location parameters $\theta_k$'s,\begin{equation}\label{eq:ptheta}P_\theta=\sum_{k=1}^K \alpha_k\delta(\theta_k).\end{equation}
Then, $Q_\theta$ is obtained by marginalizing over $x$ the joint distribution $Q^{X,Y}_{\theta}$ with density
\begin{equation}\label{eq:joint}
dQ^{X,Y}_{\theta}(x,y)= q_x(y)dP_\theta(x) \dd y=q(y-x)dP_\theta(x) \dd y. \end{equation}
In the Appendix \ref{sec:preliminaries} we state technical assumptions on the the densities $q_\theta$ and $Q_\ast$ and the parameter space. We emphasize that although some our specialized theoretical results focus on the Gaussian Mixture Model (GMM), i.e., $q_{\theta_k,\nu_k}(y)=\mathcal{N}(y;\theta_k,\Sigma_k)$ (where $\nu_k=\Sigma_k$ can be treated as fixed or as a parameter itself), our methodology is general enough to accommodate the usual multi-dimensional mixtures of location-scale families, Student's $t$ distributions, etc, as long as we impose lower bound on scale parameters to avoid degeneracies. 
To make explicit the dependence on the weights, we will write $\ell(\theta,\alpha)$. In some cases it will be convenient to see $P_\theta$ as a measure $p_\theta$ on the set $[K]:=\{1,\ldots,K\}$ such that $p_\theta(k) = P(X=\theta_k) =\alpha_k$. Likewise, we may see $Q^{X,Y}_{\theta}$ as a measure on $[K]\times\mathbb{R}^d$.

Given the above parametric model and data-generating mechanism, we define the population negative log-likelihood as
\begin{equation} \label{eq:ell} \ell(\theta)=-\mathbb{E}\left(\log \frac{d}{dy} Q_\theta(Y)\right). \end{equation}

In the next section, we will define a surrogate $\Ll(\theta)$ for $\ell(\theta)$ based on entropic optimal transport. Finally, we note that we can consider empirical versions of $\Ll(\theta)$ and $\ell(\theta)$ by taking expectations with respect to the empirical measure $Q_N$ rather than $Q_\ast$. While most of the elementary results grounding our methodology hold in both the population and sample cases (e.g., most of Proposition \ref{prop:L} and Theorem \ref{teo:em}), some others (e.g, Theorems \ref{theo:local}, \ref{teo:mixtureloss}, and \ref{teo:mixture}) rely heavily on the analysis of the \textit{population} quantity.

\section{The entropic optimal transport loss $\Ll(\theta)$}

\label{sec:framework}
To define a new loss function $\Ll(\theta)$ for estimation in the model \eqref{eq:model} based on entropic Optimal Transport, we must first provide some elementary definitions. We refer the reader to, e.g., \cite{Peyre2019,chewi2024statistical} for comprehensive treatments of the topic.

Let $P$ and $Q$ be two probability measures on $\RRd$. Given a cost function $c: \RRd \times \RRd\rightarrow \RR$, we define the entropy-regularized optimal transport loss between $P$ and $Q$ as
\begin{equation}\label{eq:sinkdef}
S(P,Q) := \inf_{\pi\in \Pi(P,Q)} \Big[\int_{\RRd \times \RRd}  c(x,y)\,\dd\pi(x,y) +  H(\pi\lvert P\otimes Q)\Big].
 \end{equation}
where $\Pi(P,Q)$ is the set of all joint distributions with marginals equal to $P$ and $Q$, respectively, 
and $H(\alpha\lvert \beta)$ denotes the relative entropy between measures $\alpha$ and $\beta$ defined as $\int \log \frac{\dd\alpha}{\dd\beta}(x)\dd\alpha(x)$ if $\alpha \ll \beta$ and $+ \infty$ otherwise. A useful alternative representation of $S(P,Q)$ is in terms of a reference Gibbs kernel with density proportional to $\exp(-c(x,y))\dd P(x)\dd y$, that holds whenever $Q$ has a density with respect to $\mathcal{L}$ (see \cite{MenaNilesWeed2019} for details):
\begin{equation}\label{eq:sinkdef2}
S(P,Q) = \inf_{\pi\in \Pi(P,Q)}  H(\pi\lvert e^{-c} P\otimes \mathcal{L})-H(Q\lvert \mathcal{L}).
 \end{equation}
We define the entropic OT loss function for parameter estimation in model \eqref{eq:model} as $\Ll(\theta):=S(P_\theta,Q_{\ast})$. To make the correspondence precise with model \eqref{eq:model} we must restrict to cost functions of the form $c(x,y)=-\log q_x(y)=-\log q(y-x)$ so that we recover the joint densities $Q^{X,Y}_\theta$ in \eqref{eq:joint} as the set of Gibbs kernels in \eqref{eq:sinkdef2}, i.e.,
\begin{equation}\label{eq:L}
    \Ll(\theta):=S(P_\theta, Q_{\ast})=\inf_{\pi\in \Pi(P_\theta,Q_{\ast})}  H\left(\pi\lvert Q^{X,Y}_\theta\right)-H( Q_\ast\lvert \mathcal{L}).
\end{equation}
 For example, in the simplest Gaussian case we let $c(x,y)=\frac{||x-y||^2}{2\sigma^2}$, and for a mixture of Laplace distributions we make $c(x,y)=\frac{|x-y|}{b}$ (up to additive constants). As the last term in \eqref{eq:L} doesn't depend on $\theta$, $\Ll(\theta)$ is indeed a discrepancy between the model \eqref{eq:joint} and data. Under this definition, the dependence on the location parameters is fully encapsulated in the base measure $P_\theta$, and by varying the cost function, we can also represent dependencies on other parameters, such as variances and scales.

Note that $\Ll$ is intimately related to $\ell$. First, in the well-specified case, we have that $\Ll(\theta^\ast)=\ell(\theta^\ast)$ at the true $\theta^\ast$. To see this, we note that if $\theta=\theta^\ast$ the coupling $\pi$ achieving the infimum is exactly the joint $Q^{X,Y}_{\theta^\ast}$ in \eqref{eq:joint}, so the relative entropy term is zero. Second, we have that $\Ll(\theta)\geq \ell (\theta)$. This follows from another useful alternative representation of the entropic OT loss, the so-called semi-dual formulation \cite{Peyre2019}:
\begin{equation}\label{eq:semidual0}S(P,Q)=\sup_{\omega\in L_1(P)} \int \omega(x)dP(x) -\int \log \left( \int e^{\omega(x)-c(x,y)}dP(x)\right)dQ(y),\end{equation}
where the semi-dual potential $\omega$ is a bounded function. In our context $\omega$ is a $K$-dimensional vector so
\begin{equation}\label{eq:semidual} \Ll(\theta)=\sup_{\omega\in \mathbb{R}^K}\Bigg[ \sum_{k=1}^K \alpha_k\omega_k -\mathbb{E} \left(\log \left( \sum_{k=1}^K \alpha_k e^{\omega_k}q_{\theta_k}(y)\right)\right)\Bigg], \end{equation}
and since optimization over $\omega$ includes $\omega=0$ it follows that $\Ll(\theta)\geq \ell(\theta)$. 

The semi-dual potentials can be interpreted as a way of \textit{tilting} the original weights $\alpha$: for each $\theta$ denote $\omega(\theta)$ the one maximizing \eqref{eq:semidual}  for each $\theta$ denote $\omega(\theta)$ the one maximizing \eqref{eq:semidual} (this $\omega(\theta)$ is uniquely defined once fixing one coordinate, as shown in Lemma \ref{lemma:alphaprop} in the Appendix) and define the vector $\alpha(\theta)\in\mathbb{R}^K$ as
\begin{equation}\label{eq:alphatheta} \alpha_k(\theta)=\frac{\alpha_{k}e^{\omega_k(\theta)}}{\sum_{k'=1}^K \alpha_{k'}e^{\omega_{k'}(\theta)}}.\end{equation}
Then, as we show in Appendix \ref{sub:semidual},
\begin{equation}\label{eq:semidual2} \Ll(\theta)= H\left(\alpha \lvert \alpha(\theta)\right)  -\mathbb{E} \left(\log \left( \sum_{k=1}^K \alpha_k(\theta) q_{\theta_k}(Y)\right)\right). \end{equation}
Therefore, the computation of $\Ll(\theta)$ amounts to the computation of $\ell(\theta)$ for a \textit{tilted} model with weights $\alpha(\theta)$, plus a relative entropy term. In particular, since this relative entropy doesn't depend on $\theta$, we have that 
\begin{equation}\label{eq:gradientL}\nabla_\theta \Ll(\theta)=\nabla_\theta \ell(\theta,\alpha(\theta)).
\end{equation}
While the analysis of the second derivatives of $\Ll$ is more complicated, it is possible to show that $\Ll$ has strictly more curvature. We summarize all this discussion in the following proposition. A complete proof appears in Appendix \ref{app:propL}.

\begin{proposition}\label{prop:L}
Let $\ell$ and $\Ll$ be as in \eqref{eq:ell} and \eqref{eq:L}, respectively. The following statements are true
\begin{itemize}
\item[(a)] $\Ll(\theta)$ and $\alpha(\theta)$ are twice continuously differentiable for any $\theta \in int(\Theta)$.
\item[(b)] $\Ll(\theta)\geq \ell(\theta)$ for any $\theta\in \Theta$
\item[(c)] $\Ll(\theta^\ast) = \ell(\theta^\ast)$ if $Q_\ast=Q_{\theta^\ast}$ and $\theta^\ast \in int(\Theta)$
\item[(d)] For any $\theta \in int(\Theta_\alpha)$, $\nabla_\theta \Ll(\theta) = \nabla_\theta \ell(\theta,\alpha(\theta))$,
\item[(e)] $\Ll$ has more curvature than $\ell$ around $\theta^*$ if $Q_\ast=Q_{\theta^\ast}$. Specifically, if $\theta^\ast \in int(\Theta_\alpha)$, $$\nabla^2_{\theta\theta} \Ll(\theta^*) = \nabla^2_{\theta\theta} \ell(\theta^*) +B^\top(\theta^\ast)A^{-1}(\theta^\ast)B(\theta^\ast),$$
where $A(\theta^\ast)$ is a $K-1\times K-1$ definite positive matrix. Explicit expressions for $A(\theta^\ast)$, $B(\theta^\ast)$ are given in Appendix \ref{app:propL}.
\item[(f)] $\Ll$ is a coercive function, i.e., $\Ll(\theta)\to \infty$ whenever $\theta\to\infty$.
\end{itemize}
(a), (b) and (d) and (f) are valid both in the population and sample case, while (c) and (d) are valid only in the population versions, in the well-specified setup ($Q_\ast=Q_{\theta^\ast}$).
\end{proposition}
Fig. \ref{fig:global}A illustrates some aspects of Proposition \ref{prop:L}. As a consequence of (b) and (c), in the population limit, $\Ll$ is also minimized at $\theta^\ast$, and the curvature statement (e) and coercivity statement (f) suggest that $\Ll$ might be a better optimization objective than $\ell$. For example, $\ell$ fails to be coercive even in the case of a mixture of two Gaussians with fixed variances and different location parameters $\theta_1,\theta_2$. Still, since both $\Ll$ and $\ell$ are typically non-convex, many local optima may exist for both functions, so not much can be deduced from Proposition \ref{prop:L} alone.  In Section \ref{sec:local}, we show that for a mixture of Gaussians, $\Ll$ will typically avoid spurious stationary point configurations that are pervasive for $\ell$, making a much stronger case in favor of $\Ll$ as an optimization objective. Before that, in the following section, we describe an algorithm for optimizing $\Ll$ and make a parallel with the usual EM algorithm.

\begin{figure*}
\includegraphics[width=1.0\textwidth]{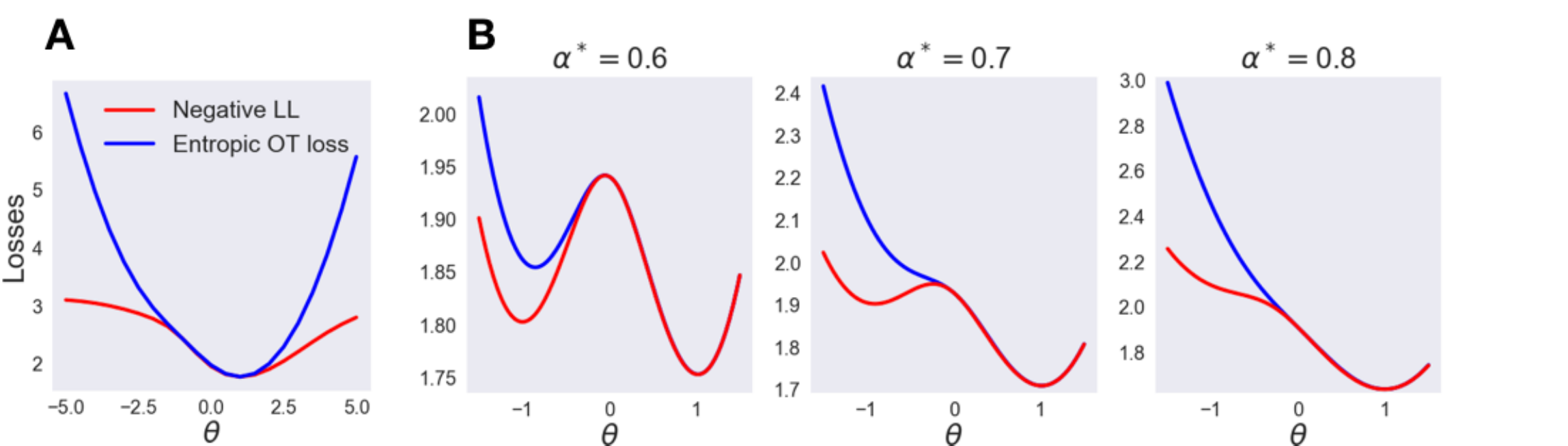}
\caption{Qualitative comparison between the log-likelihood and entropic OT loss. \textbf{A} By Proposition \ref{prop:L}, the entropic OT loss dominates the negative log-likelihood but has a larger curvature around the minimum. \textbf{B} By Theorem \ref{teo:mixture}, in the model \eqref{eq:mixtureg} $\Ll$ has fewer spurious local minima than $\ell$ for some values of $\alpha^\ast$.}
\label{fig:global}
\end{figure*}

\vspace{-1cm}
\section{Algorithmic aspects: Sinkhorn-EM}\label{sec:em}

In principle, we could consider any first-order method to find stationary points of $\Ll$ and $\ell$. In the case of $\ell$ practitioners typically appeal to the EM algorithm (an instance of a first-order method \cite{Xu199}). This algorithm exploits the underlying latent structure of the mixture model \eqref{eq:model}, and its appeal is ultimately justified by a rich body of work establishing its theoretical guarantees (e.g., \cite{Redner1984,BalWaiYu17}). Pivoting on the relationship between $\ell$ and $\Ll$ from the previous section, we describe an EM-type algorithm for optimizing $\Ll$, which we name Sinkhorn-EM, and establish some elementary convergence guarantees. \\
The main observation is that, as a consequence of the well-known variational representation of the log-likelihood $\mathbb{E}\left(\log \frac{d}{dy}Q_\theta(Y)\right)$, we can write: (see Lemma 1 in \cite{Neal1998} for a similar formulation)
\begin{equation}\label{eq:ellvariational}
    \ell(\theta)=\inf_{\pi\in \Pi(\cdot,Q_{\ast})}  H\left(\pi\lvert Q^{X,Y}_\theta\right)-H\left(Q_\ast\lvert \mathcal{L}\right),
\end{equation}
where the set $\Pi(\cdot,Q_{\ast})$ is the set of joint distributions with arbitrary first marginal and second marginal $Q_{\ast}$. By disintegration, this set can be represented by $Q_{\ast}$ along with a set of conditionals $\pi(\cdot|y)$. 
The EM algorithm exploits this variational representation, as it can be understood as coordinate descent on $\theta$ and $\pi$ to minimize $\ell$ in \eqref{eq:ellvariational} \citep{Neal1998,CsiszarT1984}. 

By comparing \eqref{eq:L} and \eqref{eq:ellvariational}, we see that $\Ll$ and $\ell$ only differ in that the variational representation of the former has an additional constraint. This observation motivates the definition of Sinkhorn-EM (SEM) as the algorithm performing coordinate ascent on $\theta$ and $\pi$ to minimize the variational representation \eqref{eq:L}.

Sinkhorn-EM only differs from EM in the E-step. For the EM algorithm, this step corresponds to the computation of the so-called \textit{responsibilities}, the set of conditionals of $X$ given $y$ at the current $\theta$ in the joint model \eqref{eq:joint}, i.e., the set of measures 
\begin{equation}\label{eq:resp}  \Psi_k\left(y,\theta,\alpha\right):= \dd Q^{X,Y}_{\theta,\alpha}(k|y)=\frac{\alpha_k q_{\theta_k}(y)}{\sum_{k=1}^K\alpha_k q_{\theta_k}(y)}.\end{equation}
Instead, Sinkhorn-EM solves an entropic Optimal Transport problem on the E-step. If $\pi_\theta$ is the optimal coupling to \eqref{eq:L} we have
\begin{equation}\label{eq:respsinkhorn}  \Psi_k\left(y,\theta,\alpha(\theta)\right)=\pi_\theta(k|y)= \dd Q^{X,Y}_{\theta,\alpha(\theta)}(k|y)=\frac{\alpha_k(\theta) q_{\theta_k}(y)}{\sum_{k=1}^K\alpha_k(\theta) q_{\theta_k}(y)}.\end{equation}

where $\alpha(\theta)$ is the \textit{tilted} version of $\alpha$ defined in \eqref{eq:alphatheta} so that the constraint on the first marginal expresses as
\begin{equation}\label{eq:respsinkhorn}\mathbb{E}\left(\Psi_k\left(Y,\theta,\alpha(\theta)\right)\right)=\alpha_k.\end{equation} 

Therefore, the familiar understanding of the E-step as the computation of \textit{responsibilities} is retained, although in a somewhat different sense. This signifies an extra computational burden compared to \eqref{eq:resp}: while each E-step of the standard EM algorithm takes time $O(n \cdot K)$, where $n$  is the sample size, the E-step of Sinkhorn-EM involves a convex optimization problem and can be solved by an efficient, celebrated algorithm due to~\citet{Sinkhorn1967} which converges in near-linear time, i.e. $\tilde O(n \cdot K)$ \citep{Altschuler2017}. However, this complexity bound hides unfavorable dependency on parameters such as the infinity norm of the underlying cost matrix and the error tolerance. In practice, this implies substantially slower execution times for each E-step of SEM. However, as we will see in Sections \ref{sec:local}, \ref{sec:experiments}, \ref{sec:celegans}, and \ref{sec:coclustering}, this overhead will be compensated by the fact that SEM will, in practice, reach better solutions compared to EM.

We can now state the Sinkhorn-EM algorithm and establish an elementary convergence property. A proof appears in Appendix \ref{app:em}.
\begin{theorem}\label{teo:em}
Suppose that weights are fixed, so that $\Theta=\Theta_\alpha$. Let Sinkhorn-EM be the algorithm defined as the sequence $(\pi^t,\theta^t)$ from the $E$ and $M$ steps below, with $\theta^0\in int(\Theta_\alpha)$:
\begin{subequations}
\begin{align}
\label{eq:estep}&\textbf{E-step}:\quad & \pi^{t+1}& =\pi_{\theta^t}= \argmin_{\pi\in \Pi\left(P_{\theta^{t}},Q_{\ast}\right)}  H\left(\pi\lvert Q^{X,Y}_{\theta^{t}}\right)\\
\label{eq:mstep} &\textbf{M-step}:\quad   &  \theta^{t+1} & = \argmax_{\theta\in\Theta_\alpha} \mathbb{E}\left(\sum_{k=1}^K\pi^{t+1}(k|Y) \log \left(\alpha_k q_{\theta_k}(Y)\right)\right)\\
&  & & \nonumber =\argmax_{\theta\in\Theta_\alpha} \mathbb{E}\left(\sum_{k=1}^K\Psi_{k}(Y,\theta^t,\alpha(\theta^t)) \log \left(\alpha_k q_{\theta_k}(Y)\right)\right),
\end{align}
\end{subequations}
where $\theta^0$ is arbitrary. Then, the sequence $\{\Ll(\theta^{t})\}$ is non-increasing and $\Ll(\theta^{t+1})<\Ll(\theta^t)$ if $\theta^t$ is not a stationary point of $\Ll$. Moreover, $\theta^t$ converges to a stationary point of $\Ll$.
\end{theorem}
 Theorem \ref{teo:em} extends the classical convergence result for EM \citep{Wu1983}, which assumes a stringent coercivity condition on $\ell$. Thanks to our coercivity statement Proposition \ref{prop:L}(f), we are able to derive this result without assuming this property on $\ell$.


\subsection{Weights update}\label{sec:alpha}

Sinkhorn-EM can, in principle, only be deployed in fixed-weights setups: as the E-step imposes a marginal constraint on weights $\alpha$, it is bound to remain on those initial weights. However, nothing prevents us from considering $\Ll$ as a function of $\alpha$ as well, as it is customary for mixture models with unknown weights. To enable weight inference, we consider simple exponentiated gradient  \cite{kivinen1997exponentiated} updates for $\alpha$. The gradient of $\Ll$ with respect to $\alpha$ is given by ($1_K$ is $K$-dimensional all-ones vector)
\begin{equation}
\label{eq:gradLalpha}\nabla_{\alpha} \Ll(\theta,\alpha)=\omega(\theta,\alpha)-1_K.
\end{equation}
Then, for a step-size $\eta>0$, current $\alpha$ and fixed $\theta$, an update for $\alpha$ reads:
\begin{equation}\label{eq:updatealpha}
\alpha^{new}_k = \frac{\alpha_k\exp\left(-\eta \nabla_\alpha \Ll(\theta,\alpha)_k\right)}{\sum_{k=1}^K \alpha_k\exp\left(-\eta \nabla_\alpha \Ll(\theta,\alpha)_k\right)}=\frac{\alpha_k\exp\left(-\eta \omega_k(\theta,\alpha)\right)}{\sum_{k=1}^K \alpha_k\exp\left(-\eta  \omega_k(\theta,\alpha)\right)}.
\end{equation}
Below we show establish the convexity of $\Ll(\theta,\cdot)$ and, under a technical spectral condition, the convergence of i) the iterates \eqref{eq:updatealpha} to the minimizer of $\Ll(\theta,\cdot)$ and ii) convergence of a coordinate descent algorithm over $\theta$ and $\alpha$ to a stationary point of $\Ll$.
\begin{proposition}\label{prop:convalpha}
For each $\theta\in int(\Theta_\alpha)$, $\Ll(\theta,\cdot)$ is convex. Also, if 
 \begin{equation} \label{eq:condkappa} \sup_{\alpha \in(0,1)^K,\sum_{k=1}^K \alpha_k=1} \lambda_2\left(\tilde{D}^{-1/2}(\alpha)\mathbb{E}\left(\Psi(Y,\theta,\alpha(\theta))\Psi^\top (Y,\theta,\alpha(\theta))\right)\tilde{D}^{-1/2}(\alpha)\right) <1, \end{equation}
 where $\tilde{D}(\alpha)$ is the diagonal matrix with entries $\alpha_k$ and $\lambda_2(A)$ denotes the second largest eigenvalue of $A$, 
 then the updates \eqref{eq:updatealpha}  converge to the minimizer of $\Ll(\theta,\alpha)$ if $\eta$ is sufficiently small. Additionally, the following block-coordinate descent algorithm converges to a stationary point of $\Ll$: from $(\theta^0,\alpha^0)\in int(\Theta)$ iteratively, i) finds a stationary point $\theta^{t+1}$ of $\Ll(\cdot,\alpha^t)$ by Sinkhorn-EM ii) finds $\alpha^{t+1}$ by minimizing $\Ll(\theta^{t+1},\alpha)$ using updates. \eqref{eq:updatealpha}
\end{proposition}


\subsection{The Gaussian Mixture Model}
Our more specialized theoretical results in Sections \ref{sec:local} and \ref{sec:twogaussians}, as well as our empirical results in Sections \ref{sec:experiments} and  \ref{sec:celegans} focus on the following GMM:
\begin{equation}\label{eq:gmm} Q_\theta=\sum_{k=1}^K \alpha_k\mathcal{N}(\theta_k,I_d\sigma_k^2)\quad \alpha\in(0,1)^K, \sum_{k=1}^K\alpha_k=1,\sigma_k^2\geq \sigma_0^2>0. \end{equation}
 In that case, the M-step \eqref{eq:mstep} reads
$\theta^{t+1}= F(\theta^t,\alpha(\theta^t))$,
with
\begin{equation}
F(\theta,\alpha)_k:= \frac{\mathbb{E}\left(Y\left(\Psi_k (Y,\theta,\alpha \right)\right)}{\mathbb{E}\left(\Psi_k \left(Y,\theta,\alpha\right)\right)}=\frac{\mathbb{E}\left(Y\left(\Psi_k (Y,\theta,\alpha \right)\right)}{\alpha_k}.
\end{equation}
Note that for the usual EM algorithm, we would otherwise consider the simpler $\theta^{t+1}=F(\theta,\alpha)$ iterations.
\vspace{-1cm}
\section{On stationary points}\label{sec:local}
Theorem \ref{teo:em} in the previous section provides us with machinery to efficiently optimize $\Ll$, a function that in Proposition \ref{prop:L} was shown to have better curvature than $\ell$ around the true $\theta^\ast$. However, since both EM and SEM can only be shown to converge to stationary points of the respective $\ell$ and $\Ll$, and since these stationary points may abound, this local property says little about the success or failure of these algorithms in practice.
Throughout this section, we assume that data was generated by the distribution \eqref{eq:gmm} with equal variances $\sigma^2$, weights $\alpha_k=1/K$, and arbitrary distinct means $\theta^\ast_k$, that we fit to the same family $Q_\theta$ in \eqref{eq:gmm} with varying locations $\theta_k$.
In Theorem \ref{theo:local} below, we make a much stronger case favoring $\Ll$ as optimization objective: when $Q_\theta$ is a spherical mixture of Gaussians with equal weights and variances $\sigma^2$, then $\Ll$ won't possess a type of pathological stationary point, the so-called \textit{many-fit-one} configurations that we show to be pervasive for the log-likelihood. We will rely on the notion of $\delta$-covering of a finite set by another:
\begin{definition}Let $\mathcal{A},\mathcal{B}$ be two subsets of $\mathbb{R}^d$ with $0<|\mathcal{B}|\leq |\mathcal{A}|<\infty$. For $\delta>0$, we say that $\mathcal{A}$ is a $\delta$-covering of $\mathcal{B}$ if we can partition $\mathcal{B}$ into $|\mathcal{A}|$ disjoint nonempty subsets $\{\mathcal{B}_x\}_{x\in\mathcal{A}}$, such that for any $x\in \mathcal{A}$ and $y\in\mathcal{B}_x$, $\lVert x-y\rVert \leq \delta$.
\end{definition}
Unlike the usual covering notions, the above definition requires that both sets be exhausted by one covering the other. In our setup, the sets will be formed by selecting arbitrary components from  $\theta = (\theta_1, \ldots, \theta_K)$. Specifically, for a set of indices $G\subseteq [K]$ and $\theta\in\mathbb{R}^{d\times K}$ we define $G(\theta):=\{\theta_k,k\in G\}$.
To state our result, we assume $d=1$ so we can arrange the values $\theta_k$ in increasing order and denote them $\theta_{(k)}$. In the general case $d>1$, we can apply Theorem \ref{theo:local} to each coordinate, as shown in the proof (Appendix \ref{sec:suplocal}).

\begin{theorem}\label{theo:local}
and for arbitrary $1\leq k_1\leq K-1$  define $G_1$ as the indices of the $k_1$ smallest $\theta^\ast_k$'s. Let $\Delta$ be the distance between the two groups, i.e.
$$\Delta_{\theta^\ast,k_1} := \inf_{g_1\in G_1,g_2\in G_1^c}|\theta^\ast_{g_1}-\theta^\ast_{g_2}|=\theta^\ast_{(k_1+1)}-\theta^\ast_{(k_1)}>0.$$

For an arbitrary $\theta\in \Theta$ and set of indices $I\subseteq [K]$ with size $|I|>k_1$, $I(\theta)$ is a $\delta$-covering of $G_1(\theta^\ast)$. 
 Then, for any $\gamma$ with $0<\gamma<1-k_1/|I|$, if the following separation condition holds
 \begin{equation}\label{eq:snr}\Delta_{\theta^\ast,k_1}\geq \frac{K\sigma}{\sqrt{2\pi}\left((1-\gamma)|I|-k_1\right)},\end{equation}
then $\theta$ cannot be a stationary point of $\Ll$ if $\delta<\gamma \Delta_{\theta^\ast,k_1}$.
\end{theorem}

Let us parse the main implication of Theorem \ref{theo:local}: as long as we can identify two groups of sufficiently separated true mixture components, \textit{many-fit-one} stationary points will be excluded for $\Ll$. These configurations are those in which we can cover $k_1$ true components with a strictly larger number of fitted components, $I(\theta)$. Intuitively, our result follows from the fact that minimizing $\Ll$ imposes an additional constraint on how mass is split for stationary points, that $\sum_{k=1}^K\theta_k=\sum_{k=1}^K\theta^\ast_k$. This balance condition, which doesn't hold for stationary points of the log-likelihood, cannot occur if many more fitted components are placed near a lower number of true components $\theta^\ast$.

This result is complemented by Proposition \ref{prop:bad}, showing an explicit construction of a \textit{many-fit-one} configuration that turns out to be a spurious local optimum for $\ell$.

\begin{proposition}\label{prop:bad}
Consider a mixture of three Gaussians in $\mathbb{R}^2$ with equal weights and variances $\sigma^2=1$ and true locations $\theta^\ast_1=(0,D)$ ,$ \theta^\ast_2=(0,-D)$ and $\theta^\ast_3=(R,0)$. For $\varepsilon>0$ define the region $\mathcal{R}^\varepsilon$ in $\mathbb{R}^6$ 
$$\mathcal{R}^\varepsilon:=\biggl\{\theta=(\theta_1^1,\theta_1^2,\theta_2^1,\theta_2^2,\theta_3^1,\theta_3^2): \theta_1^1< \frac{R}{3},\theta_2^1> \frac{2R}{3}, \theta_3^1> \frac{2R}{3},||\theta^2||<\varepsilon\biggl\}.$$
If $R$ is large enough, then $\ell$ has a spurious local minima in $\theta\in\mathcal{R}^\varepsilon$. 
\end{proposition}
Informally, the spurious optimum $\theta$ lies on a neighborhood of the configuration $\theta_1\approx\frac{1}{2}\left(\theta^\ast_1+\theta^\ast_2\right)=(0,0)$ and both $\theta_2\approx \theta_3\approx \theta^\ast_3=(R,0)$. This construction extends the one given \cite{jin2016local} on a one-dimensional setup, and illustrates the potentially catastrophic effect of spurious local minima for the negative log-likelihood: regardless of how strong a signal can be on the $y$ axis (as indicated by $D$), there is a local optimum that essentially destroys this information, as it averages out $\theta_1^\ast$ and $\theta_2^\ast$. Consequently, only the signal in the $x$ axis is recovered. (see Fig \ref{fig:bad}A).  Contrarily, Theorem \ref{theo:local} allows us to rule out stationary points for $\Ll$ over $\mathcal{R}^\varepsilon$ if $R$ and $D$ are large enough, as shown in Fig. \ref{fig:bad}B. This is stated as follows (a proof appears in Appendix \ref{app:statio})
\begin{corollary}\label{cor:statio}
If $R\geq 3D$ and $D\geq 2$ then $\Ll$ cannot have a stationary point in  $\mathcal{R}^\varepsilon$.
\end{corollary} 

\begin{figure}[!t]
  \begin{center}
    \includegraphics[width=0.9\textwidth]{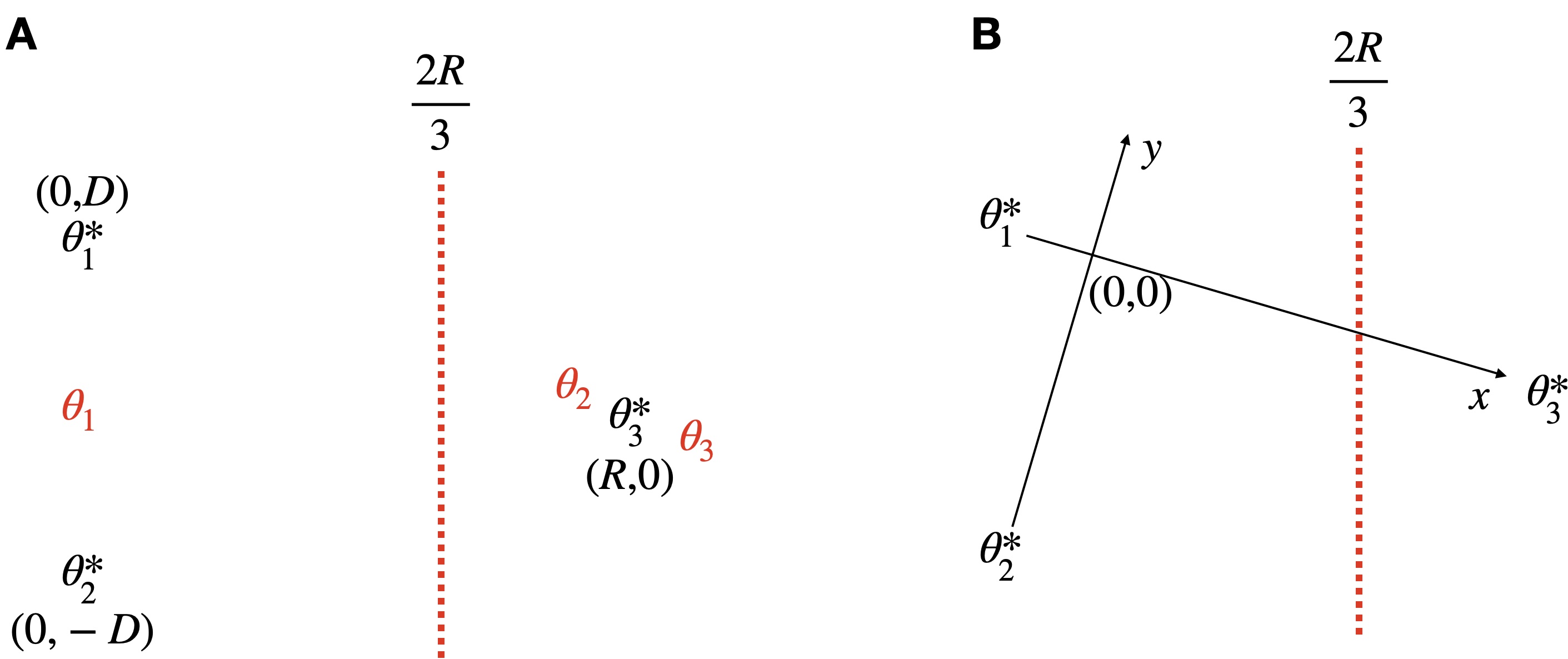}
  \end{center}
  \caption{Setup of Proposition \ref{prop:bad}. \textbf{A}. By Proposition \ref{prop:bad}, the negative log-likelihood has a spurious local minimum around the configuration $\theta$ (red), i.e., within the region $\mathcal{R}^\varepsilon$. \textbf{B} by rotating axes, we can apply Theorem \ref{theo:local} to rule out any stationary point for $\Ll$ in such region if $R$ is large enough.}
  \label{fig:bad}
\end{figure}

The relevance of the \textit{many-fit-one} configurations is that, as described in \cite{chen2020likelihood}, all spurious local maxima of the log-likelihood are characterized in terms of such and related structures as long as an absolute separation condition over components holds. In sections \ref{sec:experiments}, \ref{sec:celegans} and \ref{sec:coclustering}, we present extensive experiments on simulated and real data and argue that, in practice, EM's primary mode of failure is attributable to the pervasiveness of these configurations. We also demonstrate that SEM often avoids those, thereby explaining its superior performance.  In Appendix \ref{sec:suplocal} we provide additional discussion and a result suggesting that we may not expect completely different stationary point structures for $\Ll$.


\section{Analysis of a symmetric mixture of two Gaussians}
\label{sec:twogaussians}
In this section, we present specialized results for an unbalanced mixture of two Gaussians in $\mathbb{R}$ with a single one-dimensional parameter $\theta$. That is, denote $q_{\theta,\alpha}(y)$ the density
\begin{equation}
\label{eq:mixtureg} q_{\theta,\alpha}(y) = \alpha \mathcal{N}(y;\theta, 1) +(1-\alpha)\mathcal{N}(y;-\theta,1).
\end{equation}
We will assume that the data was generated by the above density with fixed $\theta^\ast$ (unknown) and weight $\alpha^\ast$ (known). We will use the family $q_{\theta,\alpha^\ast}$ to infer $\theta$. The first result, Theorem \ref{teo:mixtureloss}, complements our results on local minima from Section \ref{sec:local}. As illustrated in Fig \ref{fig:global}B, we show that compared to $\ell$, $\Ll$ will more often have a unique global minimum. 
\begin{theorem}\label{teo:mixtureloss}
Consider the model \eqref{eq:mixtureg}. For any $\theta^\ast>0$, the set of $\alpha^\ast$ for which $\theta^\ast$ is the unique stationary point of $\Ll$ is strictly larger than the one for $\ell$.
\end{theorem}
We also provide an in-depth analysis of the basin of attraction of Sinkhorn-EM towards the global minimum and the convergence speed: we show that as long as SEM is initialized at $\theta^0 > 0$, it enjoys fast (exponential) convergence to the global minimum of $\Ll$, and there is a large range of initializations on which it never performs worse than EM. 

\begin{theorem}\label{teo:mixture}
For the model \eqref{eq:mixtureg}, for each $\theta^*>0$ and initialization $\theta^0 > 0$, the iterates of SEM converge to $\theta^*$ exponentially fast:
\begin{equation}
\label{emgeom}
|\theta^{t}_{SEM} - \theta^*| \leq \rho^t |\theta^0 - \theta^*|\text{, with } \rho = \exp\left(-\frac{\min\{\theta^0, \theta^*\}^2}{2}\right).
\end{equation}

Moreover, there is a $\theta_{\text{fast}} \in(0,\theta^*)$ depending only on $\theta^*$ and $\alpha^*$ such that if SEM and  EM are both initialized at $\theta^0 \in [\theta_{\text{fast}}, \infty)$, then
\begin{equation}
\label{eq:fast}
|\theta^t_{SEM} - \theta^*| \leq |\theta^t_{\text{EM}} - \theta^*| \quad \quad \forall t \geq 0\,,
\end{equation}
where $\theta^t_{\text{EM}}$ are the iterates of  EM.
In other words, when initialized in this region, SEM never underperforms EM.
\end{theorem}
\vspace{-1cm}
\section{Experiments on simulated data}
\label{sec:experiments}
To empirically study the benefits of optimization of $\Ll$, we conducted a series of experiments on simulated mixtures of Gaussians \eqref{eq:gmm}. We compared the results of $\Ll$ optimization to two competitive baselines: the EM algorithm and k-means. For initialization, we used k-means++ \citep{arthur2007k}, as implemented in scikit-learn \citep{pedregosa2011scikit}. This initialization samples $\theta_1^0$ at random from data, and subsequently, each $\theta^0_k$ is sampled from the empirical distribution of not already chosen data points with probability proportional to $D^2$, where $D$ is the minimum distance between each remaining data-point and the previously selected $\theta^0_1,\ldots,\theta^0_{k-1}$. We used several initial seeds and reported either individual outcomes or the best-performing seed. To select the best initialization, we use the inertia criterion \cite{pedregosa2011scikit} for k-means and the log-likelihood for $\ell$- and $\Ll$-based estimators. We observed that using either $\Ll$ or $\ell$ for seed selection leads to only minor differences and does not affect the conclusions reported below.

We considered two performance metrics: first, as a direct measure of how well the parameters were inferred, we use the squared average distance between true and fitted centers, which we call simply the \textit{error}:
\begin{equation}\label{eq:error}
error(\theta,\theta^\ast) = \min_{\pi\in Perm(K)}\frac{1}{K}\sum_{k=1}^K||\theta_{\pi(k)}-\theta^\ast_k||^2.
\end{equation}
Since centers are defined up to permutations, we must search among all possible permutations $Perm(K)$ of $K$ indexes, a computation that we state as an optimal transport problem \cite{Peyre2019}. Second, we used the Adjusted Rand Index (ARI) \citep{hubert1985comparing} to measure the similarity between the actual and fitted clustering solutions. The ARI is a value between 0 and 1, with higher values indicating greater similarity. 

We study how results vary as functions of various parameters, including $K$, $\sigma^2_k$, $d$, $\alpha_k$, and the sample size $n$. We divide the experiments into three parts, but defer the details of the last two to the Appendix: known weights and known $K$, Unknown weights and known $K$ (Appendix \ref{sub:exp2}), and known weights and unknown $K$ (Appendix \ref{sub:exp3}).


In the first experiment, we studied clustering performance as a function of the number of components $K$,  variance parameter $\sigma^2$, and dimension $d$. We consider the following experimental setups: i)spherical variance with known variance $\sigma^2$, ii) elliptical variance with diagonal entries $\sigma_k$ sampled from a uniform distribution between $0.5\sigma^2$ and $1.5\sigma^2$, iii) and iv) same as i) and ii) but with unknown variances to be estimated. For each parameter configuration, we sample $n_{exp}=200$ datasets and run three methods on each, using $n_{seed}=5$ different random k-means++ initializations. We considered two sample sizes $N=200$ and $N=1000$ and dimensions $d=2,5,10,20$. In all cases, we use equal weights $\alpha_k=1/K$. In cases where $\sigma^2$ must be inferred, we use the identity matrix as an initial estimate. 

Fig. \ref{fig:errorari} shows a summary for experiment i). Both EM and SEM typically improve upon the k-means baseline, although the improvement for SEM is the largest (Fig. \ref{fig:errorari}A, B). As a result, SEM achieves the best performance among the three algorithms, whether measured at the level of individual seeds or at the best-seed level. The interplay between relative performance and $\sigma^2$ and $K$ is complex: although differences are observed at each $K$ if $K$ is small (e.g., $K\leq 10$), these differences vanish once considering the best seed, indicating that EM and k-means++ are still capable of finding the global solution. However, for larger values of $K$, SEM consistently avoids poor local optima, whereas EM and k-means++ struggle, even with multiple seeds. The performance of SEM also deteriorates for large $K$, suggesting that it cannot entirely escape bad local optima. Similarly, if $\sigma^2$ increases, the performance of all algorithms deteriorates, but SEM performs consistently better, in line with Theorem \ref{theo:local}, which imposes a minimum signal-to-noise ratio condition so that the no-\textit{many-fit-one} guarantee holds. Also, if $\sigma^2$ is too small, then all algorithms will recover the true solution. This suggests that there is a mid-range of values for which SEM is most beneficial. In Appendix \ref{app:details} we provide a more detailed account of experimental results by breaking them down into further experimental conditions, and here we give a summary: first, by observing that similar results are obtained for $N=200$ and $N=1000$, suggesting that we can ignore the finite-sample related errors (Figs \ref{fig:200} and \ref{fig:1000}). Second, perhaps unintuitively, larger values of $d$ are associated with lower error for three algorithms. This is because k-means++ initialization consistently yields better estimates in such cases, benefiting all algorithms. Again, there is a range of values of $d$ where Sinkhorn-EM is most beneficial, and this range depends on $\sigma^2$ as well.

Results for experiments ii), iii), and iv) are also presented in Appendix \ref{app:details}. Sinkhorn-EM still attains the lowest regardless of whether the true variance is spherical or diagonal (and different for each cluster) (Fig. \ref{fig:difvar}) and whether the variance is fixed or has to be inferred (Figs. \ref{fig:unk1}, \ref{fig:unk2}). Likewise, results from experiments with unknown weights and known $K$ (Appendix \ref{sub:exp2}) and known weights and unknown $K$ (Appendix \ref{sub:exp3}) indicate that SEM outperforms the baselines when the true weights are sufficiently close to uniform. However, if these weights differ substantially from uniformity (Appendix \ref{sub:exp2}), SEM doesn't appear to recover the centers better than EM or k-means++.

\begin{figure}
  \begin{center}
    \includegraphics[width=1.0\textwidth]{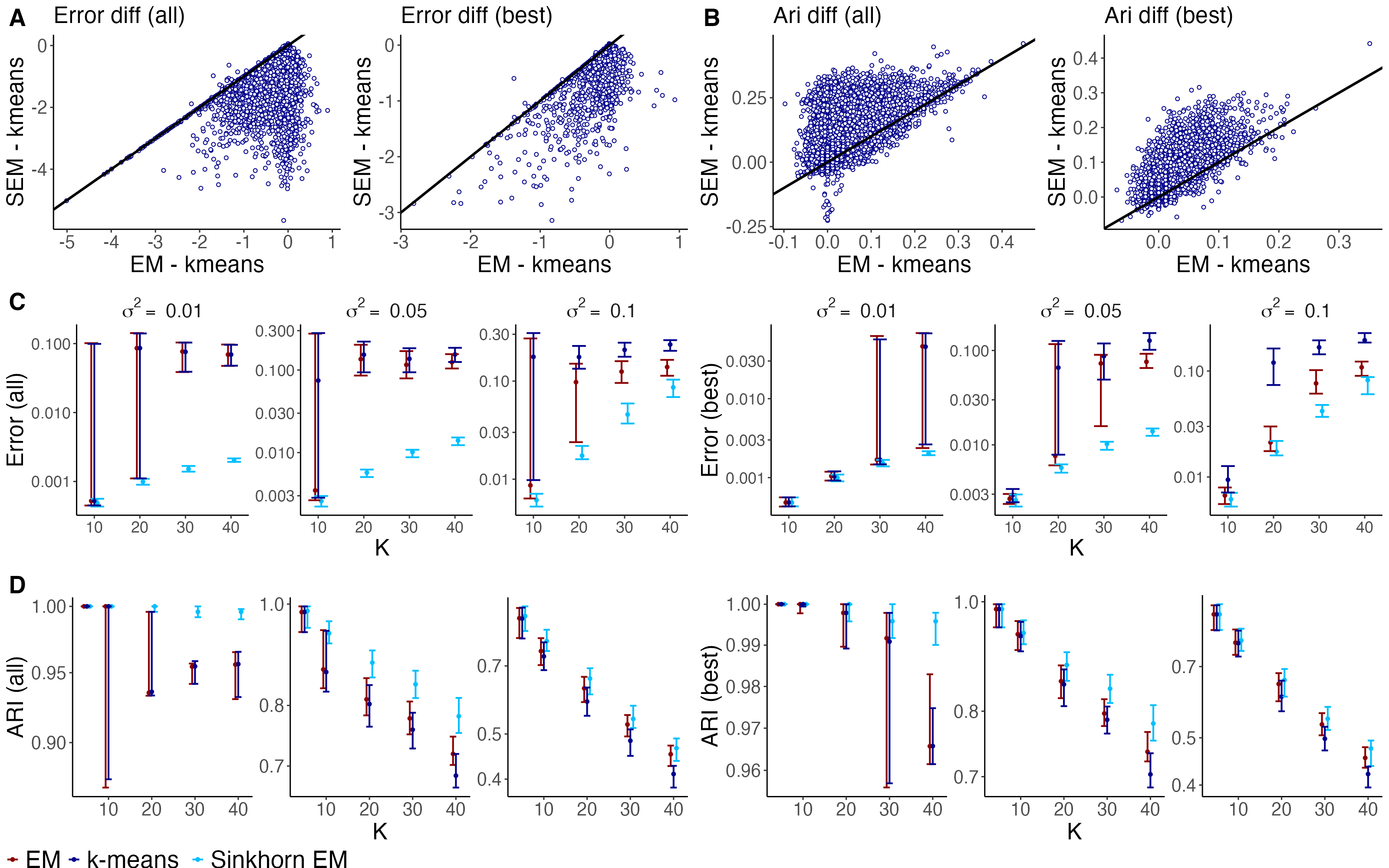}
  \end{center}
  \caption{Results for experiment in Section \ref{sec:experiments} with known weights $\alpha_k=1/K$ and known $K$, for a mixture of spherical Gaussians with variance $\sigma^2$. \textbf{A}: Error difference between Sinkhorn-EM and k-means (y-axis) and between EM and k-means (x-axis). In the left plot, each random seed is considered individually; in the right plot, the best of five seeds is considered. \textbf{B}: same as \textbf{A} but with ARI score. \textbf{C}: Comparison of errors for each algorithm for varying $K$ and $\sigma^2$. The error bar represents the interquartile range. \textbf{D}: same as \textbf{C} but with ARI score.}
  \label{fig:errorari}
\end{figure}

\vspace{-0.7cm}
\section{Application: segmentation in C.elegans microscopy}
\label{sec:celegans}
One of the conclusions of Section \ref{sec:experiments} is that SEM is most beneficial when the number of clusters $K$ is large. In this Section, we show that this observation holds for a real-world neural segmentation task that can be framed as fitting a GMM with multiple components (neurons).


\textit{C. elegans}~is a roundworm used as a model organism in neuroscience for decades due to its stereotypic brain organization and simple structure consisting of 302 neurons. Automated neuron identification and segmentation of \textit{C. elegans} is crucial for conducting high-throughput experiments for relevant biological applications. This task can be facilitated by technologies such as NeuroPAL, \citep{yemini2021neuropal}, a novel transgenic strain of \textit{C. elegans} that features deterministic coloring of each neuron, enabling the disambiguation of nearby neurons and aiding in their identification \citep{Varol2020}. 

We treat the segmentation problem of neurons in NeuroPAL \textit{C.elegans} images as one of clustering using Gaussian Mixture Models. Given a colorful volume represented by six dimensions (three spatial coordinates and three dimensions for the RGB colors), we aim to recover the location centers $\theta_k$ of the $K$ neurons in the recorded volume, along with their shapes $\Sigma_k$. We can then segment the image using the responsibilities $\Psi_k$, which encode the probabilistic assignment of the pixels to the cells.

In Fig. \ref{fig:celegans}, we compare segmentation using the EM algorithm as a baseline to compare against. Fig. \ref{fig:celegans}B illustrates the primary mode of failure of EM that explains why SEM outperforms it: EM will typically collapse two nearby cells into a single component, a pathology that SEM most often avoids. This collapse can be precisely corresponded to the \textit{many-fit-one} local optima described in Theorem \ref{theo:local} in Section \ref{sec:local}. These results indicate that SEM is a valuable alternative to the more traditional EM in this real-world setup characterized by a dense mixture of Gaussians with many ($K>20$) clusters. Details of our experiments appear in Appendix \ref{app:eleg}. 


\begin{figure}
  \begin{center}
    \includegraphics[width=1.0\textwidth]{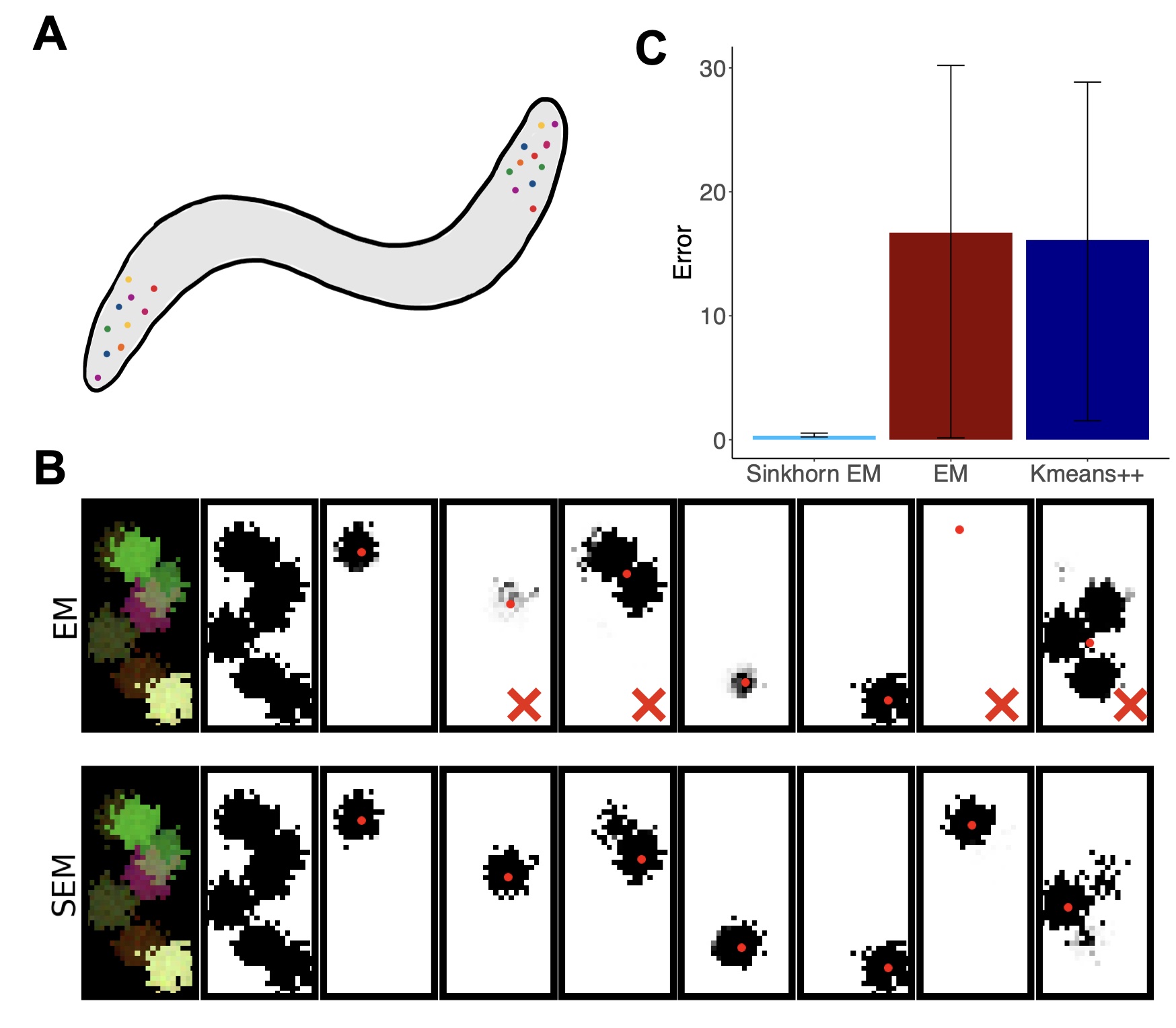}
  \end{center}
  \caption{Results on a C.elegans segmentation task. \textbf{A}: An illustration of a worm's brain. The task is to identify different locations (colored dots) in volumetric images. \textbf{B}: Comparison between a typical segmentation outcome of EM  and Sinkhorn-EM algorithms. The first column shows a true microscopic image containing a subset of neurons. The second column shows pixels containing neurons. All remaining columns show identified neurons for Sinkhorn-EM and EM methods as indicated by the responsibilities $\Psi_k$ for each neuron over pixels (grey scale). EM tends to collapse neuronal shapes, a problem that Sinkhorn-EM averts. Red dots indicate true neural centers, and red crosses indicate failures in the identification of individual components. \textbf{C}: Sinkhorn-EM consistently leads to better segmentation performance as compared to EM and k-means++ competitors.}
  \label{fig:celegans}
\end{figure}

\section{Application to co-clustering in spatial transcriptomics}
\label{sec:coclustering}

Finally, we extend our methodology to model-based co-clustering and demonstrate its benefits on a high-dimensional genomics dataset. Given a matrix $Y$ with dimensions $N$ and $M$, the aim of co-clustering (or biclustering) \citep{govaert2013co} is to simultaneously cluster the rows and columns of $X$, as opposed to the usual setup where only the rows of $X$ (i.e., observations) are clustered. The co-clustering problem has a long history in statistics (see, e.g., \cite{hartigan1972direct}) and has found relevant applications in fields such as text analysis \citep{dhillon2001co} and bioinformatics \citep{cheng2000biclustering,tan2014sparse}. Here, we focus on the model-based formulation \citep{bouveyron2019model,govaert2013co}, which maximizes a loss function (e.g., the log-likelihood) that depends on the density of the observed data given the model parameters. Model-based co-clustering is a case in which we expect to benefit from entropic OT, since likelihood-based approaches are known to be severely affected by bad local optima \citep{bouveyron2019model}.

The underlying probabilistic model here is the so-called latent block model. For a co-clustering model with $K$ clusters in the row dimension and $G$ clusters in the column dimension, let $z\in\mathbb{R}^{N\times K}$ be a binary matrix representing latent assignments of rows in $Y$ to a class $k\in[K]$ (so that $\sum_k z_{i,j}=1$). Likewise, we represent latent column assignments with a matrix $w\in\mathbb{R}^{M\times G}$. Then, there are $K\times G$ possible assignments for a particular $Y_{i,j}$ of $Y$. We assume that these entries are conditionally independent knowing $z$ and $w$ so that for certain parametric family of densities $\phi(\cdot, \theta)$ we express 
\begin{equation*}
P(Y|z,w,\theta)=\prod_{i,j,k,g} \phi\left(Y_{i,j},\theta_{k,g}\right)^{z_{i,k}w_{j,g}}.
\end{equation*}
If we denote the row and column mixture proportions $\pi_k=P(z_{i,k}=1)$ and $\rho_g = P(w_{j,g}=1)$ (these could also be treated as parameters), then the marginal likelihood of $Y$ writes as the following mixture
\begin{eqnarray}
\label{eq:coclust}  q_{\theta}(Y)= \sum_{(z,w)\in\mathcal{Z}\times \mathcal{W}}p(z)p(w)f(Y|z,w,\theta)
= \sum_{(z,w)\in\mathcal{Z}\times \mathcal{W}}\prod_{i,k}\pi_k^{z_{i,k}}\prod_{j,g} \rho_g^{w_{j,g}} \prod_{i,j,k,g}\phi\left(Y_{i,j},\theta_{k,g}\right)^{z_{i,k}w_{j,g}}.
\end{eqnarray}
Evaluation of the above likelihood is intractable, a problem that carries forward to the computation of the E step for the EM algorithm targeting \eqref{eq:coclust}: if the above equation corresponded to a usual mixture model (e.g. making $G=1$) then $z_i$ would be uniquely associated to a $Y_{i,}$ and we would be able to express the likelihood as a product over the \textit{sample}. However, the complex simultaneous dependence on $z$ and $w$ in \eqref{eq:coclust} prevents us from achieving this sample representation for the likelihood above, so a sum with exponentially many terms needs to be computed. 

It is still possible to deal with intractability with approximate methods: here we consider a Variational EM algorithm (VEM) \citep{bouveyron2019model,nadif2008algorithms} based on a factored approximation for the joint conditional probability $P(z_{i,k}w_{j,g}=1|Y,\theta)\approx P(z_{i,k}=1|Y,\theta)P(w_{j,g}=1|Y,\theta)$. This algorithm, detailed in the Appendix \ref{app:coclust}, corresponds to the iterative alternate application of the usual EM algorithm to cluster the rows and columns of $Y$ until convergence. The main drawback of this algorithm is sensitivity to initial values due to pervasive bad-local optima. To address this issue, we implement Sinkhorn VEM (SVEM), which replaces each EM instance with SEM. We show in Appendix \ref{sub:coclust} that SVEM outperforms both VEM and an off-the-shelf spectral biclustering method on simulated data.
\subsection{Application to Spatial Transcriptomics}

Spatial Transcriptomics is an umbrella term for a set of technologies that enable transcriptomic (i.e., gene-expression) profiling of samples (e.g., single-cell RNA sequencing) with spatial resolution \citep{moses2022museum}. A typical spatial transcriptomics experiment is represented by a matrix $Y\in\mathbb{R}^{N\times M}$, where $M$ genes are measured at $N$ locations, or \textit{spots}. This type of structure motivates two intertwined research questions: i) whether this high-dimensional expression data clusters coherently in regions in a way that resembles the anatomy of the tissue, and ii) whether there are genes whose expression level depends on space. As suggested by \cite{sottosanti2023co}, it is possible to simultaneously address these two questions from a co-clustering perspective.

Specifically, \cite{sottosanti2023co} develop probabilistic models with EM-type and MCMC-based inference procedures that explicitly encode spatial dependencies. In contrast, we consider an entropic optimal transport formulation that does not explicitly model spatial structure, and use spatial transcriptomics data as a testbed to study the empirical behavior of Sinkhorn EM as a clustering algorithm. In this setting, we observe that Sinkhorn EM can yield more coherent tissue partitions than standard EM under comparable experimental protocols.


We compared the performance of three methods on one dataset from the human dorsolateral prefrontal cortex (DLPFC) from the Liebler Institute for Brain Development \citep{maynard2021transcriptome}. Results are shown in Fig. \ref{fig:spatialtrans}. The clustering of spots produced by SVEM more faithfully reflects the tissue's layer organization, even though no spatial information has been explicitly encoded in the model.

\begin{figure}
  \begin{center}
    \includegraphics[width=1.0\textwidth]{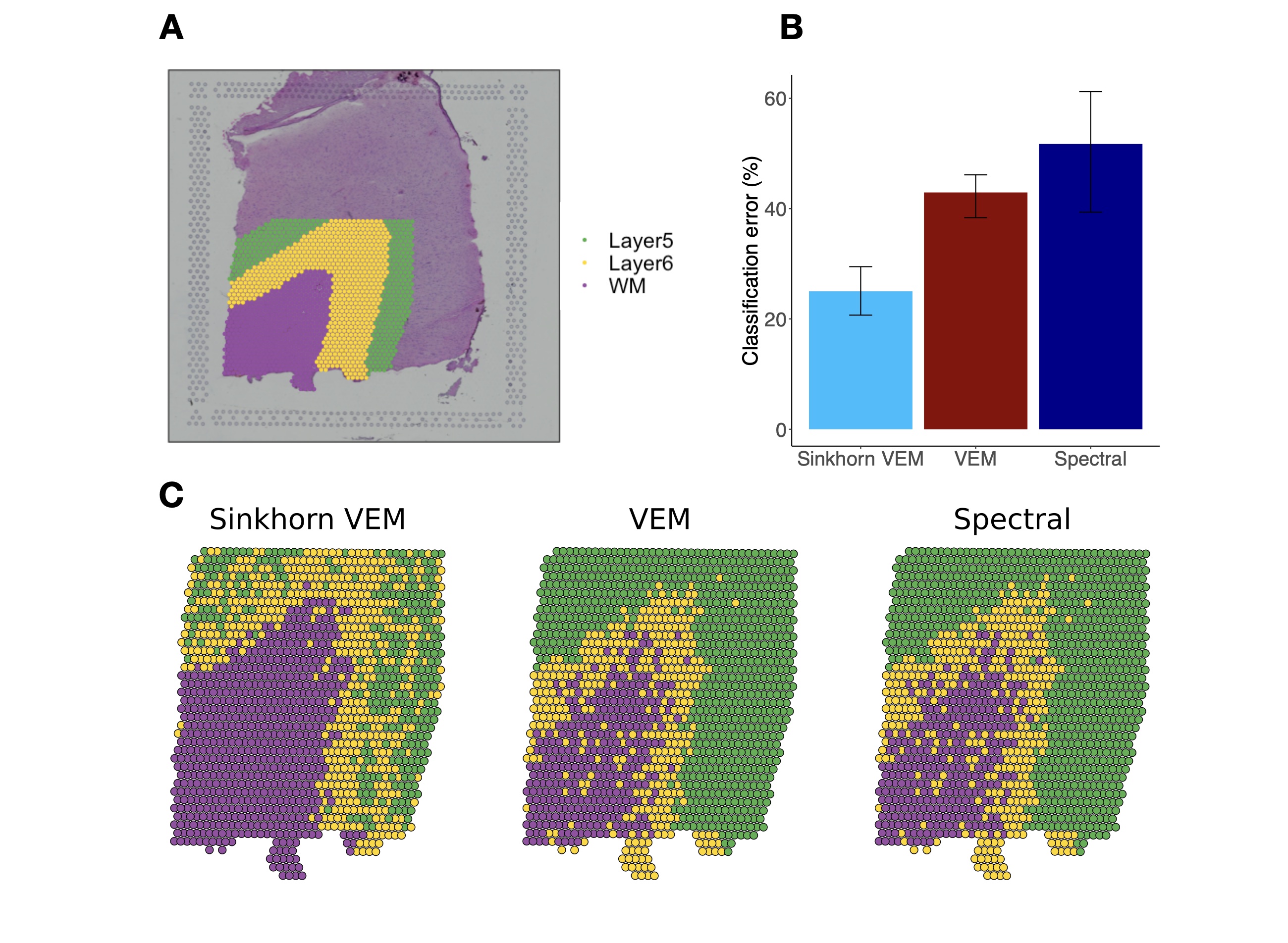}
  \end{center}
  \caption{Results for a Spatial Transcriptomic experiment. \textbf{A} Tissue sample of DLPFC, with colored dots representing spatial locations corresponding to three layers (Layer 5, Layer 6, and White Matter) where gene expression vectors were available. \textbf{B}. Comparison of clustering performance across the three methods, averaged over 100 repetitions. Bars represent 95\% confidence intervals. \textbf{C}: example of a typical recovered solution. SVEM most faithfully captures the underlying histological characterization, even when the model includes no spatial information.}
  \label{fig:spatialtrans}
\end{figure}

\section{Discussion and future directions}
\label{sec:disc}
Entropic optimal transport provides a valuable methodology for model-based clustering, often avoiding the pathologies that the log-likelihood would otherwise exhibit. Although our most specialized theoretical results pertain to the inference of location parameters in GMMs with equal variances, we stress that the validity of our methodology is more general in many aspects: first, the convergence of Sinkhorn-EM to stationary points is still guaranteed even in finite samples and a generic location-scale mixture model under mild assumptions of the base density. Additionally, scale or variance parameters can be optimized over, and our empirical results show superior performance in this regime. Third, it is also possible to optimize over weights $\alpha$, just as with the log-likelihood, although gradient-based optimization is required to optimize the entropic optimal transport loss in this case, unlike the log-likelihood, where weight updates are naturally implemented within the EM algorithm. 

We demonstrated the superiority of Sinkhorn-EM not only in simulated experiments but also in two real-world applications: neural segmentation and tissue clustering in spatial transcriptomics. Regarding neural segmentation, this task was originally introduced in \cite{nejatbakhsh2020probabilistic}. However, the tasks studied there involved a much lower number of clusters, and comparisons were against more naive baselines. Further, unlike \cite{nejatbakhsh2020probabilistic}, we provide a theoretical explanation for these results (Theorem \ref{theo:local}). The application of spatial transcriptomics also illustrates the versatility of our framework: by iteratively calling SEM as a subroutine, we were able to target a biclustering objective via variational inference, extending the benefits observed in the elementary case.

These promising results motivate two lines of future work. First, our Theorem \ref{theo:local} describes a particular setup where we avoid a type of spurious configuration. However, a characterization of local-optima structures for $\Ll$ remains lacking. One promising approach to such characterization would necessitate extending the results of \cite{chen2020likelihood} to the case of unequal weights. Second, although our methodology can be applied in both finite-sample and population regimes (Theorem \ref{teo:em} and Proposition \ref{prop:L} (a) and (c)), our most involved theoretical results (Theorems \ref{theo:local}, \ref{teo:mixtureloss} and \ref{teo:mixture}) are based on a study of the population landscape of $\Ll$. However, the fact that our successful empirical results are based on samples (sometimes with modest size) motivates a more in-depth analysis of the finite case. Regarding the EM algorithm, several thorough finite-sample convergence analyses have appeared in recent years (e.g., \cite{BalWaiYu17,Dwivedi2018}), all of which rely on establishing a non-asymptotic uniform law of large numbers to control the difference between the empirical $F^n$ and the population $F$ iterates. Unfortunately, establishing such a law for SEM would require a non-asymptotic convergence-rate analysis of the empirical transportation plans $\pi^n$ in an unbounded setting, which is currently unavailable. 
Still, a flurry of recent results in this direction \citep{groppe2024lower,pooladian2023minimax,masud2023multivariate,rigollet2025sample,sadhu2024stability,hundrieser2024limit,klatt2020empirical} provides us with a novel set of tools for a finite-sample and asymptotic analysis of $\Ll$ minimization.

Our approach has two main limitations. First, although the inferential gains are substantial, computation of $\Ll$ signifies a non-negligible extra computational burden, so the use of our methodology should focus on cases where the gains justify the extra cost; i.e., where we might not get the right optimum even when trying from many different seeds. The second limitation is that $\Ll$ may still have many bad local optima. We found that convergence to spurious local optima was more prevalent in mixtures with highly unequal weights; in such cases, Sinkhorn-EM did not substantially outperform EM.

\section{Acknowledgements}
Part of this research was performed while the author was visiting the Institute for Mathematical and Statistical Innovation (IMSI), which is supported by the National Science Foundation (Grant No. DMS-1929348). This work is supported by NSF DMS-2412895. This work used  Bridges-2 at Pittsburgh Supercomputing Center through allocation  MTH230027 from the Advanced Cyberinfrastructure Coordination Ecosystem: Services \& Support (ACCESS) program, which is supported by NSF Grant 2138259, 2138286, 2138307, 2137603, and 2138296. The author expresses gratitude to Jonathan Niles-Weed, Amin Nejatbakhsh, and Erdem Varol for insightful discussions and their contributions to a preliminary version of this work \cite{mena2020sinkhorn}.

\putbib

\end{bibunit}
\clearpage

\pagebreak 
\pagenumbering{arabic}
\begin{bibunit}
\section*{Supplementary materials: On model-based clustering with entropic optimal transport}
\begin{appendices}
\renewcommand{\thetheorem}{S.\arabic{theorem}}
\setcounter{theorem}{0}
\renewcommand{\thelemma}{S.\arabic{lemma}}
\setcounter{lemma}{0}
\renewcommand{\thecorollary}{S.\arabic{corollary}}
\setcounter{corollary}{0}
\renewcommand{\theproposition}{S.\arabic{proposition}}
\setcounter{proposition}{0}
\renewcommand\thefigure{\thesection.\arabic{figure}}    
\setcounter{equation}{0}\renewcommand\theequation{S\arabic{equation}}
\setcounter{figure}{0}    
\section{Extended preliminaries and statistical setup}

\label{sec:preliminaries}
Let's first state an existence and uniqueness result for the primal entropic optimal transport problem $S(P,Q)$ in \eqref{eq:sinkdef} with generic measures $P,Q$. It follows from Proposition 2.5 in
\cite{leonard2013survey} that under the cost integrability assumption $\mathbb{E}\left(c(X,Y)\right)<\infty$ (where the expectation is taken with respect to $(X,Y)\sim P\otimes Q$) there is a unique plan $\pi_\ast$ achieving the infimum in the primal problem. We now turn to the dual problem. It follows from e.g. Proposition 2.6 in \cite{leonard2013survey}, Section 2 in \cite{bercu2021asymptotic} and Proposition 4.3 in \cite{chewi2024statistical}
that this implies the existence of a dual potentials $(\omega_\ast,\nu_\ast)\in L_1(P)\times L_1(Q)$ 
Solving the dual problem
\begin{equation}\label{eq:semidual00} S(P,Q)=\sup_{\omega\in L_1(P),v\in L_1(Q)} \int \omega(x)\dd P(x) + \int \nu(y)\dd Q(y)  -  \int \left(e^{\omega(x)+\nu(y)-c(x,y)}-1\right)\dd P(x)\dd Q(y),
\end{equation}
and that \begin{equation}\label{eq:optpi} d\pi_\star(x,y) = \exp\left(\omega_\star(x)+\nu_\star(y)-c(x,y)\right) \dd P(x)\dd Q(y).\end{equation}
Moreover, by integrating the above expression with respect to $P$ we obtain that 
$$\nu_\star(y)=-\log\left(\int \exp(\omega_\star(x)-c(x,y))\dd P(x)\right)\quad{Q\text{-a.s.}}.$$
replacing the above expression in \eqref{eq:semidual00} yields the semi-dual formulation \eqref{eq:semidual0}, and in particular, the existence of the semidual potential $\omega_\star$. Note that uniqueness of $\pi_{\star}$ and \eqref{eq:optpi} imply that the pair $(\omega_\star,\nu_\star)$ is uniquely defined up to the translation $(\omega_\star+L,\nu_\star-L)$.  

Finally, note that can also express \eqref{eq:optpi} in terms of a single potential by replacing the above $\nu_\star(y)$, arriving at the following alternative representation
\begin{equation}\label{eq:optpisemi} d\pi_\star(x,y) = \dfrac{\exp\left(\omega_\star(x)-c(x,y)\right) \dd P(x)\dd Q(y)}{\int \exp\left(\omega_\star(x)-c(x,y)\right) \dd P(x)}.\end{equation}

\subsection{Statistical setup}\label{sub:setup}

To state the following results, we will assume that $P=P_\theta$ and that $Q$ represents the data distribution. This $Q$ could be either the population distribution $Q_\ast$ or the empirical distribution $Q=Q_N=\frac{1}{N}\sum_{i=1}^N \delta_{Y_i}$, where $Y_i\sim Q_\ast$. In some results (e.g., Proposition \ref{prop:L}(c), (e)), we will need to assume assume that we are in the well-specified setup, i.e., that $Q=Q_{\theta^\ast}$ where $\theta^\ast$ is the true data-generating parameter.

We will make the following assumptions on the parameterization of each model component $(\theta_k,\nu_k)\to q^{\nu_k}_{\theta_k}(y)$, that we denote simply as $q_{\theta_k}(y)$, for $k=1,\ldots K$.
\begin{enumerate}[label=(C\arabic*),ref=C\arabic*]
\item\label{C0} For each $\theta_k$ ,$\int \log q_{\theta_k}(y)dQ(y)>-\infty$.
\item\label{C1} The parameter space $\Theta$ is a subset of $\mathbb{R}^{d'}$ for some $d'$.
\item\label{C2} For each $y\in\mathbb{R}^d$, $q_{\theta_k}(y)$ is three times continuously differentiable.
\item\label{C3} For each $y\in\mathbb{R}^d$, $q_{\theta_k}(y)>0$.
\item\label{C4} Denote the model score $s_k(y,\theta_k)=\nabla_{\theta_k} \log q_{\theta_k}(y) $. We assume that for each $\tilde{\theta}_k$ there is a ball $B(\tilde{\theta}_k,\varepsilon)$ such that for any $1\leq k\leq K$
$$\sup_{\tilde{\theta}_k\in B(\theta_k,\varepsilon)} |s(\tilde{\theta}_k,y)|^4 \leq  f_1(y) $$
$$\sup_{\tilde{\theta}_k\in B(\theta_k,\varepsilon)}  |\nabla_{\theta_{k}} s(\tilde{\theta}_k,y)|^2 \leq  f_2(y) $$
$$\sup_{\tilde{\theta}_k\in B(\theta_k,\varepsilon)}  |\nabla^2_{\theta_k} s(\tilde{\theta}_k,y)| \leq  f_3(y) $$
for some functions $f_i$ such that for $i=1,2,3$, $\int |f_i(y)|dQ(y)<\infty,$ where all inequality symbols above are understood component-wise over all possible partial derivatives with respect to multidimensional parameters.
\item\label{C5} Each model component satisfies the coercivity condition $q_{\theta^k}(y)\to 0$ as $||\theta_k||\to\infty.$
\item \label{C6} There is an envelope function $E(y)$ such that for any $y\in\mathbb{R}^d$, $\theta\in\Psi$ and $\alpha\in[0,1]^K$ with $\sum_{k=1}^K\alpha_k=1$;
$$\log\left(\sum_{k=1}^K \alpha_k q_{\theta_k}(y)\right)<E(y),$$
with $\int |E(y)|dQ(y)<\infty$.
\item \label{C7}  Each of the $\mathbb{E}\left(\log q_{\theta_k}(Y)\right)=\int \log q_{\theta_k}(Y)dQ(Y)$ is quasi-concave in $\theta_k$.
\end{enumerate}
\begin{remark}
We require assumption \ref{C1} to guarantee that $\Ll(\theta)$ is well defined in terms of an entropic optimal transport problem. Assumptions \ref{C1},\ref{C2},\ref{C3},\ref{C4} guarantee sufficient smoothness of the negative log-likelihood $\ell$, and in turn, of $\Ll$. We ask for one extra degree of smoothness than would be typically required for $\ell$ in order to apply a recursive argument to obtain regularity for $\Ll$. Assumptions \ref{C5} and \ref{C6} are required to ensure coercivity of $\Ll$, which will be required to ensure that SEM converges to stationary points of this function. Assumption \ref{C7} is required to ensure the convergence of EM and SEM. Note that these conditions are satisfied on a location-scale family if we constrain the scale parameter $\nu_k$ to be bounded away from zero (or from a non-invertible matrix, in the multivariate case), if the underlying template density $g$ is bounded, sufficiently smooth, assigns mass to any point, and if the first three derivatives of $\log g$ are integrable with respect to $Q$. Indeed, in this case we have that $$q_{\theta_k}(y)=\frac{1}{|\det \nu_k|}g\left(\nu^{-1}_k(y-\theta_k)\right),$$
and $q_{\theta_k}(y)\to 0$ clearly if either $||\theta_k||\to\infty$ or $||\nu_k||\to\infty$. Additionally, 
$$\log\left(\sum_{k=1}^K \alpha_k q_{\theta_k}(y)\right)\leq \log\left(\max_{k=1,\ldots,K}q_{\theta_k}(y)\sum_{k=1}^K \alpha_k\right)\leq -\min_{k=1,\ldots,K} \log |\det \nu_k| + \log ||g||_\infty ,$$
which is integrable if $|\det \nu_k|$ are all bounded away from zero. Among families that satisfy the above conditions are the usual mixture of Gaussians, Laplace and Student-t distributions.

\end{remark}

We will identify the cost function with the mixture components $q_{\theta_k}$ via the relationship $c(x,y)=-\log q_x(y)=-\log q(y-x)$, where $q$ (without subscript) represents a template density. More generally, if there are component-specific variance or scale parameters (so that $q(y)=q^\nu(y)$, we will model them by choosing a suitable cost function $c(x,y)=-\log q^{\nu(x)}(y-x)$. For example, with the choice $c(x,y)=c_{\nu_x}(x,y)=\frac{1}{2}(x-y)^\top{M^{-1}(x)}(x-y)+\frac{1}{2}\log \det M(x) +\frac{K}{2}\log 2\pi$ we can represent all mixtures of multivariate Gaussians. Indeed, the mixture locations $\theta_k$ and weights $\alpha_k$ are encoded in the atomic measure $P_\theta$. The corresponding variances are encoded by a suitable choice of the function $M(x)$: in this case, it suffices to consider $M(x)=\Sigma_k$ whenever $x=\theta_k$; as we only evaluate $c$ where $P_\theta$ concentrates, the remaining choices of $M(x)$ are irrelevant. For the sake of simplicity, we will hide dependencies on scale parameters hereafter and write $\theta$ to represent both location and scale parameters. 

Once we have introduced $\theta$ and the relevant measures $P=P_\theta$ and $Q$, we will write the optimal $\omega_\star$ $\pi_\star$  as $\omega(\theta)$ and $\pi_\theta$, respectively. Since $P_\theta$ is supported on $K$ points, we can assume that $\omega(\theta)=\omega_\star\in \mathbb{R}^K$.  Note that we can relate $\omega(\theta)$ and $\pi_\theta$ using \eqref{eq:optpisemi}. In this case, if we write $x=k$ to represent $x=\theta_k$ we have
\begin{eqnarray} \nonumber d\pi_\star(k,y) &=&  
\dfrac{\alpha_k e^{\omega_k(\theta)} q_{\theta_k}(y)}{\sum_{k'=1}^K \alpha_{k'} e^{\omega_{k'}(\theta)} q_{\theta_{k'}}(y)}\dd Q(y)\\
&=&\label{eq:optpisemi2}
\dfrac{\alpha_k(\theta) q_{\theta_k}(y)}{\sum_{k'=1}^K \alpha_{k'}(\theta) q_{\theta_{k'}}(y)}\dd Q(y).
\end{eqnarray}
where we have defined
$\omega$-tilted version of $\alpha$:
\begin{equation} \label{eq:alphatheta1}\alpha_k(\omega):=\frac{\alpha_{k}e^{\omega_k}}{\sum_{k'=1}^K \alpha_{k'}e^{\omega_{k'}}}.\end{equation}
and we denote $\alpha(\theta)$ as a shortcut for $\alpha(\omega(\theta))$. Note that this definition is consistent with \eqref{eq:alphatheta} in the main text.
From this, equation \eqref{eq:respsinkhorn} in the main text follows immediately from the conditioning relation $d\pi(x,y)=d\pi(x|y)dQ(y)$.
By setting the last coordinate of $\omega_K(\theta)=0$, we can assume that $\omega(\theta)=(\omega_1(\theta),\ldots, \omega_{K-1}(\theta),0)$ is uniquely defined (Lemma \ref{lemma:alphaprop}) and so we identify $\omega(\theta)$ as an element in $\mathbb{R}^{K-1}$. 
By definition, for fixed $\theta$, $\omega(\theta) $ maximizes the following function $\tilde{\Ll}(\theta,\omega)$ with respect to $\omega$.
\begin{equation}\label{eq:tildeL} \tilde{\Ll}(\theta,\omega)=  \sum_{k=1}^K \alpha_k\omega_k -\mathbb{E} \left(\log \left( \sum_{k=1}^K \alpha_k e^{\omega_k}q_{\theta_k}(Y)\right)\right).\end{equation}

\section{Proofs and ommited results in Section \ref{sec:framework}}
\subsection{Auxiliary results \ref{prop:L}}

\begin{lemma}\label{lemma:three}
Suppose that assumptions \ref{C1},\ref{C2},\ref{C3},\ref{C4} hold. Then $\ell(\theta)$ and $\tilde{\Ll}(\theta,\omega)$ are three times continuously differentiable. Moreover $\ell(\theta,\alpha)$ is also differentiable if $\alpha\in\{\alpha\in[0,1]^K, 0<\sum_{k=1}^K \alpha_k<1,\alpha_k\}$ is treated as an additional parameter.
\end{lemma}
\begin{proof}
We start by showing that each model component $\ell_k(\theta):=\mathbb{E}(\log q_{\theta_k}(Y))$ three times continuously differentiable. This is a direct consequence of \ref{C2} and \ref{C4}, enabling us to exchange the derivation and expectation signs \citep{folland1999real}. To show that $\ell(\theta)$ is three times continuously differentiable, we will apply the same argument to the mixture. Consider now the negative log-likelihood function
$$\ell(\theta,y)=-\log\left(\sum_{k=1}^K \alpha_k q_{\theta_k}(y)\right),$$
so that $\ell(\theta)=\mathbb{E}\left(\ell(\theta,Y)\right).$ Note that
$$\nabla_{\theta_k} \ell(\theta,y)=\Psi_k(y,\theta,\alpha)s_k(\theta_k,y) ,$$
where $\Psi_k(\theta,y,\alpha)$ is defined in \eqref{eq:respsinkhorn}.
By condition \ref{C3}, $\Psi_k(y,\theta,\alpha)$ is well-defined (i.e., the denominator is not zero).
We can iteratively compute further derivatives using the product rule
$$\nabla_{\theta_k} \nabla_{\theta_{k'}} \ell(\theta,y)=\nabla_{\theta_{k'}}\Psi_k(y,\theta,\alpha)s_k(\theta_k,y)+\Psi_k(\theta,y,\alpha)\nabla_{\theta_{k'}}s_k(\theta_k,y) ,$$
and 
\begin{eqnarray}\nonumber \nabla_{\theta_k} \nabla_{\theta_{k'}}\nabla_{\theta_{k''}} \ell(\theta,y)&=&\nabla_{\theta_{k'}}\nabla_{\theta_{k''}}\Psi_k(y,\theta,\alpha)s_k(\theta_k,y)+\nabla_{\theta_{k'}}\Psi_k(y,\theta,\alpha)\nabla_{\theta_{k''}}s_k(\theta_k,y)\\ \label{eq:three} \nonumber &&+\nabla_{\theta_{k''}}\Psi_k(Y,\theta,\alpha)\nabla_{\theta_{k'}}s_k(\theta_k,y) +\Psi_k(y,\theta,\alpha)\nabla_{\theta_{k''}}\nabla_{\theta_{k'}}s_{k'}(\theta_{k'},y).\end{eqnarray}

Let's evaluate subsequent derivatives of $\Psi_k(y,\theta,\alpha)$. We have
\begin{equation}
    \label{eq:psider}
\nabla_{\theta_{k'}} \Psi_k(y,\theta,\alpha)=\Psi_k(y,\theta,\alpha)s_{k'}(\theta_{k'},y)1_{k=k'}-\Psi_{k'}(y,\theta,\alpha)\Psi_{k}(y,\theta,\alpha)s_{k'}(\theta_{k'},y)s_{k}(\theta_{k},y).
\end{equation}

And again, by iteratively applying \eqref{eq:psider} we find that $\nabla_{\theta_{k'}}\nabla_{\theta_{k''}} \Psi_k(y,\theta,\alpha)$ will consist on terms that can be expressed as polynomial expressions of $\Psi_{k}(y,\theta,\alpha)$ (up to third degree) times polynomials involving $s_k(\theta_{k},y)$ (up to third degree) or terms of the type $\nabla_{\theta_k}s_{k}(\theta_{k},y)$. Note that we can ignore all terms $\nabla_{\theta_{k'}}s_{k}(\theta_{k},y)$ if $k\neq k'$ since $s_k(\theta_k,y)$ only depends on component-specific parameters $\theta_k$. In order to justify differentiability of $\ell(\theta)$ we must show that if we take expectations in \eqref{eq:three} we can exchange expectation and derivative sign. This follows from our previous observations: indeed, first note that all polynomials of $\Psi_k(y,\theta,\alpha)$ are bounded since $\Psi_k(y,\theta,\alpha)<1$. Therefore, the first term in the right-hand side of \eqref{eq:three} can be bounded a sum of at most cubic polynomials of $s_k(\theta_k,y)$ or a quadratic polynomial of $s_k(\theta,y)$ times a gradient $\nabla s_k(\theta,y)$. In turn, these are bounded by either $f_1(y)^{3/4}$ (integrable, by Cauchy-Schwarz) and \ref{C4} or by $f_1^{1/2}(y)f_2^{1/2}(y)$. Again, this product is integrable by Cauchy-Schwarz and \ref{C4}. Regarding the second term, in \eqref{eq:three}, they can be bounded by either another quadratic polynomial $s_k(\theta_k,y)$ in times $\nabla_{\theta_k} s_k(\theta_k,y)$, whose expectation we already bounded. The third term is completely analogous. The four term is bounded by $f_3(y)$, which is integrable, by \ref{C4}. 

Consequently, $\ell(\theta)$ is three times differentiable, by the same domination argument that we applied to each individual log-likelihood. Continuity of the third derivative follows from a similar argument, using this time that the third derivative of $q_{\theta_k}(y)$ is continuous. Note also that the above argument extends trivially to the case where $\alpha$ is deemed a parameter. Indeed, in this case we consider the parameterization $(\theta_k,\nu_k,\alpha_k)\to \alpha_k q_{\theta_k}^{\nu_k}(y)$, which enjoys the same regularity as we established if $\alpha$ was kept fixed. Finally, it remains to show that $\tilde{\Ll}(\theta,\omega)$ is also three times continuously differentiable. This follows from the observation that 
\begin{eqnarray} \nonumber \tilde{\Ll}(\theta,\omega) &=& \sum_{k=1}^K\omega_k \alpha_k -\mathbb{E} \left(\log \left( \sum_{k=1}^K \alpha_k e^{\omega_k} q_{\theta_k}(Y)\right)\right)\\ \nonumber
 &=&\nonumber \sum_{k=1}^K\omega_k \alpha_k-  \mathbb{E}\left[\log\left(\sum_{k=1}^K \alpha_k e^{\omega_k}\right)\right]-\mathbb{E} \left(\log \left( \sum_{k=1}^K \alpha_k(\omega) q_{\theta_k}(Y)\right)\right) \\
 &=& \label{eq:tildeLell}\sum_{k=1}^K\omega_k \alpha_k-  \log\left(\sum_{k=1}^K \alpha_k e^{\omega_k}\right)+\ell(\theta,\alpha(\omega)),
\end{eqnarray}
where $\alpha(\omega)$ is defined in \eqref{eq:alphatheta1}. From this, $\tilde{\Ll}(\theta,\omega)$ is clearly finite. Also, since the first two terms are three times continuously differentiable in $\omega$. The last term also has this regularity by virtue of the regularity of the function $\ell(\theta,\alpha)$ that we just established, in combination with the fact that $\alpha(\theta)$ also has sufficient smoothness.
\end{proof}

\begin{lemma}\label{lemma:alphaprop}
For each $\theta$, The quantities $\alpha(\theta)$ in \eqref{eq:alphatheta} and $\pi_\theta$ in \ref{eq:L} are uniquely defined. Also, there is unique $\omega(\theta)\in\mathbb{R}^K$ (up to a constant translation) maximizing \eqref{eq:semidual}. 
\end{lemma}
\begin{proof}
Under the assumption that $\mathbb{E}(c(X,Y))<\infty$, existence and uniqueness of $\pi_\theta=\pi_\star(\theta)$ follow from the discussion in Section \ref{sec:preliminaries}. This discussion also guarantees the existence (up to fixing one coordinate) of $\omega(\theta)=\omega_\star(\theta)\in\mathbb{R}^K$. This also implies the existence and uniqueness of $\alpha(\theta)$, defined as a function of $\alpha(\theta)$. Thus, it remains to show that $\mathbb{E}(c(X,Y))<\infty$. This follows from the fact that in our case $c(x,y)=-\log q(y-x)$, so the conclusion follows from assumption \ref{C0}. Indeed, in this case
$$\mathbb{E}\left(c(X,Y)\right)=-\sum_{k=1}^K \alpha_k \int \log q_{\theta_k}(y)dQ(y)<\infty$$
\end{proof}
\begin{lemma}\label{lemma:alphader}
Under the assumptions \ref{C1},\ref{C2},\ref{C3},\ref{C4}, the function $\theta\to \omega(\theta)$ is twice continuously differentiable. Consequently, $\theta\to\alpha(\theta)$ and $\theta\to \pi_\theta(k,y)$ are also twice continuously  differentiable, for all $k=1,\ldots K$, $y\in\mathbb{R}^d$. 
\end{lemma}

\begin{proof}

As we fixed the last coordinate, we can assume that $\omega(\theta)\in \mathbb{R}^{K-1}$. By definition, $\omega(\theta) $ maximizes $\tilde{\Ll}(\theta,\omega)$ in \eqref{eq:tildeL} with respect to $\omega$. By Lemma \ref{lemma:three}, $\tilde{\Ll}$ is three times continuously differentiable. In particular, it is differentiable with respect to $\omega$, for fixed $\theta$.
 Therefore, $\omega(\theta)$ must satisfy the first order condition
\begin{equation}\label{eq:first0}0 = \nabla_\omega \tilde{\Ll}(\theta, \omega(\theta)).\end{equation}
We now rely on the implicit function theorem: for a given $\theta$, if the derivative of $\nabla_\omega \tilde{\Ll}(\omega,\theta)$ with respect to $\omega$ is invertible, then $\alpha(\theta)$ is differentiable on a neighborhood around $\theta$. This derivative is the Hessian of $\tilde{\Ll}(\theta,\omega)$ with respect to $\omega$. We show in Lemma \ref{lemma:hessian} that this Hessian is negative definite, and hence invertible. Therefore, $\omega(\theta)$ is differentiable and we differentiation with repect to $\theta$ in \eqref{eq:first0} yields
\begin{equation} \label{eq:second} \nabla^2_{\omega\theta}\tilde{\Ll}(\theta, \omega(\theta))+ \nabla^2_{\omega\omega}  \tilde{\Ll}(\theta, \omega(\theta))\nabla_\theta \omega(\theta)=0.\end{equation}
Solving for $\nabla_\theta \omega(\theta)$ above yields
$$\nabla_\theta \omega(\theta)= - \left[\nabla^2_{\omega\omega} \tilde{\Ll}(\theta, \omega(\theta))\right]^{-1} \nabla^2_{\omega\theta} \tilde{\Ll}(\theta, \omega(\theta)).$$
Note that the right-hand side is differentiable since it is the composition of differentiable functions ($\tilde{\Ll}$ is three times continuously differentiable and $\omega(\theta)$ is continuously differentiable). Therefore, $\omega(\theta)$ is twice continuously differentiable. The proofs for $\alpha(\theta)$ and $\pi_\theta(k,y)$ are analogous after noticing that they are (twice) differentiable functions of $\omega(\theta)$.
\end{proof}

\begin{lemma}\label{lemma:hessian}
The Hessian of $\tilde{\Ll}(\theta,\omega)$ with respect to $\omega$ is negative definite at any $\omega \in\mathbb{R}^{K-1}$
\end{lemma}
\begin{proof}

Note that, by elementary differentiation
\begin{eqnarray*} \nabla_\omega \tilde{\Ll} (\theta, \omega)&=& \alpha - \mathbb{E} \left(\frac{\alpha e^\omega q_\theta(Y)}{\sum_{i=1}^K \alpha_k e^\omega_k q_{\theta_k(Y)}}\right)\\
 &=& \alpha - \mathbb{E}\left(\Psi(Y,\theta,\alpha(\omega))\right).
\end{eqnarray*}
The expressions $\alpha e^\omega q_\theta(Y)$ and  $\Psi(Y,\theta,\alpha(\omega))$ are understood as the vectors with entries $\alpha_ke^{\omega_k}q_{\theta_k}(Y)$ and $\Psi_k(Y,\theta,\alpha(\omega))$ for $k\leq K-1$, respectively, and $\alpha(\omega)$ is the $\omega$-tilted version of $\alpha$ defined in \ref{eq:alphatheta1}):
Further differentiation yields
\begin{equation} \label{eq:seclw} \nabla^2_{\omega\omega}\tilde{\Ll}  (\theta, \omega)= -\mathbb{E}\left(Diag(\Psi(Y,\theta,\alpha(\omega))-\Psi(Y,\theta,\alpha(\omega))\Psi(Y,\theta,\alpha(\omega))^\top\right), \end{equation}

Now, define the (symmetric) matrix $I(\theta,\omega)$ as the extension of $\nabla^2_{\omega\omega}\tilde{\Ll}  (\theta, \omega)$ to the entire range of indexes $k=1,\ldots K$, i.e., $I(\theta,\omega)$ is the right hand side in \eqref{eq:seclw} viewed this time as a $K\times K$ matrix. To create this matrix, we will fill in the remaining entries using the missing coordinate $\Psi_{K}(Y,\theta,\alpha(\omega))$ and recalling that $\omega_K=0$.
By definition this matrix coincides with $ \nabla^2_{\omega\omega}\tilde{\Ll}(\theta,\omega)$ on the first $K-1$ row and column coordinates. Further, note that $-I(\theta,\omega)$ is the Laplacian matrix of a weighted graph since 
$-I(\theta,\omega)= D(\theta,\omega)- W(\theta,\omega)$ with weight matrix $W_{k,k'}(\theta,\omega)=\mathbb{E}\left(\Psi_k(Y,\theta,\alpha(\omega))\Psi_{k'}(Y,\theta,\alpha(\omega))\right)$ and diagonal $D(\theta,\omega)$ with $D(\theta,\omega)_{k,k}=\mathbb{E}\left(\Psi_{k}(Y,\theta,\alpha(\omega))\right)$, satisfying
\begin{align*} \sum_{k'=1}^{K} W(\theta,\omega)_{k,k'} =& \sum_{k'=1}^{K} \mathbb{E}\left(\Psi_k(Y,\theta,\alpha(\omega))\Psi_{k'}(Y,\theta,\alpha(\omega))\right)\\
=& \mathbb{E}\left(\Psi_k(Y,\theta,\alpha(\omega))\sum_{k'=1}^{K} \Psi_{k'}(Y,\theta,\alpha(\omega))\right)\\
=& \mathbb{E}\left(\Psi_k(Y,\theta,\alpha(\omega))\right)\\
=& D(\theta,\omega)_{k,k}.\end{align*}

Then, $I(\theta,\omega)$ is a negative weighted Laplacian matrix and
so $$x^\top I(\theta,\omega) x = \frac{1}{2}\sum^{K}_{k,k'} I(\theta,\omega)_{k,k'}(x_k-x_{k'})^2\leq 0.$$

The above expression is zero only if $x$ is a constant vector since all off-diagonal entries of $I(\theta,\omega)$ are strictly. Indeed, $I(\theta,\omega)_{k,k'}=-\mathbb{E}\left(\Psi_k(Y,\theta,\alpha(\omega)) \Psi_{k'}(Y,\theta,\alpha(\omega))\right)$ and for each $k\leq K$, $\theta_k \in\mathbb{R}^d$, $Y\in\mathbb{R}^d,\omega\in\mathbb{R}^{K-1}$ we have $0<\Psi_k(Y,\theta,\alpha(\omega))<1$ so that the quantity inside the expectation is strictly negative and will remain to be so after taking expectations.

Let's now conclude that $\nabla^2\tilde{\Ll}_{\omega\omega}(\theta, \omega)$ is negative definite. Suppose it is not, then, for some $z\in\mathbb{R}^{K-1}$
$z^\top \nabla^2_{\omega\omega}  \tilde{\Ll} (\theta, \omega) z=0$. We can pad this vector by a zero to obtain $\tilde{z}=(z_1,\ldots z_{K-1},0)\in\mathbb{R}^K$ so that $\tilde{z}^\top I(\theta,\omega)\tilde{z}=0$. By the above discussion, for this to be the case, we require that $\tilde{z}$ is a constant vector. In turn, by construction, this implies that $z=0$. Therefore, $\nabla^2_{\omega\omega} \tilde{\Ll}(\theta, \omega)$ is negative definite. 
\end{proof}
\subsection{Proof of equation \eqref{eq:semidual2}}\label{sub:semidual}
It is immediately implied by \eqref{eq:tildeLell} that
$$ \Ll(\theta)=\tilde{\Ll}(\theta,\omega(\theta)) =\sum_{k=1}^K\omega_k(\theta) \alpha_k-  \log\left(\sum_{k=1}^K \alpha_k e^{\omega_k(\theta)}\right)+\ell(\theta,\alpha(\omega(\theta))),$$

Also, note that
\begin{eqnarray*} H(\alpha|\alpha(\theta))&=&\sum_{k=1}^K \alpha_k \log \left(\frac{\alpha_k}{\alpha_k(\theta)}\right) \\
&=& \sum_{k=1}^K \alpha_k \log \alpha_k - \sum_{k=1}^K \alpha_k \log \alpha_k(\theta)\\
&=&\sum_{k=1}^K \alpha_k \log \alpha_k - \sum_{k=1}^K \alpha_k \log \frac{\alpha_k e^{\omega_k(\theta)}}{\sum_{k'=1}^K\alpha_{k'} e^{\omega_{k'}(\theta)}}\\
&=&\sum_{k=1}^K \alpha_k \log \alpha_k - \sum_{k=1}^K \alpha_k \log \frac{\alpha_k e^{\omega_k(\theta)}}{\sum_{k'=1}^K\alpha_{k'} e^{\omega_{k'}(\theta)}}\\
&=& \sum_{k=1}^K \alpha_k \log \left(\sum_{k'=1}^K\alpha_{k'} e^{\omega_{k'}(\theta)}\right) +\alpha_k \log \alpha_k  - \sum_{k=1}^K \alpha_k \log \alpha_k -\sum_{k=1}^K \alpha_k \log e^{\omega_k(\theta)}\\
&=& \log \left(\sum_{k=1}^K\alpha_{k} e^{\omega_{k}(\theta)}\right)-\sum_{k=1}^K \alpha_k  \omega_k(\theta).
\end{eqnarray*}
Therefore
\begin{equation} \label{eq:semidual22} \Ll(\theta)=H(\alpha|\alpha(\theta))+\ell(\theta,\alpha(\theta)),\end{equation}
so that \eqref{eq:semidual2} holds.

\subsection{Proof of Proposition~\ref{prop:L}}
\label{app:propL}
Regarding (a), in separate Lemmas \ref{lemma:alphaprop} and \ref{lemma:alphader} we establish that $\omega(\theta)$ and $\alpha(\theta)$ are well defined and twice continuously differentiable, and that $\tilde{\Ll}(\theta,\omega)$ in \eqref{eq:tildeL} is three times continuously differentiable. It only remains to show that $\Ll$ is twice differentiable. But this follows from the fact that $\Ll(\theta)=\tilde{\Ll}(\theta,\omega(\theta))$, an identity that we can differentiate twice. We already showed (b) in the main text. Let's show (c) detail. Note that in the well-specified setup ($Q=Q_{\theta^\ast}$), the coupling $\pi=\pi_\theta$ achieving the infimum in \eqref{eq:sinkdef2} is exactly the joint $Q^{X,Y}_{\theta^\ast}$ in \eqref{eq:joint} when $\theta=\theta^\ast$. To see this it suffices to show that this coupling is valid in the sense that $\pi_\theta\in \Pi(P_\theta,Q_{\ast})$: if that was the case, then $H(\pi_\theta|Q_{\theta}^{X,Y})=H(Q_{\theta}^{X,Y}|Q_{\theta}^{X,Y})=0$ and $\Ll$ in \eqref{eq:joint} cannot be minimized further since the term $H(Q_\ast|\mathcal{L})$ doesn't depend on $\pi$. Now, note that for every $y$ in the support of $Q_\theta$,
$$\int dQ_{\theta}^{X,Y}(x,y)= \int q_{x}(y)\dd P_\theta(x)\dd y= \sum_{k=1}^K \alpha_k q_{\theta_k}(y) dy=\dd Q_\theta(y),$$
and that for $x=\theta_k$ and $k=1,\ldots K$,

$$\int dQ_{\theta}^{X,Y}(x,y)= \int q_{x}(y)\dd P_\theta(x)\dd y=  \alpha_k \int q_{\theta_k}(y) dy=\alpha_k.$$
The above two displays imply that the marginals of $\pi_\theta$ are $Q_\theta$ and $P_\theta$. Since we are assuming that $\theta=\theta^\ast$ then the former marginal is $Q_{\theta^\ast}$, which, by assumption, coincides with the data generating mechanism $Q_\ast$. 

To prove (d) and (e), we must compute the first and second derivatives of $\Ll$. Note that we can write $\Ll(\theta)=\tilde{\Ll}(\theta,\omega(\theta))$ where $\tilde{\Ll}$ is defined in \eqref{eq:tildeL}. As in the proof of that Lemma, we identify $\omega(\theta)$ (the maximizer of $\tilde{\Ll}(\theta,\cdot)$) as a vector in $\mathbb{R}^{K-1}$. Since $\omega(\theta)$ and $\tilde{\Ll}$ are twice differentiable (Lemma \ref{lemma:alphader} and Lemma \ref{lemma:three}), we can obtain the first and second derivatives of $\Ll(\theta)$ using the chain rule 
\begin{equation}\label{eq:partialL}
\nabla_\theta \Ll(\theta) = \nabla_\theta\tilde{\Ll}(\theta, \omega(\theta)) + \nabla_\omega\tilde{\Ll}(\theta, \omega(\theta)) \nabla_\theta \omega(\theta),
\end{equation}
and  
\begin{eqnarray}
\nonumber \nabla^2_\theta \Ll(\theta) &=&\nabla^2_\theta \tilde{\Ll}(\theta, \omega(\theta))  + \nabla^2_{\omega\theta}  \tilde{\Ll}(\theta, \omega(\theta)) \nabla_\theta \omega(\theta)
 +  \left(\nabla^2_{\theta\omega}  \tilde{\Ll}(\theta, \omega(\theta))\right)^\top \nabla_\theta \omega(\theta) \\
\label{eq:sec} & &
 +\nabla_\theta \omega(\theta)^\top \nabla^2_{\omega\omega} \tilde{\Ll}(\theta, \omega(\theta)) \nabla_\theta \omega(\theta)+ \nabla_\omega \tilde{\Ll}(\theta, \omega(\theta)) \nabla^2_{\theta\theta} \omega(\theta).
\end{eqnarray}
 But by optimality properties of $\omega(\theta)$ in \eqref{eq:first0} and \eqref{eq:second} and using that the second derivative of $\tilde{\Ll}$ with respect to $\omega$ is invertible (Lemma \ref{lemma:hessian}), we get that for any $\theta$ 
\begin{equation}\label{eq:sec3}
\nabla^2_{\theta\theta} \Ll(\theta)  = \nabla^2_{\theta\theta} \tilde{\Ll} (\theta, \omega(\theta))  - \nabla^2_{\theta\omega} \tilde{\Ll}(\theta, \omega(\theta))\left(\nabla^2_{\omega\omega} \tilde{\Ll}  (\theta, \omega(\theta))
\right)^{-1} \nabla^2_{\omega\theta} \tilde{\Ll}(\theta, \omega(\theta)).
\end{equation}
Let's study the first term on the right-hand side.
Note that, since $\theta\in\int(\Omega_\alpha)$ (i.e., mixture weights are not parameters), for any $\omega$ and derivative order $l\in\{1,2\}$
\begin{eqnarray*} \nonumber \nabla^l_{\theta^l}\tilde{\Ll}(\theta,\omega) &=& -\nabla^l_{\theta^l} \left[\mathbb{E} \left(\log \left( \sum_{k=1}^K \alpha_k e^{\omega_k} q_{\theta_k}(Y)\right)\right)\right]\\ \nonumber
 &=&-\nabla^l_{\theta^l}\left[\mathbb{E} \left(\log \left( \sum_{k=1}^K \alpha_k(\omega) q_{\theta_k}(Y)\right)\right)\right] - \nabla^l_{\theta^l} \mathbb{E}\left[\log\left(\sum_{k=1}^K \alpha_k e^{\omega_k}\right)\right]\\
 &=&-\nabla^l_{\theta^l}\left[\mathbb{E} \left(\log \left( \sum_{k=1}^K \alpha_k(\omega) q_{\theta_k}(Y)\right)\right)\right].
\end{eqnarray*}
In the second-to-last line, we renormalized the weights $\alpha_k e^{\omega_k}$ to obtain the \textit{bona fide} probability measure $\alpha(\omega)$. In the last line, we used that this quantity doesn't depend on $\theta$. Note now that the term inside the expectation in the right-hand side above is the density of the mixture distribution $Q_{\theta,\alpha(\omega)}$
with mixture components parameterized by $\theta$ and tilted weights $\alpha=\alpha(\omega)$. Therefore, the expectation is simply the negative log-likelihood $\ell(\theta,\alpha(\omega))$ for a mixture with tilted weights. Consequently, differentiating with respect to $\theta$ yields the same formula as if we were differentiating the negative-log likelihood, and so we arrive at
$$ \nabla^l_{\theta^l}\tilde{\Ll}(\theta,\omega)= \nabla^l_{\theta^l} \ell(\theta,\alpha(\omega)),$$
so that we can re-state \eqref{eq:partialL}  as
\begin{eqnarray*}
\nabla_\theta \Ll(\theta) &=&\nabla_\theta \ell(\theta,\alpha(\omega(\theta))) +\nabla_\omega \tilde{\Ll}(\theta, \omega(\theta))\nabla_\theta \omega(\theta)\\
&=& \nabla_\theta \ell(\theta,\alpha(\theta)),
\end{eqnarray*}
where in the last line we used the optimality condition \eqref{eq:first0}. This establishes (d) (in our notation, $\alpha(\theta)$ is a shortcut for $\alpha(\omega(\theta))$. Additionally, we can re-express
\eqref{eq:sec3} as
\begin{equation*} \nabla^2_{\theta\theta} \Ll(\theta)  = \nabla^2_{\theta\theta}  \ell(\theta,\alpha(\omega(\theta))) - \nabla^2_{\omega,\theta} \tilde{\Ll}(\theta, \omega(\theta))\left(\nabla^2_{\omega\omega}\tilde{\Ll} (\theta, \omega(\theta))
\right)^{-1} \nabla^2_{\theta,\omega} \tilde{\Ll}(\theta, \omega(\theta)).
\end{equation*}
We will conclude the proof of (e) by specializing the above representation to the case where $Q=Q_{\theta^\ast}$ and $\theta=\theta^\ast$.  Note first that in this case $\omega({\theta^\ast})=0$: indeed, by (c),  $\pi_{\theta^\ast}=Q^{X,Y}_{\theta^\ast}$. It follows from \eqref{eq:optpisemi2} and the uniqueness of $\pi_{\theta^\ast}$ and of $\alpha(\theta)$ (Lemma \ref{lemma:alphaprop}) that
$\alpha(\theta^\ast)=\alpha(\omega({\theta^\ast}))=\alpha(0)=\alpha$, an so $\nabla^2 \ell(\theta,\alpha(\omega(\theta)))=\nabla^2\ell(\theta,\alpha)=\nabla^2 \ell(\theta)$.
 Therefore, 
\begin{equation*}\label{eq:sec4}
\nabla^2 \Ll(\theta^*) =  \nabla^2 \ell(\theta^*)  - \nabla^2_{\theta,\omega} \tilde{\Ll}(\theta, 0)\left(\partial^2_{\omega,\omega} \tilde{\Ll}  (\theta^*, 0)
\right)^{-1} \nabla^2_{\theta\omega} \tilde{\Ll}(\theta^*, 0),
\end{equation*}
and we conclude by identifying $$A(\theta^\ast)=-\nabla^2_{\omega\omega} \tilde{\Ll} (\theta^*, 0)
^{-1},B(\theta^\ast)=\nabla^2_{\theta,\omega} \tilde{\Ll}(\theta^*, 0). $$

It only remains to show (f). Let $\theta^n$ be a sequence with $||\theta^n||\to \infty$; we will show that  $\Ll(\theta^n)\to \infty$. Since, by \eqref{eq:semidual22} $\Ll(\theta)=H(\alpha|\alpha(\theta))+\ell(\theta,\alpha(\theta))$ and since $H(\alpha|\alpha(\theta))$ is positive it suffices to show that $\ell(\theta,\alpha(\theta))\to\infty$. Let now $\mathcal{A}$ be the set of indices $k=1,\ldots K$ such that $\lVert \theta^n_k\rVert \to\infty$. We next show that if $k\notin \mathcal{A}$, $\alpha_k(\theta^n)\to 0$. Indeed, the the optimality condition \eqref{eq:respsinkhorn} implies
\begin{eqnarray*}
\sum_{k\in \mathcal{A}} \alpha_k &=& \mathbb{E}\left(\sum_{k\in \mathcal{A}} \Psi_k(Y,\theta^n,\alpha(\theta^n))\right)\\
&=& \mathbb{E}\left(\frac{\sum_{k\in \mathcal{A}}\alpha_k(\theta^n)q_{\theta^n_k}(Y)}{\sum_{k\in \mathcal{A}^c}\alpha_k(\theta^n)q_{\theta^n_k}(Y)+\sum_{k\in \mathcal{A}}\alpha_k(\theta^n)q_{\theta^n_k}(Y)}\right)\\
\end{eqnarray*}
Note that the argument of the expectation is bounded by $1$, and it is positive. Therefore, we can apply the reverse Fatou's lemma \citep{shiryaev2016probability}
\begin{eqnarray*}
\sum_{k\in \mathcal{A}} \alpha_k &=& \limsup_{n\to\infty} \mathbb{E}\left(\sum_{k\in \mathcal{A}} \Psi_k(Y,\theta^n,\alpha(\theta^n))\right)\\
&\leq & \mathbb{E}\left(\limsup_{n\to\infty} \frac{\sum_{k\in \mathcal{A}}\alpha_k(\theta^n)q_{\theta^n_k}(Y)}{\sum_{k\in \mathcal{A}}\alpha_k(\theta^n)q_{\theta^n_k}(Y)+\sum_{k\in \mathcal{A}^c}\alpha_k(\theta^n)q_{\theta^n_k}(Y)}\right)\\
&\leq & \mathbb{E}\left(\frac{\sum_{k\in \mathcal{A}}\limsup_{n\to\infty} \left(\alpha_k(\theta^n)q_{\theta^n_k}(Y)\right)}{\sum_{k\in \mathcal{A}}\limsup_{n\to\infty} \left(\alpha_k(\theta^n)q_{\theta^n_k}(Y)\right)+\sum_{k\in \mathcal{A}^c}\limsup_{n\to\infty} \alpha_k(\theta^n)\limsup_{n\to\infty} q_{\theta^n_k}(Y)}\right)\\
\end{eqnarray*}
Let's look at the two terms separately. Since for $k\in\mathcal{A}$ $\theta_k^n$ is diverging, so $q_{\theta^n_k}(y)\to 0$. Conversely, since $\theta_k^n$ remains bounded for $k\in\mathcal{A}^c$, then by continuity of $q_{\theta_k}$ and assumption \ref{C3}, $\limsup_{n\to\infty} q_{\theta^n_k}(y)>\tilde{q}_k(y)>0$ for any $y\in\mathbb{R}^d$ and some measurable $\tilde{q}_k(y)$. Therefore, if there is any $k'\in\mathcal{A}^c$ such that $\limsup_{n\to\infty} \alpha_{k'}(\theta^n)>0$ we would conclude
\begin{eqnarray*}
\sum_{k\in \mathcal{A}} \alpha_k 
&\leq & \mathbb{E}\left(\frac{\sum_{k\in \mathcal{A}}\limsup_{n\to\infty} \left(\alpha_k(\theta^n)q_{\theta^n_k}(Y)\right)}{\sum_{k\in \mathcal{A}}\limsup_{n\to\infty} \left(\alpha_k(\theta^n)q_{\theta^n_k}(Y)\right)+\limsup_{n\to\infty} \alpha_{k'}(\theta^n)\limsup_{n\to\infty} \tilde{q}_{k'}(Y)}\right)\\
&\leq &\mathbb{E}\left(\frac{0}{0+\limsup_{n\to\infty} \alpha_{k'}(\theta^n)\limsup_{n\to\infty} \tilde{q}_{k'}(Y)}\right)\\ 
&\leq & 0,
\end{eqnarray*}
contradicting that each mixture weight is positive. Therefore, $\alpha_k(\theta^n)$ must vanish in the set $\mathcal{A}^c$. 
Let's now express the negative log-likelihood as
\begin{eqnarray*}\ell(\theta^n,\alpha(\theta^n))=-\mathbb{E}\left(\log \left(\sum_{k\in \mathcal{A}} \alpha_k(\theta^n)q_{\theta^n_k}(Y)+\sum_{k\in \mathcal{A}^c} \alpha_k(\theta^n)q_{\theta^n_k}(Y)\right)\right).\end{eqnarray*}
The second summation converges to $0$ by the above argument, and since $q_{\theta_k}(y)$ is bounded above ($\theta_k$ lies on a compact). The first term also converges to zero: indeed, for $k\in\mathcal{A}$, $0<\alpha_k(\theta^n)< 1$, and $q_{\theta_k}^n(Y)\to 0$ by coercivity of the mixture components \ref{C5}. If we were able to argue that we can exchange limit and expectation, we would conclude that
$$\lim_{n\to\infty}\ell(\theta^n,\alpha(\theta^n))=-\mathbb{E}\left(\lim_{n\to\infty}\log \left(\sum_{k\in \mathcal{A}} \alpha_k(\theta^n)q_{\theta^n_k}(Y)+\sum_{k\in \mathcal{A}^c} \alpha_k(\theta^n)q_{\theta^n_k}(Y)\right)\right) = -\mathbb{E}\left(\log(0)\right)=+\infty,$$
and the proof would be concluded. To justify that this exchange is valid, we rely on the envelope condition \ref{C6}, which implies
$$\log \left(\sum_{k\in \mathcal{A}} \alpha_k(\theta^n)q_{\theta^n_k}(y)+\sum_{k\in \mathcal{A}^c} \alpha_k(\theta^n)q_{\theta^n_k}(y)\right)\leq E(y),$$
where the right-hand side satisfies $\int |E(y)|dQ(y)<\infty$. The conclusion follows from dominated convergence.

\section{Proof of results in Section \ref{sec:em}}
\label{app:em}
\subsection{Proof of Theorem \ref{teo:em}}
The proof borrows from the one for the EM algorithm introduced in \cite{Wu1983}. For an arbitrary coupling $\pi$ between the set measure $\alpha_k$ in the set $[K]$ and the marginal distribution of $Y\sim Q_\ast$ we define the function
$$m(\theta,\pi)=-\mathbb{E}\left(\sum_{k=1}^K\pi(k|Y) \log \left(\alpha_k q_{\theta_k}(Y)\right)\right).$$

By the definition of the $M$ step \eqref{eq:mstep} after having computed $\pi^{t+1}$ from the previous $E$-step \eqref{eq:estep} we have, since $\theta^{t+1}$ minimizes the function $m(\cdot, \pi^{t+1})$:
$$m(\theta^{t+1},\pi^{t+1})\leq  m(\theta^{t},\pi^{t+1})$$
Here we are implicitly using the fact that since $\theta^t$ is fixed, there is one-to-one correspondence between a coupling $\tilde{\pi}\in \Pi(P_{\theta^t},Q_\ast)$ (solving the $E$ step) with a coupling $\pi\in \Pi(\alpha,Q_\ast)$. We can add the relative entropy term $H\left(\pi^{t+1}\lvert \alpha\otimes Q\right)$ to obtain
\begin{align*}m(\theta^{t+1},\pi^{t+1})+H\left(\pi^{t+1}\lvert \alpha\otimes Q_\ast\right)&\leq  m(\theta^{t},\pi^{t+1})+H\left(\pi^{t+1}\lvert \alpha\otimes Q_\ast\right)\\
&= -\mathbb{E}\left(\sum_{k=1}^K\pi^{t+1}(k|Y) \log \left(\alpha_k q_{\theta^t_k}(Y)\right)\right)+H\left(\pi^{t+1}\lvert \alpha\otimes Q_\ast\right)\\
&=  H\left(\pi^{t+1}\lvert \alpha \otimes Q_\ast\right)\\ 
&= \inf_{\tilde{\pi}\in\Pi\left(P_{\theta^t},Q_\ast\right)} H\left(\tilde{\pi}\lvert Q_{\theta^t}^{X,Y}\right)\\
& = \Ll(\theta^t).\end{align*}
In the second-to-last equality, we used that $\pi^{t+1}$ solves the previous E step, and in the last, we used the definition of $\Ll(\theta^t)$. To conclude, we note that the left-hand side above also expresses as $H\left(\tilde{\pi}^{t+1}\lvert Q^{X,Y}_{\theta^{t+1}}\right)$ but the coupling $\pi^{t+1}$ may not be optimal for the problem defining $\Ll(\theta^{t+1})$, so 
$$ \Ll(\theta^{t+1})\leq H\left(\tilde{\pi}^{t+1}\lvert Q^{X,Y}_{\theta^{t+1}}\right) \leq m(\theta^{t+1},\pi^{t+1})+H\left(\pi^{t+1}\lvert \alpha\otimes Q_\ast\right)\leq \Ll(\theta^t).$$

We now show the inequality is strict if $\theta^t$ is not a stationary point. Note that $\Ll(\theta)=m(\theta,\pi_{\theta})$ where $\pi_\theta$ minimizes $m(\theta,\cdot)$. By the chain rule and optimality of $\pi_\theta$, and using the already established differentiability of $\pi_\theta$ (Lemma \ref{lemma:alphader}) we have
\begin{equation}\label{eq:sta}\nabla_\theta \Ll(\theta)=\nabla_\theta  m (\theta,\pi_\theta) + \nabla_\pi m(\theta,\pi_\theta) \nabla_\theta  \pi_\theta=\nabla_\theta m (\theta,\pi_\theta).\end{equation}

Since $\theta^t$ is not a stationary point for $\Ll$ the above implies that $\nabla_\theta  m(\theta^t,\pi_{\theta^t})\neq 0$. Therefore, $\theta^t$ cannot globally minimize the (negative of) the function defining the M-step $m(\cdot,\pi_{\theta^t})=m(\cdot,\pi^{t+1})$, implying that the $M$ step leads to a strict decrease of this function, i.e.
$m(\theta^{t+1},\pi^{t+1}) <  m(\theta^{t},\pi^{t+1})$. By the same argument as above this implies $\Ll(\theta^{t+1})<\Ll(\theta^{t})$.

It only remains to show that the sequence $\theta^t$ converges to a stationary point. By inspecting the proof of \cite{Wu1983}, it suffices to show that $\Ll(\theta)$ satisfies a set of conditions. Most of these conditions are required to ensure non-degeneracy and are trivially satisfied in our case; the only condition that requires some care is showing that $\Ll$ is coercive, which we already established in Proposition \ref{prop:L}(f). The proof is concluded.
\subsection{Proof of Proposition \ref{prop:convalpha}}
\begin{proof}

Let's compute the first and second derivatives of $\Ll(\theta,\alpha)$ with respect to $\alpha$. We proceed similarly as in the proof of Proposition \ref{prop:L}(e), by using the fact that $\Ll(\theta,\alpha)=\tilde{\Ll}(\theta,\alpha,\omega(\theta,\alpha))$ (here, we express $\alpha$ as an additional parameter to distinguish it from non-weight parameters $\theta$ such as location or scales.

We already showed that 
$\nabla_\alpha \Ll(\theta,\alpha)= \nabla_\alpha \tilde{\Ll}\left(\theta,\alpha,\omega(\theta,\alpha)\right)$ (see the proof of Proposition \ref{prop:L}(d). Also,
it is clear to see that from \eqref{eq:tildeL} that

\begin{equation} \label{eq:nablaLalpha} \nabla_\alpha \tilde{\Ll}\left(\theta,\alpha,\omega\right) = \omega - \tilde{D}^{-1}(\alpha)\mathbb{E}\left(\Psi(Y,\theta,\alpha,\omega)\right),\end{equation}
where $\tilde{D}(\alpha)$ is the diagonal matrix with entries $\tilde{D}(\alpha)(k,k)=\alpha_k$. In particular, if $\omega=\omega(\theta,\alpha)$ we have by \eqref{eq:respsinkhorn} that $\mathbb{E}\left(\Psi(Y,\theta,\alpha,\omega(\theta,\alpha)\right)=\alpha$ and so 
$$\nabla_\alpha \Ll(\theta,\alpha) = \omega(\theta,\alpha) -\tilde{D}^{-1}(\alpha) \alpha= \omega(\theta,\alpha)-1_K,$$
We must now evaluate the second derivatives. We will use the formula
\begin{equation}\label{eq:nabla2alpha} \nabla^2_{\alpha\alpha} \Ll(\theta,\alpha) = \nabla^2_{\alpha\alpha} \tilde{\Ll}(\theta,\alpha,\omega(\theta,\alpha))-\nabla^2_{\alpha\omega}\tilde{\Ll}(\theta,\alpha,\omega(\theta,\alpha))(\nabla^2_{\omega\omega}\tilde{\Ll}(\theta,\alpha,\omega(\theta,\alpha)))^{\dagger}\nabla^2_{\alpha\omega}\tilde{\Ll}(\theta,\alpha,\omega(\theta,\alpha)).\end{equation}
Where $A^\dagger$ represents the Moore-Penrose pseudo-inverse of $A$. Note that a similar expression was already established in the Proof of Proposition \ref{prop:L}(e), equation \eqref{eq:sec3}. The only difference is that here we use the pseudo-inverse as an alternative to having to deal with a degenerate Hessian $\nabla^2_{\omega\omega}\tilde{\Ll}$ due to multiple $\omega$ achieving the maximum in the definition of $\omega(\theta)$. This is a standard equivalence (see, for example, \cite{golub1973differentiation}), so we skip the details. 

Suppose that we are able to show that $\nabla^2_{\alpha \alpha}\tilde{\Ll}(\theta,\alpha,\omega(\theta,\alpha))$ is positive definite. Then, from the fact that $\nabla^2_{\omega\omega} \tilde{\Ll}(\theta,\alpha,\omega)$ is negative definite (Lemma \ref{lemma:hessian}) in addition to \eqref{eq:nabla2alpha} we would conclude that $\Ll(\theta,\cdot)$ is convex. Let's show that this is the case: by differentiating \eqref{eq:nablaLalpha} with respect to $\alpha$, we obtain
$$\nabla^2_{\alpha \alpha}\tilde{\Ll}(\theta,\alpha,\omega) = \tilde{D}^{-1}(\alpha) \mathbb{E}\left(\Psi(Y,\theta,\alpha,\omega),\Psi(Y,\theta,\alpha,\omega)^\top\right) \tilde{D}^{-1}(\alpha).$$
This term is clearly positive definite, since for $x\in\mathbb{R}^d\setminus\{0\}$, $$x^\top \nabla^2_{\alpha \alpha}\tilde{\Ll}(\theta,\alpha,\omega)x = \mathbb{E}\left(\lVert \Psi^\top(Y,\theta,\alpha,\omega)\tilde{D}^{-1}(\alpha)x \rVert^2 \right)>0,$$
as all diagonal entries in $\tilde{D}(\alpha)$ are strictly positive, and since $\Psi$ is also strictly positive.

To establish the convergence of the iterates \eqref{eq:updatealpha} we use a standard result for Bregman divergence-based proximal methods: specifically, by Theorem 4.1 in \cite{bolte2018first} specializing to the entropy function $h(x)=x\log x$, it suffices to show that for some  $L>0$, $\nabla^2_{\alpha,\alpha} \Ll(\theta,\alpha)-\nabla^2_{\alpha\alpha} h(\alpha)=L \nabla^2_{\alpha,\alpha} \Ll(\theta,\alpha)-\tilde{D}^{-1}(\alpha)$ is negative semi definite. We will show this to be the case under assumption \eqref{eq:condkappa}. 
Indeed, differentiating \eqref{eq:nablaLalpha} with respect to $\omega$ we obtain
$$\nabla^2_{\alpha\omega} \tilde{\Ll}(\theta,\alpha,\omega)= I_K - \tilde{D}^{-1}(\alpha)\left(D(\theta,\alpha,\omega)-\mathbb{E}\left(\Psi(Y,\theta,\alpha,\omega)\Psi(Y,\theta,\alpha,\omega)^\top\right)\right),$$
where $D(\theta,\alpha,\omega)$ is the diagonal matrix with entries $D_{k,k}(\theta,\alpha,\omega)=\mathbb{E}\left(\Psi_k(Y,\theta,\alpha,\omega)\right)$.  
Therefore, using the fact that at $\omega=\omega(\theta,\alpha)$ we have $D(\omega,\alpha,\omega(\theta,\alpha))=\tilde{D}(\alpha)$ we get  $$\nabla^2_{\alpha\omega} \tilde{\Ll}(\theta,\alpha,\omega(\theta,\alpha))=\tilde{D}^{-1}(\alpha)\mathbb{E}\left(\Psi(Y,\theta,\alpha,\omega(\theta,\alpha))\Psi(Y,\theta,\alpha,\omega(\theta,\alpha))^\top\right)$$
Moreover, we already showed in the proof of Lemma \ref{lemma:hessian}
that \begin{eqnarray*}
\nabla^2_{\omega\omega} \tilde{\Ll}(\theta,\alpha,\omega(\theta,\alpha))&=&-\left(D(\theta,\alpha,\omega(\theta,\alpha))-\mathbb{E}\left(\Psi(Y,\theta,\alpha,\omega(\theta,\alpha))\Psi(Y,\theta,\alpha,\omega(\theta,\alpha))^\top\right)\right)\\&=&
-\left(\tilde{D}(\alpha)-\mathbb{E}\left(\Psi(Y,\theta,\alpha,\omega(\theta,\alpha))\Psi(Y,\theta,\alpha,\omega(\theta,\alpha))^\top\right)\right).
\end{eqnarray*}
Consistent with our previous notation, we write $\Psi(Y,\theta,\alpha,\omega(\theta,\alpha))$ simply as $\Psi(Y,\theta,\alpha(\theta)))$. Putting all together, and denoting $W(\theta,\alpha)=\mathbb{E}\left(\Psi(Y,\theta,\alpha(\theta))\Psi(Y,\theta,\alpha(\theta))^\top\right)$ we have established that
\begin{eqnarray}\nonumber \nabla^2_{\alpha\alpha} \Ll(\theta,\alpha) &=& \tilde{D}^{-1}(\alpha) W(\theta,\alpha) \tilde{D}^{-1}(\alpha) +\tilde{D}^{-1}(\alpha)  W(\theta,\alpha)\left(\tilde{D}(\alpha)-W(\theta,\alpha)\right)^\dagger W(\theta,\alpha) \tilde{D}^{-1}(\alpha)\\ && \label{eq:nabla2alpha}
\end{eqnarray}
It then remains to bound each of the terms above by $\kappa D^{-1}(\alpha)$ in the Loewner order. First, following our previous calculation, note that for $x\in\mathbb{R}^d\setminus\{0\}$, 
\begin{eqnarray*} x^\top\tilde{D}^{-1}(\alpha) W(\theta,\alpha) \tilde{D}^{-1}(\alpha)x&=&\mathbb{E}\left(\lVert \Psi^\top(Y,\theta,\alpha(\theta))\tilde{D}^{-1}(\alpha)x \rVert^2 \right)\\
&=& \mathbb{E}\left(\left(\sum_{k=1}^K \frac{\Psi_k(Y,\theta,\alpha(\theta))x_k}{\alpha_k}\right)^2\right)\\
&=& \mathbb{E}\left(\left(\sum_{k=1}^K \frac{\sqrt{\Psi_k(Y,\theta,\alpha(\theta))}x_k}{\sqrt{\alpha_k}} \frac{\sqrt{\Psi_k(Y,\theta,\alpha(\theta))}}{\sqrt{\alpha}_k}\right)^2\right)\\
&\leq & \mathbb{E}\left(\sum_{k=1}^K \frac{\Psi_k(Y,\theta,\alpha(\theta))x^2_k} {\alpha_k}\sum_{k=1}^K\frac{\Psi_k(Y,\theta,\alpha(\theta))}{\alpha_k}\right)\\ &\leq & \sum_{k=1}^K \frac{x^2_k} {\alpha_k} \mathbb{E}\left(\sum_{k=1}^K\frac{\Psi_k(Y,\theta,\alpha(\theta))}{\alpha_k}\right)
\\
&=& \sum_{k=1}^K \frac{x^2_k} {\alpha_k} \mathbb{E}\left(\sum_{k=1}^K\frac{\alpha_k}{\alpha_k}\right)\\
&=& Kx^\top \tilde{D}^{-1}(\alpha)x.
\end{eqnarray*}
In the above chain of inequalities, we used Cauchy-Schwarz, the fact that $\Psi(Y,\theta,\omega)\leq 1$, and that $\mathbb{E}\left(\Psi_k(Y,\theta,\alpha(\theta)\right)=\alpha_k$. It then suffices to establish a similar bound for the second term in \eqref{eq:nabla2alpha}. To do so, Denote $\tilde{W}(\theta,\alpha)= \tilde{D}^{-1/2}(\alpha)W(\theta,\alpha)\tilde{D}^{-1/2}$. Then, 

$$\tilde{D}^{-1}(\alpha)  W(\theta,\alpha)\left(\tilde{D}(\alpha)-W(\theta,\alpha)\right)^\dagger W(\theta,\alpha) \tilde{D}^{-1}(\alpha)= \tilde{D}^{-1/2} \tilde{W}(\theta,\alpha) \left(I_K-\tilde{W}(\theta,\alpha)\right)^\dagger \tilde{W}(\theta,\alpha)\tilde{D}^{-1/2}(\alpha).$$

By a simple diagonalization argument, eigenvalues of $\tilde{W}(\theta,\alpha)\left(I_K-\tilde{W}(\theta,\alpha)\right)^\dagger \tilde{W}(\theta,\alpha)$  $\lambda_k(\alpha)^2/(1-\lambda_k(\alpha))$ if $\lambda_k(\alpha)\neq 1$ and $0$ if $\lambda_k(\alpha)=1$, where $\lambda_k(\alpha)$ are the eigenvalues of $\tilde{W}(\theta,\alpha)$ (not in order). Then, for $x\in\mathbb{R}^d\setminus\{0\}$,
\begin{eqnarray*} x^\top \tilde{D}^{-1}(\alpha)  W(\theta,\alpha)\left(\tilde{D}(\alpha)-W(\theta,\alpha)\right)^\dagger W(\theta,\alpha) \tilde{D}^{-1}(\alpha)x &\leq& \max_{k: \lambda_k(\alpha)\neq 1} \left\{\frac{\lambda_k(\alpha)^2}{1-\lambda_k(\alpha)}\right\} x^\top\tilde{D}^{-1/2}(\alpha)\tilde{D}^{-1/2}(\alpha)x \\ &\leq&  \kappa  x^\top\tilde{D}^{-1}(\alpha)x,
\end{eqnarray*}
where $\kappa$ is independent of $\alpha$, by virtue of \eqref{eq:condkappa}. Above, we have used that all eigenvalues of $\tilde{W}(\theta,\alpha)$ are upper bounded by one, and that since the graph induced by the weight matrix $W(\theta,\alpha)$ is connected (all entries are positive), then the second largest eigenvalue must be smaller than one.

We have established the convergence of iterates \eqref{eq:updatealpha} if $\eta$ is small enough (this value should be chosen as a function of the above convexity bound). It only remains to show that the block-coordinate descent described in Proposition \ref{prop:convalpha} converges to a stationary point. This follows directly from Theorem 2 in \cite{razaviyayn2013unified}. Indeed, it suffices to note that the Sinkhorn-EM step maximizes an upper bound $u_1(\cdot,\theta^t,\alpha)$ to $\Ll(\cdot,\alpha)$ that satisfies the conditions in that Theorem. These conditions are 1) $u_1(\cdot,\theta^t,\alpha)$ is convex, 2) $u_1(\theta^t,\theta^t,\alpha)=\Ll(\theta^t,\alpha)$ and 3) $\nabla_\theta u_1(\theta^t,\theta^t,\alpha)=\nabla_\theta \Ll(\theta,\alpha)$.  In our case, this function is $m(\theta,\pi(\theta^t)$ where $m$ is defined in the proof of Theorem \ref{teo:em}. This function is clearly quasi convex over $\theta$ under our concavity assumption \ref{C7}, and it also clearly satisfies conditions 2) and 3) by construction. Regarding an upper bound for $\Ll(\theta,\cdot)$, we can use the function itself, which we indeed minimize using updates \eqref{eq:updatealpha}.
\end{proof}

\section{Supplementary material for Section \ref{sec:local}}\label{sec:suplocal}
\subsection{Supplementary discussion}

The existence of local-minima structures for $\ell$ has been recently studied in detail in \cite{chen2020likelihood}, showing that local minima of the negative log-likelihood (with equal weights and fixed variance). They show that if all components are all well separated from each other (i.e. if $\Delta_{min}:=\min_{k_1,k_2\leq n} \norm{\theta_{k_1}^\ast-\theta_{k_2}^\ast}$ is large enough) then the components of any spurious local minima $\theta$ of $\ell$ partition into groups forming either a \textit{many-fit-one} configuration or a \textit{one-fits-many} configuration, where a $\theta_k$ is placed near to the \textit{average} of a group of true $\theta^\ast_k$. 

Although appealing, the result of \cite{chen2020likelihood} suffers from three drawbacks: first, the required lower bound on $\Delta_{min}$ is too stringent to provide any insight in practical setups ($\Delta_{min}\geq 18(\sqrt{2\pi}+1)K^5\sigma$). Moreover, it provides only necessary conditions for local optimality; it does not indicate which such configurations are ultimately realized as local minima. Third, other components of any such local minimum may satisfy a degenerate so-called ``near empty association'' condition distinct from the above \textit{many-fit-one} and \textit{one-fit-many}. Even if a complete characterization of local minima of $\ell$ is far from complete, in the main text we take the above characterization as a starting point to frame our Theorem \ref{theo:local}.

\subsection{On local minima of $\Ll$ and $\ell$}

From Proposition \ref{prop:L}(c), every stationary point of $\Ll$ is a stationary point for the log-likelihood on a GMM for some \textit{tilted} weights $\alpha(\theta)$. Then, we have the following corollary. 
\begin{corollary}\label{cor:local}
    If $\theta$ is a stationary point for $\Ll$, then it is a stationary point for $\ell$ on the model \eqref{eq:gmm} with weights $\alpha(\theta)$. If, under these weights, $\theta$ is a local minimum for $\ell$, then it also must be a local minimum for $\Ll$.
\end{corollary}

\begin{proof}
Let $\theta$ be a stationary point of $\Ll$. By Proposition \ref{prop:L}(d) it satisfies $\nabla_\theta \ell(\theta,\alpha(\theta))=0$. In the GMM of equation \eqref{eq:gmm}, this means that $\theta$ is a stationary point on a GMM with tilted weights $\alpha(\theta)$ instead of $\alpha$. This statement is equivalent to saying $$\nabla_\theta \tilde{\Ll}(\theta,\omega(\theta))=\nabla_\theta  \ell(\theta,\alpha(\theta)).$$ Likewise, by taking derivatives in the definition of $\tilde{\Ll}$ we can verify that
$$\nabla^2_{\theta\theta} \tilde{\Ll}(\theta,\omega(\theta))=\nabla^2_{\theta\theta} \ell(\theta,\alpha(\theta)),$$
and by \eqref{eq:sec3} this means that
$$\nabla^2_{\theta\theta} \Ll(\theta) =\nabla^2_{\theta\theta} \ell(\theta,\alpha(\theta))  - \nabla^2_{\theta\omega}  \tilde{\Ll}(\theta, \omega(\theta))\left(\nabla^2_{\omega\omega}  \tilde{\Ll}  (\theta, \omega(\theta))
\right)^{-1} \nabla^2_{\omega\theta}  \tilde{\Ll}(\theta, \omega(\theta)).$$
So, if $\theta$ is a local minimum for $\ell$ then $\nabla^2_{\theta\theta} \ell(\theta,\alpha(\theta))$ is positive definite. By the same arguments as in the proof of Proposition \eqref{prop:L}(d), the second term on the right-hand side above is positive semidefinite. Therefore, in that case, $\nabla^2_{\theta\theta} \Ll(\theta)$ is positive definite, i.e., $\theta$ is a local minimum for $\Ll$.  
\end{proof}
\subsection{Proof of Theorem \ref{theo:local}}

Let us write $\Psi_k(Y,\theta,\alpha(\theta))$  \eqref{eq:respsinkhorn} in this setup: $$\Psi_k(Y,\theta,\alpha(\theta))=\frac{\alpha_k(\theta)e^{-\frac{1}{2\sigma^2}||Y-\theta_k||^2}}{\sum_{k=1}^K\alpha_k(\theta)e^{-\frac{1}{2\sigma^2}||Y-\theta_k||^2}},$$
where $\alpha(\theta)$ is the vector of weights arising in the semi-dual optimal transport formulation. By definition of $\alpha(\theta)$, $\mathbb{E}(\Psi_k(Y,\theta,\alpha(\theta)))=\frac{1}{K}$. The first-order optimality conditions for $\theta$ read 
$$\mathbb{E}((Y-\theta_k)\Psi_k(Y,\theta,\alpha(\theta))=0,$$
or equivalently, \begin{equation}\label{eq:first} \mathbb{E}(Y\Psi_k(Y,\theta,\alpha(\theta)))=\frac{1}{K}\theta_k.\end{equation} We now show that this implies that the sum of the true weights must equal the sum of the model weights:
\begin{equation}
\label{eq:balance}
     \sum_{k=1}^K {\theta_k^\ast} = \sum_{k=1}^K {\theta_k}.
\end{equation}
Indeed, since $\sum_{k=1}^K \Psi_k(Y,\theta,\alpha(\theta))=1$, by adding over all components in \eqref{eq:first} we obtain
$$
    \frac{1}{K}\sum_{k=1}^K {\theta_k^\ast} =\mathbb{E}(Y)=\mathbb{E}\left(Y\left( \sum_{k=1}^K \Psi_k(Y,\theta,\alpha(\theta))\right)\right)= \frac{1}{K}\sum_{k=1}^K {\theta_k}.
$$
First-order optimality also implies a lower bound for averages of groups of $\theta_k$. Indeed, let $S\subseteq[K]$ be an arbitrary set of indices. Equation \eqref{eq:first} implies that
\begin{equation}\label{eqfirstgroup} 
    \frac{1}{K}\sum_{k\in S} {\theta_k} = \mathbb{E}\left(Y\left( \sum_{k\in S} \Psi_k(Y,\theta,\alpha(\theta))\right)\right).
\end{equation}

We will prove an equivalent version of Theorem \ref{theo:local}.As in the original statement, we define $G_1\subseteq K$ as the set of $k_1$ smallest coordinates and $G_2$ as the remaining $k_2=K-k_1$, and define the separation
$$\Delta = \inf_{g_1\in G_1,g_2\in G_2}|\theta^\ast_{g_1}-\theta^\ast_{g_2}|=\theta^\ast_{(k_1+1)}-\theta^\ast_{(k_1)}>0.$$

We will show that if $I$ is a set of indices with $I>k_2$, and such that $I(\theta)$ that covers $G_2(\theta^\ast)$ for some $\delta>0$. Then, for any $0<\gamma<1-k_2/|I|$, if 
 \begin{equation}\Delta \geq \frac{K\sigma}{\sqrt{2\pi}\left((1-\alpha)|I|-k_2\right)},\end{equation}

then $\theta$ cannot be a stationary point of $\Ll$ if $\delta<\gamma \Delta$.

The right-hand side of \eqref{eqfirstgroup} can be written as $\mathbb{E}(Yf(Y))$ for some $f(\cdot)\in[0,1]$. Although the above relations are multi-dimensional, hereafter, we only need to look at the first coordinate, which w.l.o.g. is our direction of interest, since the stationary points of $\Ll_R$ on a rotated GMM are the rotations of the stationary points of $\Ll$ for the original model.  Likewise, we can assume that the rightmost component of $G_1(\theta^\ast)$ equals 0 so that $G_1$ is the group of negative components, $\theta^\ast_k\leq 0$ for all $k\in G_1$, and that $\theta^\ast_k\geq 0$ for $k\in G_2$. 
If $0\leq\phi(\cdot)\leq \frac{1}{\sqrt{2\pi}}$ and $0\leq\Phi(\cdot)\leq 1$ are the standard Gaussian pdf. and cdf., respectively, and if $\mathbb{E}_k$ denotes expectation under the $k$-th component (i.e. $\mathcal{N}(\theta_k^\ast,\sigma^2)$), we have:
\begin{align} \label{eq:trunc}  \mathbb{E}(Yf(Y)) &\geq  \mathbb{E}(Y 1_{Y\leq 0}) \\ \nonumber &=\frac{1}{K}\sum_{k=1}^K \mathbb{E}_k(Y 1_{Y\leq 0}) \\ \nonumber &= \frac{1}{K}\left(\sum_{k=1}^K \theta^\ast_k\Phi\left(-\frac{\theta^\ast_k}{\sigma}\right)-\sigma\phi\left(\frac{\theta^\ast_k}{\sigma}\right)\right) \\ \nonumber &\geq 
\frac{1}{K}\left(\sum_{k=1}^K \theta^\ast_k\Phi\left(-\frac{\theta^\ast_k}{\sigma}\right)-\frac{\sigma}{\sqrt{2\pi}}\right) \\ \nonumber &\geq 
\frac{1}{K}\sum_{i\in G_1} \theta^\ast_k-\frac{\sigma}{\sqrt{2\pi}},\end{align}
where in the third line, we used simple properties of the truncated Gaussian distribution, and in the fourth, we have used that $0\leq \Phi(\cdot)\leq 1$ and that $\theta^\ast_k\geq 0$ in $G_2(\theta^\ast)$. By combining  \eqref{eqfirstgroup} and \eqref{eq:trunc}  we obtain
\begin{equation}\label{boundfirstorder}\frac{1}{K}\sum_{k\in S} {\theta_k} \geq \frac{1}{K}\sum_{k\in G_1} \theta^\ast_k-\frac{\sigma}{\sqrt{2\pi}}.\end{equation}

To conclude, we will use the above relations to obtain a lower bound on $\delta$. Take $S=I^c$. Recall that by assumption $I(\theta)$ is a $\delta$-covering of $G_2(\theta^\ast)$. We have by \eqref{eq:trunc}
\begin{equation*} \frac{1}{K}\sum_{k\in I} {\theta_k}+\frac{1}{K}\sum_{k\in I^c} {\theta_k} \geq \frac{1}{K}\sum_{k\in I} \theta_k +\frac{1}{K}\sum_{k\in G_1} \theta^\ast_k  -\frac{\sigma}{\sqrt{2\pi}}.\end{equation*}
By the above inequality, and \eqref{eq:balance}:
\begin{equation*} \frac{1}{K}\sum_{k\in G_1} {\theta^\ast_k}+\frac{1}{K}\sum_{k\in G_2} {\theta^\ast_k} \geq \frac{1}{K}\sum_{k\in I} \theta_k +\frac{1}{K}\sum_{k\in G_1} \theta^\ast_k  -\frac{\sigma}{\sqrt{2\pi}},\end{equation*}
so that 
\begin{equation*} \frac{1}{K}\sum_{k\in I} \theta_k  -\frac{\sigma}{\sqrt{2\pi}} \leq \frac{1}{K}\sum_{k\in G_2} {\theta^\ast_k}.\end{equation*}
This implies that for any sequence $\tilde{\theta}_k$, 
\begin{equation*} \frac{1}{K}\sum_{k\in I} (\tilde{\theta}^\ast_k-\theta_k)    \geq \frac{1}{K}\sum_{k\in I} \tilde{\theta}^\ast_k -\frac{1}{K}\sum_{k\in G_2} {\theta^\ast_k}-\frac{\sigma}{\sqrt{2\pi}} .\end{equation*}
Let's now use the fact that $I(\theta)$ covers $G_2(\theta^\ast)$: for each $\theta_k$ in $I(\theta)$ call $\tilde{\theta}^\ast_k$ the component in $G_2(\theta^\ast)$ that is $\delta$-close to $\theta_k$. 

\begin{equation*} \delta|I|\geq \sum_{k\in I} |\tilde{\theta}^\ast_k-\theta_k| \geq \sum_{k\in I} \tilde{\theta}^\ast_k-\theta_k   \geq \sum_{k\in I} \tilde{\theta}^\ast_k -\sum_{k\in G_2} {\theta^\ast_k} -\frac{K\sigma}{\sqrt{2\pi}}.\end{equation*}
Note that the defined $\tilde{\theta}^\ast_k$ exhaust the set $G_2(\theta^\ast)$, and must give rise to $(|I|-k_2)$ duplicate indexes that we denote $\tilde{G}_2$. With this observation, we can write (abusing notation) the difference of the sums in the right-hand side above in terms of these duplicates:
\begin{equation*} \delta\tilde{k}\geq \sum_{k\in I} (\tilde{\theta}^\ast_k-\theta_k)  \geq \sum_{k\in \tilde{G}_2} \tilde{\theta}^\ast_k -\frac{K\sigma}{\sqrt{2\pi}}.\end{equation*}
Now, since each of the duplicates is a member of $G_2(\theta^\ast)$, we have $\tilde{\theta}^\ast_k\geq \theta^\ast_l$ where $\theta^\ast_l=\theta^\ast_{(k_1+1)}$ is the smallest component in $G_2(\theta^\ast)$, which is at least at a distance $\Delta$ above $\theta^\ast_r=0=\theta^\ast_{(k_1)}$, the largest component in $G_1(\theta^\ast)$. Therefore, 
\begin{equation*} \delta|I| \geq (|I|-k_2)\Delta -\frac{K\sigma}{\sqrt{2\pi}}.\end{equation*}
Under the separation condition on $\Delta$, we have that
$$\frac{K\sigma}{\sqrt{2\pi}}<((1-\gamma)|I|-k_2)\Delta.$$
The previous two displays then imply that
$$\delta |I|\geq (|I|-k_2)\Delta -((1-\gamma)|I| -k_2)\Delta \geq |I|\gamma \Delta,$$
and so the proof is concluded. To revert the statement as it appears in the main text, it suffices to flip the signs of $\theta^\ast$ and argue by symmetry.

\subsection{Proof of Proposition \ref{prop:bad}}
The proof is essentially the same as the one for Theorem 1 in \cite{jin2016local}, but a slightly more careful argument is needed for the two-dimensional computations. The negative log-likelihood  $\ell$ here writes ($\theta_k^j$ represents the $j-$th coordinate of $\theta_k$):
$$\ell(\theta)=-\mathbb{E}\left(\log\left(\sum_{k=1}^3e^{-\frac{(x-\theta^1_k)^2}{2}}e^{-\frac{(y-\theta_k^2)^2}{2}}\right)\right)+\log\left(6\pi\right).$$
Define  $\tilde{\theta}_1=(0,0)$ and $\tilde{\theta}_2=\tilde{\theta}_3=(R,0)$. Clearly, $\theta=(\tilde{\theta}_1,\tilde{\theta}_2,\tilde{\theta}_3)$ is in the interior of $\mathcal{R}^\varepsilon$ for every $\varepsilon>0$. Let's compute the likelihood of this configuration when $R\rightarrow\infty$:
$$m_0:=\lim_{R\rightarrow \infty}\ell(\tilde{\theta})=1+\frac{D^2}{3}+\log\left(6\pi\right)-\frac{\log 2}{3}.$$

For a fixed vector of second coordinates $\bar{\theta}^2$ we consider the regions in $R^6$
\begin{align*}R_1(\bar{\theta}^2):=\{\theta: \theta^1_1=\frac{R}{3},\theta_2^1\geq\frac{2R}{3},\theta_3^1\geq \frac{2R}{3},\theta^2=\bar{\theta}^2\},\\
\mathcal{R}_2(\bar{\theta}^2):=\{\theta: \theta^1_1\leq\frac{R}{3},\theta_2^1=\frac{2R}{3},\theta_3^1\geq \frac{2R}{3},\theta^2=\bar{\theta}^2\},\\
\mathcal{R}_3(\bar{\theta}^2):=\{\theta: \theta^1_1\leq\frac{R}{3},\theta_2^1\geq\frac{2R}{3},\theta_3^1= \frac{2R}{3},\theta^2=\bar{\theta}^2\},\end{align*}

  It is easy to see that
\begin{align*}\lim_{\mathcal{R}\rightarrow \infty}\inf_{\theta\in R_1(\bar{\theta}^2)}\ell(\theta)&=\infty, \\ 
\lim_{R\rightarrow \infty}\inf_{\theta\in R_2(\bar{\theta}^2)}\ell(\theta)& =1+\frac{1}{6}\left((D-\bar{\theta}^2_1)^2+(D+\bar{\theta}^2_1)^2\right)+\frac{\bar{\theta}_3^2}{6}+\log\left(6\pi\right),\\
\lim_{R\rightarrow \infty}\inf_{\theta\in R_3(\bar{\theta}^2)}\ell(\theta)& =1+\frac{1}{6}\left((D-\bar{\theta}^2_1)^2+(D+\bar{\theta}^2_1)^2\right)+\frac{\bar{\theta}_2^2}{6}+\log\left(6\pi\right).
\end{align*}
The first line follows from Jensen's inequality and basic moment relationships for a Gaussian distribution. The second and third limits follow from the fact that asymptotically the infimum is attained at $\theta^1_1=0,\theta^1_2=2R/3,\theta^1_3=R,\theta^2=\bar{\theta}$ for $R_2(\bar{\theta}^2)$ and at $\theta^1_1=0,\theta^1_2=R,\theta^1_3=2R/3,\theta^2=\bar{\theta}$ for $R_3(\bar{\theta}^2)$. \\
\noindent Now, for $\varepsilon>0$ define $\mathcal{R}^\varepsilon_k:=\cup_{||\bar{\theta}||< \varepsilon}\mathcal{R}_k(\bar{\theta})$ for $k=1,2,3$ so that $R^\varepsilon_k$ are the three 5-dimensional faces of $\mathcal{R}^\varepsilon$.
Call $m^\varepsilon_k=\lim_{R\rightarrow \infty}\inf_{\theta\in \mathcal{R}^\varepsilon_k}\ell(\theta)$. Since the above limits are simultaneously minimized with respect to $\bar{\theta}$ if  $\bar{\theta}^2=0$ we have that $$m_1^\varepsilon=\infty, m_2^\varepsilon=m^\varepsilon_3=1+\frac{L}{3}+\log\left(6\pi\right).$$  Therefore, $m_0<\min\{m^\varepsilon_1,m^\varepsilon_2,m^\varepsilon_3\}$ so, as in \cite{jin2016local} we conclude the existence of a local minimum $\theta^\varepsilon$ in the interior of $\mathcal{R}^\varepsilon$, whenever $R=R(\varepsilon)$ is sufficiently large. By the same continuity argument as in \cite{jin2016local}, this minimum has a smaller likelihood than the global maximizer $\theta^\ast$.

\subsection{Proof of Corollary \ref{cor:statio}}
\label{app:statio}

Define new rotated axes as described in Fig. \ref{fig:bad}B, with the origin at $\theta_2^\ast$. The new coordinates are given by
$$\theta^{\ast,new}_{1,x}=-\frac{2D^2}{\sqrt{R^2+D^2}}, \theta^{\ast,new}_{2,x}=0, \theta^{\ast,new}_{3,x}=\sqrt{R^2+D^2}-\frac{2D^2}{\sqrt{R^2+D^2}}.$$
We split these components in the two leftmost $\theta^{\ast,new}_1,\theta^{\ast,new}_2$ and the rightmost $\theta^{\ast,new}_3$. The minimum separation between groups is given by $$\Delta(R)=\theta^{\ast,new}_{3,x}-\theta^{\ast,new}_{2,x}=\sqrt{R^2+D^2}-\frac{2D^2}{\sqrt{R^2+D^2}}=D\left(\sqrt{\frac{R^2}{D^2}+1}-\frac{2}{\sqrt{\frac{R^2}{D^2}+1}}\right).$$ Taking $\tilde{k}=2,k_2=1,\alpha=0.45,K=3$ and $\sigma^2=1$ we obtain, by Theorem \ref{theo:local}, that for \begin{equation}\label{eq:rule} \Delta(R)\geq \frac{3}{\sqrt{2\pi}((1-0.45)\times 2-1)}\geq \frac{30}{\sqrt{2\pi}}\geq 5\end{equation} we can rule out stationary points for $\Ll$ with $\delta\geq 0.45\Delta(R)$. Note also that since $$x-\frac{2}{x}\geq \frac{4}{5}x$$ whenever $x^2\geq 10$ we have that
$$\Delta(R)\geq \frac{4}{5}D\sqrt{\frac{R^2}{D^2}+1}$$ holds if $R^2\geq 9 D^2$. Therefore, $$\delta\geq \frac{9}{25}D\sqrt{\frac{R^2}{D^2}+1}$$ for a stationary point if $R^2\geq 9D^2$ and if \eqref{eq:rule} holds. To ensure \eqref{eq:rule} we can additionally impose that $D\geq \frac{5\sqrt{5}}{4\sqrt{2}}>2 $. 
We must now go back to the original coordinates: this $\delta$ breaks down into $\delta_x$ and $\delta_y$ in the original coordinates, with $$\delta_x=\frac{R}{\sqrt{R^2+D^2}}\delta\geq \frac{9}{25} \frac{R}{\sqrt{R^2+D^2}} D\sqrt{\frac{R^2}{D^2}+1} \geq \frac{9}{25}R> \frac{R}{3}$$ for a stationary point of $\Ll$. In contrast, Proposition \ref{prop:bad} anticipates a bad local optimum for the log-likelihood with $\delta_x \leq R/3$ if $R$ is sufficiently large.

\section{Proof of results in Section \ref{sec:twogaussians}} 
The proof of Theorem \ref{teo:mixtureloss} relies on an analysis of the functions $\Ll,\ell$ and their derivatives. In this Appendix we will write $\Ll_{\alpha^\ast}$ $\ell_{\alpha^\ast}$ to emphasize the dependence on the true mixture weights $\alpha^\ast$. Additionally, we write $q_{\theta,\alpha}(y)$ to denote the density of the distribution of a mixture of two unit-variance Gaussians with centers $(\theta,-\theta)$ and weights $(\alpha,1-\alpha)$, i.e.,
$$q_{\theta,\alpha}(y)= \alpha \mathcal N(y; \theta, 1) + (1- \alpha) \mathcal N(y; -\theta, 1).$$

Fig.~\ref{fig:1} depicts the main properties of the functions that will be used in the proofs.
The first row shows $\Ll_{\alpha^*}(\theta)\geq \ell_{\alpha^*}(\theta)$, which follows from Proposition \ref{prop:L}(b).
The second through fourth rows illustrate the behavior of the derivatives $\Ll'$ and $\ell'$.
We show in Proposition \ref{prop:deri} that $\Ll'_{\alpha^*}(\theta)\geq \ell_{\alpha^*}'(\theta)$ for all $\theta < 0$, which is clearly visible in the second and fourth row. In the third row, we plot the absolute values of the derivatives, with stationary points visible as cusps.
In the last row, we plot an important auxiliary function described in more detail below.

As mentioned in the main text, we assume $\alpha^*\geq 0.5$ by a simple symmetry argument. The fourth column in Fig.~\ref{fig:1} illustrates this symmetry. 

We make several additional definitions for the proof of Theorem~\ref{teo:mixture}.
Let $\omega=\omega(\theta)$ be the semi-dual weight defined in \eqref{eq:semidual}. The first-order optimality conditions for $\omega$ read
\begin{equation}\label{eq:optimalpha} \alpha^*=\int \frac{e^{\omega_1}\alpha^*e^{-(\theta-y)^2/2}}{e^{\omega_1}\alpha^*e^{-(\theta-y)^2/2}+e^{\omega_2}(1-\alpha^*)e^{-(\theta+y)^2/2}}q_{\theta^*,\alpha^\ast}(y)\dd y.\end{equation}
The above condition can be expressed in terms of the \textit{tilted} $\alpha(\theta)$ introduced in \eqref{eq:alphatheta}: $\alpha(\theta)$ the unique number in $[0,1]$ satisfying
\begin{equation}\label{eq:optimalpha2} \alpha^*=G(\theta,\alpha(\theta)),\end{equation} 
with $G(\theta,\alpha)$ defined as
\begin{align}\label{eq:optimalpha3} G(\theta,\alpha)&:=\int \frac{\alpha e^{-(\theta-y)^2/2}}{\alpha e^{-(\theta-y)^2/2}+(1-\alpha)e^{-(\theta+y)^2/2}}q_{\theta^*,\alpha^\ast}(y)\dd y \\
&= \int \frac{\alpha e^{\theta y}}{\alpha e^{\theta y }+(1-\alpha)e^{-\theta y}}q_{\theta^*,\alpha^\ast}(y)\dd y\\&=
\mathbb{E}\left(\Psi_1(Y,\theta,\alpha)\right). \end{align}

We plot the tilting $\alpha(\theta)$ in the last row of Fig.~\ref{fig:1}.

To analyze the behavior of Sinkhorn-EM and vanilla EM, we also introduce the auxiliary function $F(\theta, \alpha)$ defined by
\begin{equation}\label{eq:def_f}
F(\theta, \alpha) := \int y \frac{\alpha e^{\theta y} - (1-\alpha) e^{-\theta y}}{\alpha e^{\theta y} + (1-\alpha) e^{-\theta y}} q_{\theta^*,\alpha^\ast}(y) \dd y\,.
\end{equation}

With this notation, the updates of SEM satisfy
\begin{equation*}
\theta^{t+1}_{SEM} = F(\theta^t_{SEM}, \alpha(\theta^t_{SEM}))\,,
\end{equation*}
where $\alpha(\theta)$ is defined in~\eqref{eq:optimalpha2}. 
On the other hand, the updates of EM satisfy
\begin{equation*}
\theta^{t+1}_{EM} = F(\theta^t_{EM}, \alpha^*).
\end{equation*}

\subsection{Proof of Theorem \ref{teo:mixtureloss}}
Note first that we can assume $\alpha^* > 0.5$. The case $\alpha^*=0.5$ is trivial since the entropic OT loss coincides with the negative log-likelihood (last column of Fig.~\ref{fig:1}) and SEM and EM define the same algorithm. Indeed, in this case, by Lemma \ref{lemma:alpha12}, $\alpha(\theta)=1/2$ so that $F(\theta,\alpha(\theta))=F(\theta,1/2)$.

First, we show that $\Ll_{\alpha^\ast}$ never has spurious stationary points on $(0, \infty)$. Suppose there was such $\theta'>0$. Then, the SEM algorithm initialized at that value would remain there, by virtue of Theorem \ref{teo:em}. Since Theorem~\ref{teo:mixture} guarantees that SEM converges to $\theta^*$ for any positive initialization, this implies $\theta'=\theta^\ast$.


We now show that if $\Ll_{\alpha^*}$ has a spurious stationary point, then so does $\ell_{\alpha^*}$.
Suppose that $\Ll_{\alpha^*}$ has a stationary point $\theta \in (-\infty, 0]$.
In Proposition \ref{prop:deri} we show that if $\theta \leq 0$, then $\Ll'_{\alpha^*}(\theta)>\ell'_{\alpha^*}(\theta)$.
Therefore, if $\theta$ is stationary point of $\Ll_{\alpha^*}$, then $\ell'_{\alpha^*}(\theta) < 0$.
Since $\ell_{\alpha^*}$ is continuously differentiable and $\ell'_{\alpha^*}(0) = (2{\alpha^*}^2 - 1)^2 > 0$, there must be a $\theta' \in (\theta, 0)$ such that $\ell'_{\alpha^*}(\theta')  = 0$.
Therefore, $\ell_{\alpha^*}$ also has a spurious stationary point.

Finally, to show that the set of $\alpha^*$ for which $\ell_{\alpha^*}$ has a spurious stationary point is strictly larger than the corresponding set for $\Ll_{\alpha^*}$, we note that the arguments in the proof of Theorem 1 and Lemma 4 in~\citep{Xu2018} establish that there is $\delta > 0$ such that if $\alpha^* = 0.5 + \delta$ then $\ell_{\alpha^*}$ has a single spurious stationary point on $(-\infty, 0)$, and if $\alpha^* > 0.5 + \delta$, then $\ell_{\alpha^*}$ does not have any spurious stationary points. Sine $\ell'_{\alpha^*}(\theta)$ is a continuous function of $\alpha^*$, this implies that $\ell'_{0.5 + \delta}$ is nonnegative for all $\theta < 0$.
Since $\Ll'_{0.5 + \delta}(\theta)>\ell'_{0.5 + \delta}(\theta)$ for all $\theta < 0$, we obtain that $\Ll'_{0.5 + \delta}$ has no spurious stationary points.

\begin{figure}[H]
\hspace{-1.5cm}
\includegraphics[width=1.2\textwidth]{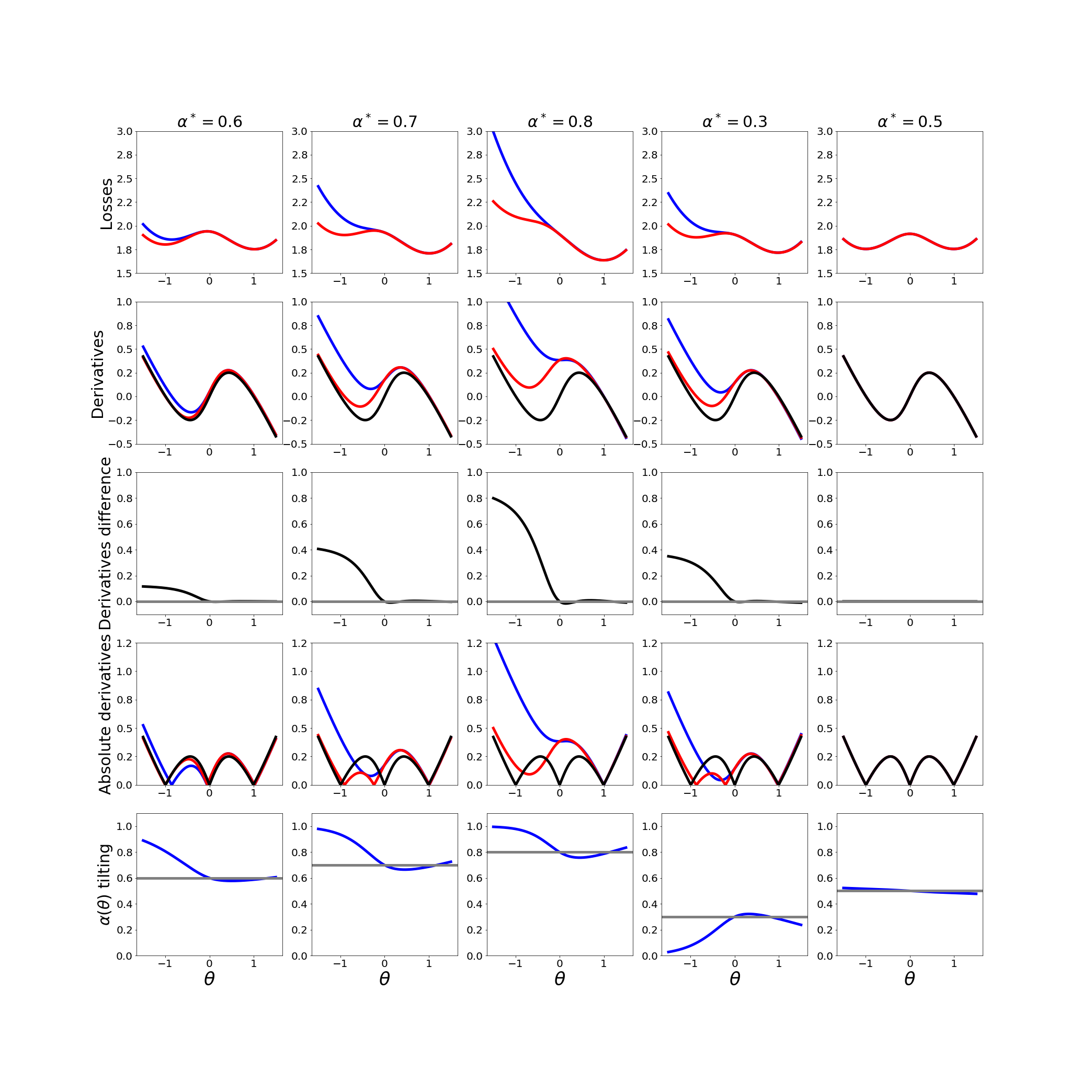}.
\caption{Behavior of $\Ll$, $\ell$ and their derivatives for different values of $\alpha^*$.  Black lines correspond to the reference $\alpha=0.5$ (also in the last column). \textbf{First row} entropic OT ($\Ll$, blue) and negative log likelihood $\ell$ (red). \textbf{Second row} derivatives of $\Ll$ and $\ell$. \textbf{Third row} difference between the derivatives $\Ll$ and $\ell$. \textbf{Fourth row} absolute value of the derivatives. \textbf{Fifth row} optimal $\alpha(\theta)$ from the semi-dual entropic OT formulation. }
\label{fig:1}
\end{figure}

\subsection{Proof of Theorem \ref{teo:mixture}}
\subsubsection{Proof of equation \eqref{emgeom}}
Let us fix $\alpha^* > 0.5$ (otherwise the proof is trivial since Sinkhorn-EM and EM iterations coincide). We first recall the results of \cite[Theorem 1]{DasTzaZam17}, where the bound \eqref{emgeom} is stated for the EM algorithm in the symmetric mixture ($\alpha^*=0.5$). Let us denote $\theta^{t}_{EM_0}$ for the iterates of EM on the symmetric mixture, initialized at $\theta^0 > 0$.
We write $\theta^t_{SEM}$ for the iterates of SEM on the \emph{asymmetric} mixture.
We will show that, for all $t \geq 0$, $\theta^t_{SEM}$ and $\theta^{t}_{EM_0}$ satisfy
\begin{align}
\theta^* \leq \theta^t_{SEM} & \leq \theta^t_{EM_0} \quad \text{if $\theta^0 \geq \theta^*$,} \label{monotone>theta}\\
\theta^* \geq \theta^t_{SEM} & \geq \theta^t_{EM_0} \quad \text{if $0 < \theta^0 \leq \theta^*$.} \label{monotone<theta}
\end{align}
This will then prove the claim since it implies
\begin{equation*}
|\theta_{SEM}^{t} - \theta^*| \leq |\theta_{EM_0}^{t} - \theta^*| \leq \rho^{t} |\theta_{0} - \theta^*|.
\end{equation*}

It remains to prove~\eqref{monotone>theta} and~\eqref{monotone<theta}.
Recall the function $F$ defined in~\eqref{eq:def_f}.
We first show that 
\begin{equation}\label{eq:fcases}
F(\theta,\alpha(\theta)) \begin{cases} \leq \theta^*, &
  0< \theta< \theta^*  \\ = \theta^*,  &
  \theta =\theta^* \\ \geq \theta^*, &
\theta>\theta^*  \end{cases}.\end{equation}
This implies the first inequalities of ~\eqref{monotone>theta} and~\eqref{monotone<theta}.
To show \eqref{eq:fcases}, notice first that clearly $F(\theta^*,\alpha(\theta^*))=F(\theta^*,\alpha^*)=\theta^*$. 
Therefore, it is enough to establish that $\theta \mapsto F(\theta,\alpha(\theta))$ is non-decreasing.
Let us define $f(\theta)=F(\theta,\alpha(\theta))$. We then have
 \begin{equation} \label{eq:partialf}f'(\theta)=\frac{\partial F}{\partial \theta}(\theta,\alpha(\theta))+\frac{\partial F}{\partial \alpha}(\theta,\alpha(\theta))\alpha'(\theta),\end{equation} 
and
 \begin{eqnarray}\label{eq:partialF1}
 \frac{\partial F}{\partial \theta}(\theta,\alpha)&=&4\alpha(1-\alpha)\int y^2\frac{q_{\theta^*,\alpha^\ast}(y)}{\left(\alpha e^{\theta y }+(1-\alpha)e^{-\theta y}\right)^2}\dd y \geq 0,\\ \label{eq:partialF2}
  \frac{\partial F}{\partial \alpha}(\theta,\alpha)&=&2\int y\frac{q_{\theta^\ast,\alpha^\ast}(y)}{\left(\alpha e^{\theta y }+(1-\alpha)e^{-\theta y}\right)^2}\dd y.
 \end{eqnarray}
 Additionally, by taking derivatives with respect to $\theta$ in  \eqref{eq:optimalpha2} we have
 \begin{equation}\label{eq:alphaprime}
\alpha'(\theta)=-\frac{\partial G}{\partial \alpha}(\theta,\alpha(\theta))^{-1}\frac{\partial G}{\partial \theta}(\theta,\alpha(\theta)),
\end{equation}
and likewise, 
 \begin{eqnarray}\label{eq:partialG1}
 \frac{\partial G}{\partial \theta}(\theta,\alpha)&=&2\alpha(1-\alpha)\int y\frac{q_{\theta^*,\alpha^\ast}(y)}{\left(\alpha e^{\theta y }+(1-\alpha)e^{-\theta y}\right)^2}\dd y ,\\ \label{eq:partialG2}
  \frac{\partial G}{\partial \alpha}(\theta,\alpha)&=&\int \frac{q_{\theta^*,\alpha^\ast}(y)}{\left(\alpha e^{\theta y }+(1-\alpha)e^{-\theta y}\right)^2}\dd y> 0.
 \end{eqnarray}
The conclusion follows by replacing \eqref{eq:partialF1},\eqref{eq:partialF2},\eqref{eq:alphaprime},\eqref{eq:partialG1} and \eqref{eq:partialG2} in \eqref{eq:partialf} and invoking the Cauchy-Schwarz inequality.

We now show the second inequalities in~\eqref{monotone>theta} and~\eqref{monotone<theta}.
To this end, we will first show
\begin{equation} \label{eq:Falphacases} F(\theta,\alpha(\theta))\begin{cases}\geq F(\theta,0.5) & 0\leq \theta\leq \theta^*, \\ \leq F(\theta,0.5) & \theta\geq \theta^*.\end{cases}\end{equation}

Let $\phi$ denote the density of a standard Gaussian random variable.
We can write
\begin{align*}\frac{F(\theta,\alpha)-F(\theta,0.5)}{2\alpha-1} &= \int y \cdot \frac{\left(\alpha^*e^{\theta^* y}+(1-\alpha^*)e^{-\theta^* y}\right)}{\left(e^{\theta y}+e^{-\theta y}\right)\left(\alpha e^{\theta y}+(1-\alpha)e^{-\theta y}\right)}\phi(y)e^{-{\theta^*}^2/2}\dd y.\\
&=: \int y  \cdot \rho_{\theta, \alpha}(y) \dd y\,.
\end{align*}
It is straightforward to verify that for $\alpha, \alpha^* \geq 1/2$, if $\leq \alpha \leq \alpha^*$ and $\theta \leq \theta^*$, then
\begin{equation*}
\rho_{\theta, \alpha}(y) \geq \rho_{\theta, \alpha}(-y) \quad \forall y \geq 0.
\end{equation*}
On the other hand, if $\alpha \geq \alpha^*$ and $\theta \geq \theta^*$, then
\begin{equation*}
\rho_{\theta, \alpha}(y) \leq \rho_{\theta, \alpha}(-y) \quad \forall y \geq 0.
\end{equation*}
In particular, this yields that for $\alpha, \alpha^* \geq 1/2$,
\begin{equation*}
\frac{F(\theta,\alpha)-F(\theta,0.5)}{2\alpha-1} \begin{cases} \geq 0 & \text{if $\alpha \leq \alpha^*$ and $0 \leq \theta \leq \theta^*$} \\
\leq 0 & \text{if $\alpha \geq \alpha^*$ and $\theta \geq \theta^*$.}
\end{cases}
\end{equation*}
To complete the proof of~\eqref{eq:Falphacases}, we used the facts, proved in Lemma \ref{lemma:atheta} that $\alpha(\theta) \geq 1/2$ and that $\alpha(\theta) \leq \alpha^*$ if $0 \leq \theta \leq \theta^*$ and $\alpha(\theta) \geq \alpha^*$ if $\theta \geq \theta^*$.

Now, we use the fact that the iterates $\theta_{EM_0}^t$ satisfy (see for example \cite{DasTzaZam17})
\begin{equation}\label{eq:em0_iterates}
\theta_{EM_0}^{t+1} = F(\theta_{EM_0}^{t}, 0.5)\,.
\end{equation}


With this, we can now show the second two inequalities in~\eqref{monotone>theta} and~\eqref{monotone<theta}.
We proceed by induction.
Let's first suppose $\theta^0 \geq \theta^*$.
Then indeed for $t = 0$, we have $\theta^* \leq \theta^t_{SEM} \leq \theta^t_{EM_0}$.
If this relation holds for some $t$, then we have
\begin{align*}
\theta^{t+1}_{SEM} & = F(\theta^{t}_{SEM}, \alpha(\theta^{t}_{SEM})) \\
& \leq F(\theta^{t}_{SEM}, 0.5) \\
& \leq F(\theta^{t}_{EM_0}, 0.5) \\
& = \theta^{t+1}_{EM_0}\,,
\end{align*}
where the first inequality uses~\eqref{eq:Falphacases}, the second uses the fact that $F$ is an increasing function in its first coordinate~\eqref{eq:partialF1}, and the final equality is~\eqref{eq:em0_iterates}.
The proof of the second inequality in~\eqref{monotone<theta} is completely analogous.
\subsubsection{Proof of equation \eqref{eq:fast}}
Again, we can assume that $\alpha^\ast>0.5$. Suppose first that $\theta^0> \theta^*$.
  In this case, it suffices to show that \begin{equation}
      \label{eq:falpha}
 F(\theta, \alpha) \leq F(\theta, \alpha^*), \text{ for all } \alpha \geq \alpha^* \text{ and } \theta \geq \theta^*
  \end{equation}
  Indeed, in that case we can appeal to a similar inductive argument as in the proof of Theorem~\ref{teo:mixture}, equation~\eqref{emgeom}, to compare the iterates of SEM  $\theta_{SEM}^{t+1} = F(\theta_{SEM}^{t}, \alpha(\theta_{SEM}^{t}))$ to those of EM, $\theta_{EM}^{t+1} = F(\theta_{EM}^{t}, \alpha^*)$.
  Specifically, if \eqref{eq:fast} holds for some $t$, then, since $\theta^t_{SEM}\geq \theta^*$ we have, by Lemma \ref{lemma:atheta}, that $\alpha(\theta^t_{SEM})\geq \alpha^*$. Then,
  \begin{eqnarray*}
\theta^{t+1}_{SEM}&=&F(\theta^t_{SEM},\alpha(\theta^t_{SEM}))\\
&\leq & F(\theta^t_{SEM}, \alpha^*)\\
&\leq & F(\theta^t_{EM}, \alpha^\ast)\\
&= & \theta^{t+1}_{EM}.
  \end{eqnarray*}
  Above, in the first line, we used \eqref{eq:falpha}, and in the second, that $F$ is increasing on its first argument. 
  Suppose now that $\theta^0\in [\theta_{\text{fast}},\theta^*]$, for some $0<\theta_{\text{fast}}<\theta^*$, and  assume that 
  \begin{equation}
      \label{eq:falpha2}
 F(\theta, \alpha) \leq F(\theta, \alpha^*), \text{ for all } \frac{1}{2} < \bar{\alpha} \leq \alpha \leq \alpha^* \text{ and } \theta_{\text{fast}}\leq \theta \leq \theta^*.
  \end{equation}
Where $\bar{\alpha}=\inf_{\theta\in\mathbb{R}} \alpha(\theta)$.  The fact that $\bar{\alpha}>0.5$ follows from Lemma \ref{lemma:atheta}: since $\alpha(\theta)$ is increasing if $\theta>\theta^*$ and decreasing if $\theta<0$ and $\alpha(0)=\alpha(\theta^*)= \alpha^*$, the minimum of $\alpha(\theta)$ must occur on $[0,\theta^*]$. But since $\alpha(\theta)$ is continuous \textcolor{red}{proove}, such a minimum must be achieved, contradicting that $\alpha(\theta)>1/2$.
Assuming \eqref{eq:falpha2}, the same argument yields $\theta_{\text{fast}}\leq \theta^t_{EM}\leq \theta^t_{SEM} \leq \theta^*$.  
  Let's then show \eqref{eq:falpha} and \eqref{eq:falpha2}. We have that
  \begin{align*}
 \nonumber \frac{F(\theta,\alpha)-F(\theta,\alpha^*)}{2(\alpha-\alpha^*)}=\frac{1}{2(\alpha-\alpha^*)}&\int y \left[\frac{\alpha e^{\theta y} - (1-\alpha) e^{-\theta y}}{\alpha e^{\theta y} + (1-\alpha) e^{-\theta y}} -\frac{\alpha e^{\theta y} - (1-\alpha) e^{-\theta y}}{\alpha e^{\theta y} + (1-\alpha) e^{-\theta y}} \right] q_{\theta^*,\alpha^\ast}(y)\dd y \\
 \frac{1}{2(\alpha-\alpha^*)}&\int y \frac{2(\alpha(1-\alpha^*)-(1-\alpha)\alpha^*)e^{\theta y}e^{-\theta y} q_{\theta^*,\alpha^\ast}(y)}{\left(\alpha e^{\theta y} + (1-\alpha) e^{-\theta y}\right)\left(\alpha^* e^{\theta y} + (1-\alpha^*) e^{-\theta y}\right)}\dd y\\
 =&\int y\frac{q_{\theta^*,\alpha^\ast}(y)}{\left(\alpha e^{\theta y} + (1-\alpha) e^{-\theta y}\right)\left(\alpha^* e^{\theta y} + (1-\alpha^*) e^{-\theta y}\right)}\dd y \\  \label{eq:fdif}=& \int_{y\geq 0} f_\theta(y)\dd y,
  \end{align*}
where \begin{align*}f_\theta(y):=y\frac{g_\theta(y)\phi(y) e^{-{\theta^*}^2/2 }}{\left(\alpha e^{\theta y} + (1-\alpha) e^{-\theta y}\right)\left(\alpha^* e^{\theta y} + (1-\alpha^*) e^{-\theta y}\right)\left(\alpha e^{-\theta y} + (1-\alpha) e^{\theta y}\right)\left(\alpha^* e^{-\theta y} + (1-\alpha^*) e^{\theta y}\right)},\end{align*}
and \begin{align*}
g_\theta(y):=&\left(\alpha e^{-\theta y} + (1-\alpha) e^{\theta y}\right)\left(\alpha^* e^{-\theta y} + (1-\alpha^*) e^{\theta y}\right)q_{\theta^*}(y) \\ &-\left(\alpha e^{\theta y} + (1-\alpha) e^{-\theta y}\right)\left(\alpha^* e^{\theta y} + (1-\alpha^*) e^{-\theta y}\right)q_{\theta^*}(-y)\\ 
= & L \left(e^{y\left(2\theta-\theta^*\right)}-e^{-y\left(2\theta-\theta^*\right)}\right)+M\left(e^{y\theta^*}-e^{-y\theta^*}\right)+N \left(e^{y\left(2\theta+\theta^*\right)}-e^{-y\left(2\theta+\theta^*\right)}\right)
\end{align*}
with
\begin{eqnarray*}
L &=&(1-\alpha^*)^2(1-\alpha)- \alpha {\alpha^*}^2,\\
M &=& (2\alpha^*-1)(\alpha+\alpha^*-2\alpha\alpha^*), \\
N  &= & \alpha^*(1-\alpha^*)(1-2\alpha).
\end{eqnarray*}
Notice that for $y\geq 0$ and $\theta>\theta^*$ the three above differences of exponentials are positive, and that $e^{y\left(2\theta-\theta^*\right)}-e^{-y\left(2\theta-\theta^*\right)}\geq e^{y\theta^*}-e^{-y\theta^*}$. Moreover, if $1/2<\alpha^*,\alpha < 1$, then $N<0$ and $M>0$. Additionally, 
$L+M$ is negative if $\alpha,\alpha^*\in(1/2,1]$ since
\begin{eqnarray*}M+L&=& (1-\alpha^*)^2(1-\alpha)- \alpha {\alpha^*}^2+ (2\alpha^*-1)(\alpha+\alpha^*-2\alpha\alpha^*)\\
&=&1-2\alpha^*+{\alpha^*}^2-\alpha+2\alpha^*\alpha-\alpha{\alpha^*}^2-\alpha{\alpha^*}^2+2\alpha^\ast\alpha+2{\alpha^*}^2-4\alpha{\alpha^*}^2-\alpha-\alpha^*+2\alpha\alpha^*\\
&=&6\alpha\alpha^*-6\alpha{\alpha^*}^2+3{\alpha^*}^2-3\alpha^*-2\alpha+1\\
&=&
(1-2\alpha)(1+3\alpha^*(\alpha^*-1))<0,
    \end{eqnarray*}
    and $1+3\alpha^*(\alpha^*-1)>1/4$ if $\alpha^*\in(1/2,1]$.
Therefore,  
\begin{eqnarray}\nonumber g_\theta(y)< -|L+M|\left(e^{y\theta^*}-e^{-y\theta^*}\right)\leq 0. \end{eqnarray}
This proves that when $\alpha \geq \alpha^* > 1/2$, we have $F(\theta, \alpha) \leq F(\theta, \alpha^*)$, as claimed, so \eqref{eq:falpha} is proven.

Let's now show \eqref{eq:falpha2}.
It suffices to show that for $\theta  \in [\theta_{\text{fast}}, \theta^*]$, $1/2<\bar{\alpha}\leq \alpha\leq \alpha^*$ and $y>0$ we have
$g_\theta(y) \leq 0$.
First, note that, since $N < 0$, for any $\theta > \theta^*/2$, the term $N \left(e^{y\left(2\theta+\theta^*\right)}-e^{-y\left(2\theta+\theta^*\right)}\right)$ is always eventually dominant, so there exists a $y^*>0$ such that
\begin{equation*}
g_\theta(y) <0 \quad \forall \theta > \theta^*/2, y > y^*\,.
\end{equation*}
It therefore suffices to focus on the compact interval $[0, y^*].$ To proceed, let us consider what happens when $\theta = \theta^*$.
As before, we have that if $\alpha>1/2$
\begin{equation*}
g_{\theta^*}(y) < -|L+M| (e^{y \theta^*} - e^{- y \theta^*}) \quad \forall y \geq 0,
\end{equation*}

Let us now examine the derivative $\frac{\partial }{\partial \theta} g_\theta(y)$:
\begin{equation*}
\frac{\partial }{\partial \theta} g_\theta(y) = 2 y L(e^{y\left(2\theta-\theta^*\right)}+e^{-y\left(2\theta-\theta^*\right)}) + 2y N(e^{y\left(2\theta+\theta^*\right)}+e^{-y\left(2\theta+\theta^*\right)})\,.
\end{equation*}
Which is a negative quantity since $L,N<0$. We can bound its absolute value as follows:
\begin{eqnarray*}
\left|\frac{\partial }{\partial \theta} g_\theta(y)\right| \leq 2|y^*|\left(2|L|e^{y\theta^*}+|N|\left(e^{3y\theta^*}+e^{-y\theta^*}\right)\right)\leq
8|y^*||L+N|e^{3y\theta^*}.
\end{eqnarray*}
Note that the above bound is uniform in $\theta<\theta^*$ and $y\in[0,y^\ast]$. Then, for any $\theta<\theta^*$ we can write
\begin{eqnarray*}
g_{\theta}(y) &=& g_{\theta^*}(y) - \int_{\theta}^{\theta^*} \frac{\partial }{\partial \theta} g_\theta(y) \dd \theta \leq g_{\theta^*}(y) +  \sup_{\theta\in [\theta,\theta^\ast], y \in[0,y^*]}\left|\frac{\partial }{\partial \theta} g_\theta(y)\right| \cdot (\theta^* - \theta) \\
 &\leq & g_{\theta^*}(y) + 8|y^*||L+N|e^{3y\theta^*}(\theta^* - \theta).
\end{eqnarray*}
We want to show that we can choose $\theta'=\theta_{\text{fast}}$ close enough to $\theta^*$ such that the above sum is negative. To do so, first note that there is some $K_\varepsilon>0$ such that we can bound
$e^{3y\theta^*}\leq K_\varepsilon(e^{y\theta^*}-e^{-y\theta^*})$ on $y\in[\varepsilon,y^*]$. Indeed, such $K_\varepsilon$ is given by
$$K_\varepsilon:=\sup_{y\in[\varepsilon,y^*]} \frac{e^{2\theta^*y}}{1-e^{2\theta^*y}}.$$

Then, we can pick $\theta^0_{\text{fast}}$ such that $$\theta^*-\theta_{\text{fast}}<\frac{|L+M|}{8 |y^*| |L+N| K_\varepsilon},$$
and consequently, for $\theta\in [\theta^0_{\text{fast}},\theta]$ and $y\in[\varepsilon,y^\ast]$ 
$$g_\theta(y)\leq  g_{\theta^*}(y) + |L+M|\left(e^{y\theta^*}-e^{-y\theta^*})\right)\leq g_{\theta^*}(y)-g_{\theta^*}(y)\leq 0.$$
Therefore, it only remains to study the behavior of $g_\theta(y)$ when $y\leq \varepsilon$. To do so, let's rewrite $g_\theta(y)$ using hyperbolic sine
\[
g_\theta(y)
=2\Big(
L\sinh(y(2\theta-\theta^*))
+M\sinh(y\theta^*)
+N\sinh(y(2\theta+\theta^*))
\Big).
\]

The hyperbolic sine admits the Taylor expansion
\[
\sinh t = t + O(t^3)
\quad\text{as } t\to 0.
\]

Substituting into the expression for \(g_\theta(y)\) gives
\[
\begin{aligned}
g_\theta(y)
&=2\Big(
L[y(2\theta-\theta^*)+O(y^3)]
+M[y\theta^*+O(y^3)]
+N[y(2\theta+\theta^*)+O(y^3)]
\Big)\\
&=2y\Big(
L(2\theta-\theta^*)+M\theta^*+N(2\theta+\theta^*)
\Big)+O(y^3).
\end{aligned}
\]

Now, for $\theta \in [\theta^0_{\text{fast}},\theta^*]$ we can write $\theta=\theta^*-\delta$ with $ 0<\delta<\theta^*-\theta^0_{\text{fast}}$. Then
\[
g_\theta(y)
=2y\Big(
\theta^*(L+M+3N)-2\delta(L+N)
\Big)+O(y^3).
\]

One checks directly that
\[
L+M+3N = 1-2\alpha,
\qquad
L+N = 1-\alpha-\alpha^*.
\]
Hence
\[
g_\theta(y)
=2y\Big(
\theta^*(1-2\alpha)+2\delta(\alpha+\alpha^*-1)
\Big)+O(y^3).
\]
Since \(\alpha^*\ge \alpha>\bar{\alpha}>1/2\) and \(\theta^*>0\), we have \(1-2\alpha<0\) and
\(\alpha+\alpha^*-1>0\). Now, choose
$$0<\delta< \delta_0:=\frac{\theta^*(2\bar{\alpha}-1)}{4(\alpha+\alpha^*-1)},$$
so that
\[
g_\theta(y)
\leq y(1-2\bar{\alpha})\theta^*+O(y^3).
\]
Therefore, we can take $\theta_{\text{fast}}=\max\{\theta^0_{\text{fast}}, \theta^*-\delta_0\}$. By the above expansion, we have guaranteed that if we take $\varepsilon$ small enough, $g_\theta(y)\leq 0$ for $\theta\in [\theta_{\text{fast}},\theta^*]$ and $y\leq \varepsilon$.



\subsection{Intermediate results for Theorem \ref{teo:mixtureloss} and \ref{teo:mixture}}
\begin{proposition}\label{prop:deri}
For the asymmetric mixture of two Gaussians \eqref{eq:mixtureg} we have that for all $\theta<0$
\begin{equation} \label{eq:ineqLl}\Ll'_{\alpha^*}(\theta)>\ell'_{\alpha^*}(\theta)\end{equation} 
\end{proposition}
\begin{proof}
Let us write $\ell(\theta, \alpha) = - \mathbb{E} \log q_{\theta,\alpha}(Y)$, the expected negative log-likelihood function in the model \eqref{eq:mixtureg} (that is, we will deem the mixture weights as a parameter).
In our notation, $\ell(\theta,\alpha^\ast)=\ell(\theta)=\ell_{\alpha^\ast}(\theta)$.
We have
\begin{equation}\label{eq:partiall} \frac{\partial}{\partial \theta} \ell(\theta, \alpha) = \int y\left[\frac{\alpha e^{-(\theta-y)^2/2}-(1-\alpha) e^{-(\theta+y)^2/2}}{\alpha e^{-(\theta-y)^2/2}+(1-\alpha) e^{-(\theta+y)^2/2}}\right]q_{\theta^*,\alpha^*,}(y)\dd y-\theta = F(\theta, \alpha) - \theta\,,\end{equation}
where $F$ is defined in~\eqref{eq:def_f}.

We have $\ell'_{\alpha^*}(\theta) = \frac{\partial}{\partial \theta} \ell(\theta, \alpha^*) = F(\theta, \alpha^*) - \theta$.
Likewise, recalling that $\Ll_{\alpha^*}(\theta)=\tilde{\Ll}(\theta,\omega(\theta))$ where  $\omega(\theta)$ solves \eqref{eq:semidual}, then we have
\begin{equation}\label{eq:partialLmix} \Ll'_{\alpha^*}(\theta) =\frac{\partial \tilde{\Ll}}{\partial \theta}(\theta,\omega(\theta))=\frac{\partial }{\partial \theta}\ell(\theta,\alpha(\theta)) = F(\theta, \alpha(\theta)) - \theta.
\end{equation}

Then, to establish \eqref{eq:ineqLl} it suffices to show that for each $\theta<0$,
\begin{equation*}
F(\theta, \alpha^*) < F(\theta, \alpha(\theta))\,.
\end{equation*}

Moreover, since Lemma \ref{lemma:atheta} shows that $\alpha(\theta)> \alpha^*>0.5$, it suffices to show that $F(\theta, \alpha)$ is a strictly increasing function of $\alpha$ for $\alpha \geq 0.5$ and an arbitrary $\theta\leq 0$.

Recall~\eqref{eq:partialF2}, which shows
\begin{equation*}
\frac{\partial F}{\partial \alpha}(\theta,\alpha)=2\int y\frac{q_{\theta^*,\alpha^\ast}(y)}{\left(\alpha e^{\theta y }+(1-\alpha)e^{-\theta y}\right)^2}\dd y\,.
\end{equation*}

Since $\alpha^* > 0.5$, we have $q_{\theta^*,\alpha^\ast}(y) > q_{\theta^*,\alpha^\ast}(-y)$ for all $y > 0$.
Furthermore, for $\theta \leq 0$ and $\alpha \geq 0.5$, it holds
\begin{equation*}
\frac{1}{\left(\alpha e^{\theta y }+(1-\alpha)e^{-\theta y}\right)^2} \geq \frac{1}{\left(\alpha e^{-\theta y }+(1-\alpha)e^{\theta y}\right)^2} \quad \forall y > 0\,.
\end{equation*}
Therefore
\begin{align*}
\frac{\partial F}{\partial \alpha}(\theta,\alpha)& =2\int y\frac{q_{\theta^*,\alpha^\ast}(y)}{\left(\alpha e^{\theta y }+(1-\alpha)e^{-\theta y}\right)^2}\dd y \\
& = 2\int_{y > 0} y\frac{q_{\theta^*,\alpha^\ast}(y)}{\left(\alpha e^{\theta y }+(1-\alpha)e^{-\theta y}\right)^2}\dd y + 2\int_{y > 0} (-y)\frac{q_{\theta^*,\alpha^\ast}(-y)}{\left(\alpha e^{-\theta y }+(1-\alpha)e^{\theta y}\right)^2}\dd y > 0\,,
\end{align*}
which proves the claim.
\end{proof}

\begin{lemma}\label{lemma:atheta}
Suppose $\alpha^*>1/2$
\begin{itemize}
\item[(a)]for any $\theta$, $\alpha(\theta)> 1/2.$

\item[(b)] $\alpha(\cdot)$ is decreasing (increasing) whenever $\theta<0$ ($\theta>\theta^*$) and in either case $\alpha(\theta)\geq \alpha^*$. Moreover, $\lim_{\|\theta\|\rightarrow \infty}\alpha(\theta)=1$.
\item[(c)]$\alpha(\theta)\leq \alpha^*$ whenever $0\leq \theta\leq \theta^*$.

\end{itemize}\end{lemma}

\begin{proof}
We begin by recalling the function $G$, defined in~\eqref{eq:optimalpha3}.
By~\eqref{eq:partialG2}, this function is a strictly increasing function of $\alpha$ and $\alpha(\theta)$ is defined as the unique number in $[0, 1]$ satisfying
\begin{equation*}
G(\theta, \alpha(\theta)) = \alpha^*,
\end{equation*}
and $\alpha(\theta) > p$ if and only if $G(\theta, p) < \alpha^*$.
Let us first prove (a).
It suffices to show that $G(\theta, 1/2) < \alpha^*$.
Write $\phi_{\theta^*}$ for the density of $\cN(\theta^*, 1)$.
We then have
\begin{equation}
\label{eq:Gtheta12} G(\theta, 1/2) = \int \frac{e^{\theta y}}{e^{\theta y} + e^{-\theta y}}q_{\theta^*,\alpha^\ast}(y) \dd y= \int \frac{\alpha^* e^{\theta y} + (1-\alpha^*) e^{-\theta y}}{e^{\theta y} + e^{-\theta y}} \phi_{\theta^*}(y) \dd y\,.
\end{equation}
But if $\alpha^* > 1/2$, then $\frac{\alpha^* e^{\theta y} + (1-\alpha^*) e^{-\theta y}}{e^{\theta y} + e^{-\theta y}} < \alpha^*$ for all $\theta$ and $y$.
Since $\phi_{\theta^*}(y)$ is a probability density, we obtain that $G(\theta, 1/2) < \alpha^*$, as desired.

It is straightforward to see that $\alpha(0) = \alpha(\theta^*) = \alpha^*$ (see the proof of Proposition \ref{prop:L}).
To show monotonicity,  we rely on the formula \eqref{eq:alphaprime} for $\alpha'(\theta)$. 
If $\theta<0$, the conclusion is a direct consequence of \eqref{eq:alphaprime} and Lemma \ref{lemma:h}(b). 
If $\theta > \theta^*$, the conclusion follows similarly from Lemma \ref{lemma:h}(c), but the argument is more delicate, as applying this lemma requires that $\alpha(\theta)\geq \alpha^*$. 
Suppose that there exists a $\theta > \theta^*$ for which $\alpha'(\theta) < 0$.
Let us denote by $\theta_0$ the infimum over all such $\theta$.
By \eqref{eq:alphaprime}, $\frac{\partial{G}}{\partial \theta}(\theta_0,\alpha(\theta_0))$ must be therefore nonnegative, which by Lemma~\ref{lemma:h}(c) implies that $\alpha(\theta_0) < \alpha^* = \alpha(\theta^*)$. But since $\alpha'(\theta) \geq 0$ for all $\theta \in [\theta^*, \theta_0)$, this is a contradiction.
Therefore $\alpha'(\theta) \geq 0$ for all $\theta \geq \theta^*$, as claimed.

Finally, to show the limit statement, note first that $\alpha(\theta)$ is increasing if $\theta>\theta^*$, and it is also monotonic; therefore, it must have a limit. We will show that the limit must be one. Indeed, suppose that $\lim_{\theta\to\infty}\alpha(\theta)<1$. Then, there is a subsequence $\theta_n\to\infty$ such that $\alpha(\theta_n)<1-\varepsilon$ for $\varepsilon>0$. From the condition $\alpha^*=G(\theta,\alpha(\theta))$, we conclude that it must be the case that for all $\theta\in\mathbb{R}$
\begin{align*}
\alpha^* &  = \mathbb{E}\left(g(\theta,Y)\right),\quad \text{ with } g(\theta,Y)= \frac{\alpha(\theta) e^{\theta Y}}{\alpha(\theta) e^{\theta Y} + (1 - \alpha(\theta))e^{-\theta Y}}
\end{align*}
where the expectation is taken with respect to the true mixture centered at $\theta$ and $-\theta$ and weights $\alpha^*$ and $1-\alpha^*$. Note that by definition, $g(\theta,Y)<1$. Let $m>0$ be an arbitrary constant. Then, we can bound, along this subsequence 
\begin{eqnarray*} \mathbb{E}\left(g(\theta_n,Y)\right)&=&\mathbb{E}\left(g(\theta_n,Y)1_{Y>-m}\right)+\mathbb{E}\left(g(\theta_n,Y)1_{Y\leq-m}\right)\\
&\leq & \mathbb{P}\left(Y>-m\right)+\mathbb{E}\left(e^{-2m\theta_n}\frac{\alpha(\theta_n)}{1-\alpha(\theta_n)}\right)\\
&\leq &
\mathbb{P}\left(Y>-m\right)+e^{-2m\theta_n}\frac{1-\varepsilon}{\varepsilon}
\end{eqnarray*}

Now, note that $\mathbb{P}(Y\leq 0)= \alpha p+(1-\alpha)(1-p)=1-\alpha^*+p(2\alpha^*-1)$ where $p=\Phi(\theta^*)<1$.  Since $\alpha^*>1/2$, by continuity and mononicity of $m\to\mathbb{P}(Y\leq -m)$, we can choose $m>0$ so that $\mathbb{P}(Y\leq -m)>1-\alpha^*+p(2\alpha^*-1)$. Thus, for such a choice, 

\begin{eqnarray*} \mathbb{E}\left(g(\theta_n,Y)\right)&\leq &
\alpha^*-p(2\alpha^*-1)+e^{-2m\theta_n}\frac{1-\varepsilon}{\varepsilon}.
\end{eqnarray*}
As $\theta_n\to \infty$, and since the term containing the exponential above vanishes, we can find $n$ large enough so that $e^{-2m\theta_n}(1-\varepsilon)/\varepsilon<p(2\alpha^*-1)/2$. For such an $n$, 
\begin{eqnarray*} \mathbb{E}\left(g(\theta_n,Y)\right)&\leq &
\alpha^*+p(2\alpha^*-1)-\frac{p}{2}(2\alpha^*-1)<\alpha^*-\frac{p}{2}(2\alpha^*-1)<\alpha^*,
\end{eqnarray*}
This contradicts that $\alpha^*   = \mathbb{E}\left(g(\theta_n,Y)\right)$. The case $\theta\to-\infty$ is analogous, but now we work with the right tail of $Y$.

Let's now prove (c). Notice it suffices to show that (i) $\alpha'(0)<0$ and (ii) the only solutions to the equation $\alpha(\theta)=\alpha^*$ are $\theta=0$ and $\theta=\theta^*$. 
The first claim is a simple consequence of \eqref{eq:alphaprime} and Lemma~\ref{lemma:h}(b).

The second claim is a bit more involved. Suppose $\alpha(\theta)=\alpha^*$. By simple algebra (as in the proof of theorem \ref{teo:mixture}), it can be shown that the following relation holds
$$\int_{y\geq 0} \frac{2\alpha^*(1-\alpha^*)(2\alpha^*-1)e^{-{\theta^*}^2}\left(e^{2\theta y}-1\right)\left(e^{2\theta^*y} -e^{2\theta y}\right)}{e^{(\theta^*+2\theta)y}\left(\alpha^* e^{\theta y}+(1-\alpha^*)e^{-\theta y}\right)\left((1-\alpha^*) e^{\theta y}+\alpha^* e^{-\theta y}\right)}\phi(y)\dd y=0.$$ 
The integral above can only be zero if $\theta=0$ or $\theta=\theta^*$; otherwise, the integrand is either positive or negative for each value of $y\geq0$. This concludes the proof.
\end{proof}

\begin{lemma}\label{lemma:h} Suppose $\alpha^*> 0.5$.
Let
$$G_\theta(\theta,\alpha):=\frac{1}{2\alpha(1-\alpha)}\frac{\partial{G}}{\partial \theta}(\theta,\alpha)=\int \frac{y}{\left(\alpha e^{\theta y}+(1-\alpha)e^{-\theta y}\right)^2}q_{\theta^*,\alpha^\ast}(y)\dd y.$$
Then,
\begin{itemize}
\item[(a)] For each $\theta\geq 0$, $G_\theta$ is a decreasing as function of $\alpha$. Conversely, for each $\theta\leq 0$, $G_\theta$ is an increasing function of $\alpha$.
\item[(b)]$G_\theta(\theta,\alpha)\geq 0$ if $\theta\leq 0$ and $\alpha > 1/2$.
\item[(c)]$G_\theta(\theta,\alpha)\leq0$ if $\theta\geq \theta^*$ and $\alpha\geq \alpha^*$.
\end{itemize}
\end{lemma}
\begin{proof}
To see (a), notice that 
$$\frac{\partial G_\theta}{\partial \alpha}(\theta,\alpha)=-2\int  \frac{y\left( e^{\theta y}-e^{-\theta y}\right)}{\left(\alpha e^{y\theta}+(1-\alpha)e^{-y\theta}\right)^3}q_{\theta^*,\alpha^\ast}(y)\dd y.$$
The integrand is either positive (if $\theta>0$) or negative (if $\theta<0$) for each $y$, and the conclusion follows.

To prove (b) and (c), we note that~\eqref{eq:partialF2} implies that
\begin{equation*}
G_\theta(\theta, \alpha) = \frac 12 \frac{\partial F}{\partial \alpha}(\theta, \alpha)\,.
\end{equation*}
But we have already shown in the proof of Proposition~\ref{prop:deri} that $\frac{\partial F}{\partial \alpha}(\theta, \alpha) > 0$ for all $\theta \leq 0$ and $\alpha > 1/2$. This proves (b).

Likewise, the proof of Theorem \eqref{teo:mixture}, equation \eqref{eq:fast}, shows that $F(\theta, \alpha) \leq F(\theta, \alpha^*)$ for all $\theta\geq \theta^*$ and $\alpha\geq \alpha^*$. This proves that $G_\theta(\theta, \alpha^*) = \frac 12 \frac{\partial F}{\partial \alpha}(\theta, \alpha^*) \leq 0$.
To conclude, we appeal to part (a): since $\theta\geq\theta^*>0$, $G_\theta$ is decreasing as a function of $\alpha$, and hence, $G_\theta(\theta,\alpha)\leq G_\theta(\theta,\alpha^*)\leq0$.
\end{proof}
\begin{lemma}\label{lemma:alpha12}
Suppose $\alpha^*=0.5$. Then $\alpha(\theta)=0.5$ for all $\theta\in\mathbb{R}$. 

\end{lemma}
\begin{proof}
By \eqref{eq:Gtheta12}  $$G(\theta, 1/2) = \frac{1}{2} \int \frac{ e^{\theta y} +  e^{-\theta y}}{e^{\theta y} + e^{-\theta y}} \phi_{\theta^*}(y) \dd y\,=\frac{1}{2} \int \phi_{\theta^*}(y)\dd y\ =\frac{1}{2}.$$

Then, since $\alpha(\theta)$ is the only such that $G(\theta, \alpha(\theta))=1/2$, it must be that $\alpha(\theta)=1/2$.
\end{proof}
\section{Additional synthetic experiments}
In this section, we present additional experiments that supplement Sections \ref{sec:experiments} and \ref{sec:coclustering} in the main text.
\subsection{Unknown weights and known $K$ (supplement to Section \ref{sec:experiments}) }
\label{sub:exp2}
The setup is the same as experiment i) in Section \ref{sec:experiments}, but with unknown (and not necessarily uniform) weights. To measure non-uniformity we consider weights sampled from a Dirichlet distribution with concentration parameter $\gamma$ (a constant $K$-dimensional vector).  Smaller values of $\gamma$ indicate a larger deviation from the uniform distribution. Then, for each choice of $\gamma\in\{1,10,20,50,1000\}$, $K\in\{10,20,40\},d\in\{2,5,10,15,20\}$ and $\sigma_k^2=\sigma=0.1$, we ran $n_{exp}=80$ experiments under eight experimental made by the following experimental setups a)whether weights are updated, b)wether variances are updated, c)whether the initial weights are set to the true parameters or the uniform distribution $\alpha_k=1/K$.
We treat $\alpha$ as a parameter and update it using \eqref{eq:updatealpha} by performing coordinate descent on $\theta$ and $\alpha$ until convergence. In detail, to update $\theta$ we apply SEM with current $\alpha^t$ until convergence, leading to $\theta^{t+1}$. To update $\alpha$ we successively apply mirror descent updates \eqref{eq:updatealpha} with current $\theta^t$ until convergence, leading to $\alpha^{t+1}$. On each update, we start with $\eta=1$ and make $\eta\leftarrow \eta/2$ until $\Ll$ decreases with respect to initial $\alpha^t$. We always initialize $\alpha$ as the vector of uniform weights and study performance as a function of deviations from uniformity that we measure with the parameter $\gamma$, the concentration of a Dirichlet distribution. For each $\gamma$, we consider $n_{exp}=200$ experiments where weights are sampled from a Dirichlet distribution with parameters $\gamma/K$. Smaller values of $\gamma$ indicate a larger deviation from the uniform distribution. For comparison, we also include the algorithm that doesn't update weights as a baseline, i.e., we perform inference on a model with misspecified weights $\alpha_k=1/K$.

Results are summarized in Fig. \ref{fig:errorweights1},\ref{fig:errorweights2} and \ref{fig:errorweights3}. Benefits of the SEM algorithm are observed for moderately large values of $\gamma$ (e.g., $\gamma>10$) accross all experimental conditions. However, for smaller values of $\gamma$ performance is similar to comparable to EM, suggesting that the most beneficial regime occurs when the true distribution is close to uniform.

 \begin{figure}[H]
  \begin{center}
    \includegraphics[width=1.0\textwidth]{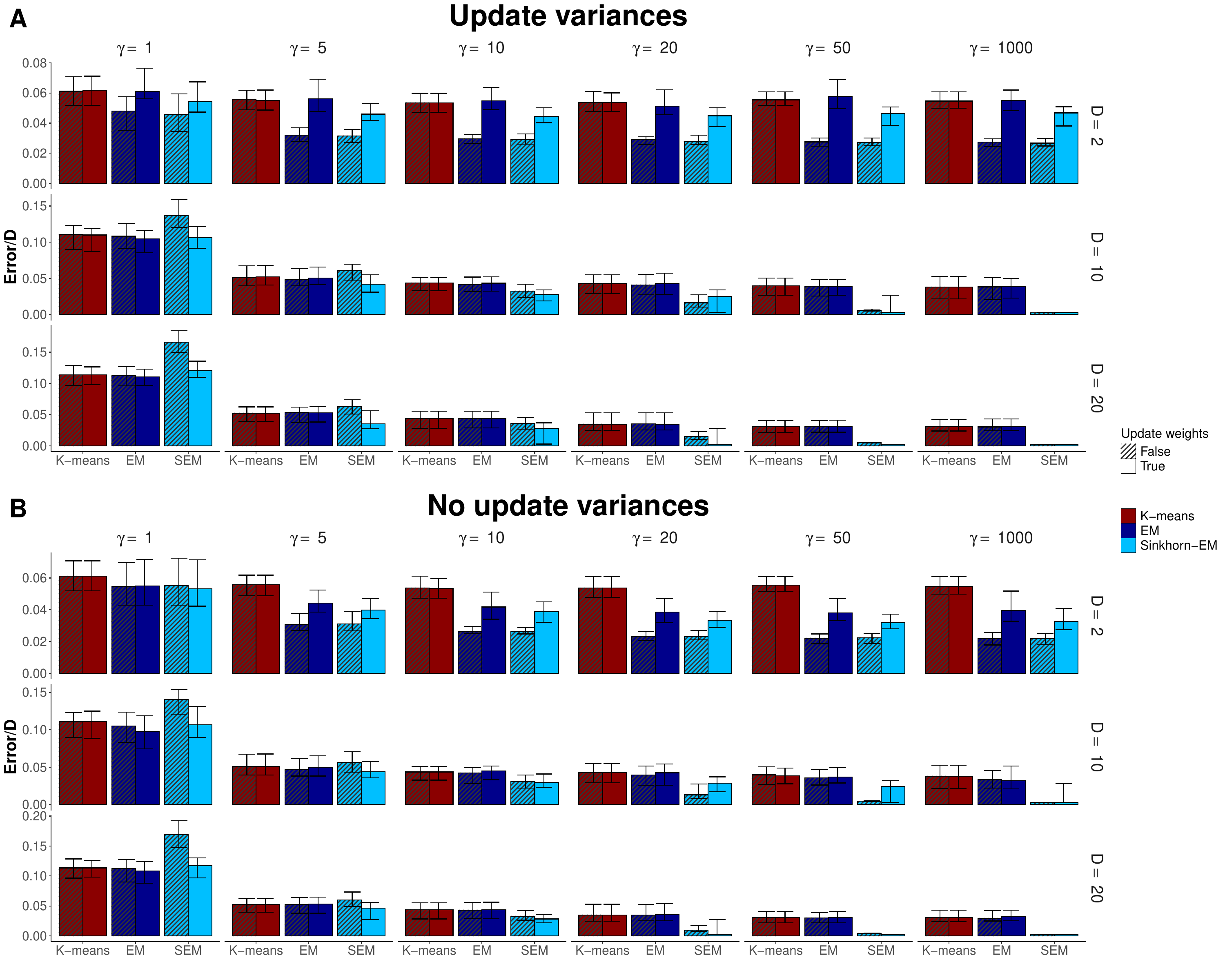}
  \end{center}
  \caption{Results (errors) for the unknown weights case in Section \ref{sub:exp2}, with starting weights set to uniform $\alpha_k=1/K$. A and B show results for the update/no update variance cases, respectively, and different rows show different number of components $K$. Wether weights are updated or not is indicated by shading in bars.Bars indicate the interquartile range. }
  \label{fig:errorweights1}
\end{figure}

 \begin{figure}[H]
  \begin{center}
    \includegraphics[width=1.0\textwidth]{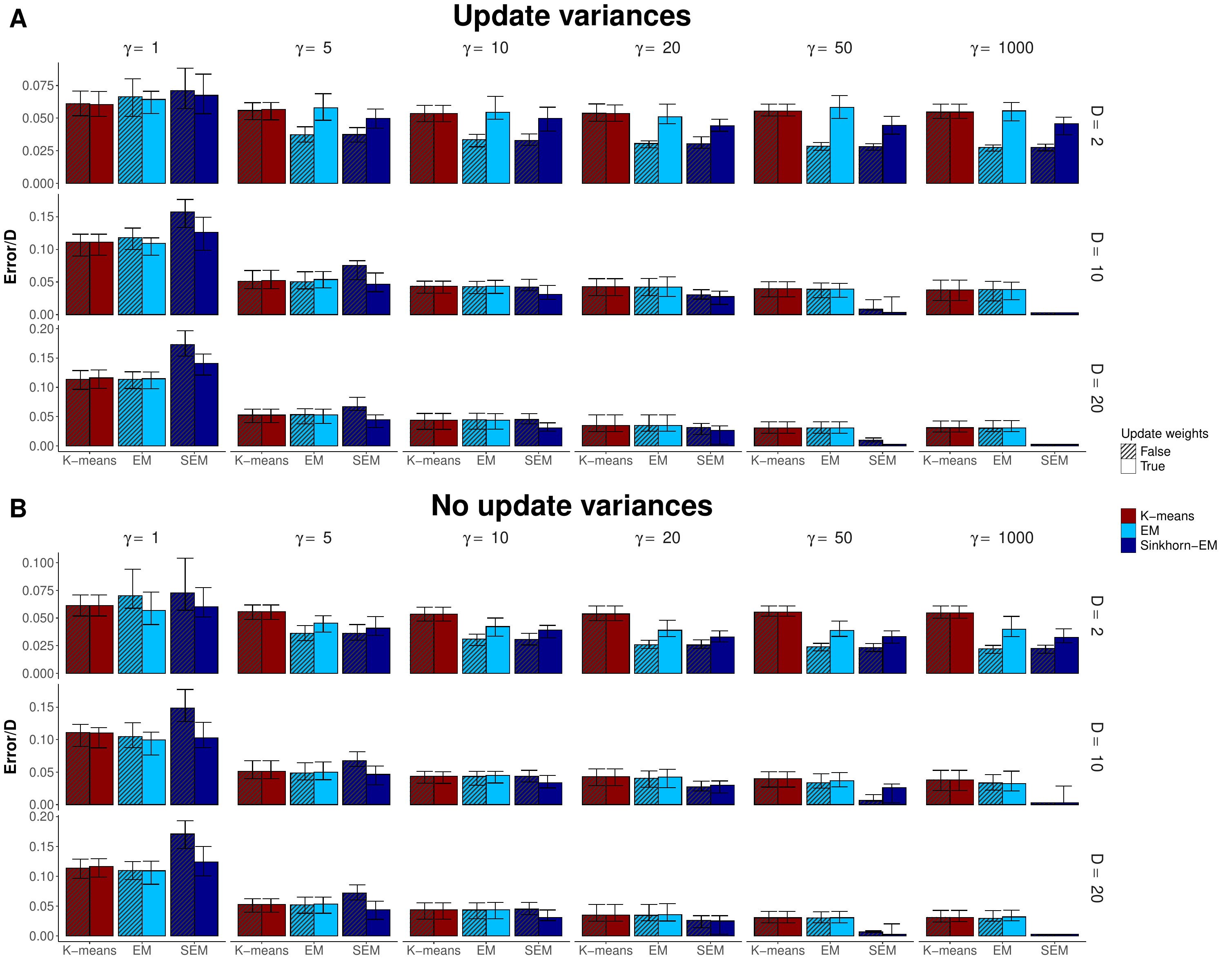}
  \end{center}
  \caption{Results (errors) for the unknown weights case in Section \ref{sub:exp2}, with starting weights set to their true values.  A and B show results for the update/no update variance cases, respectively, and different rows show different number of components $K$. Wether weights are updated or not is indicated by shading in bars.Bars indicate the interquartile range. }
  \label{fig:errorweights2}
\end{figure}

 \begin{figure}[H]
  \begin{center}
    \includegraphics[width=1.0\textwidth]{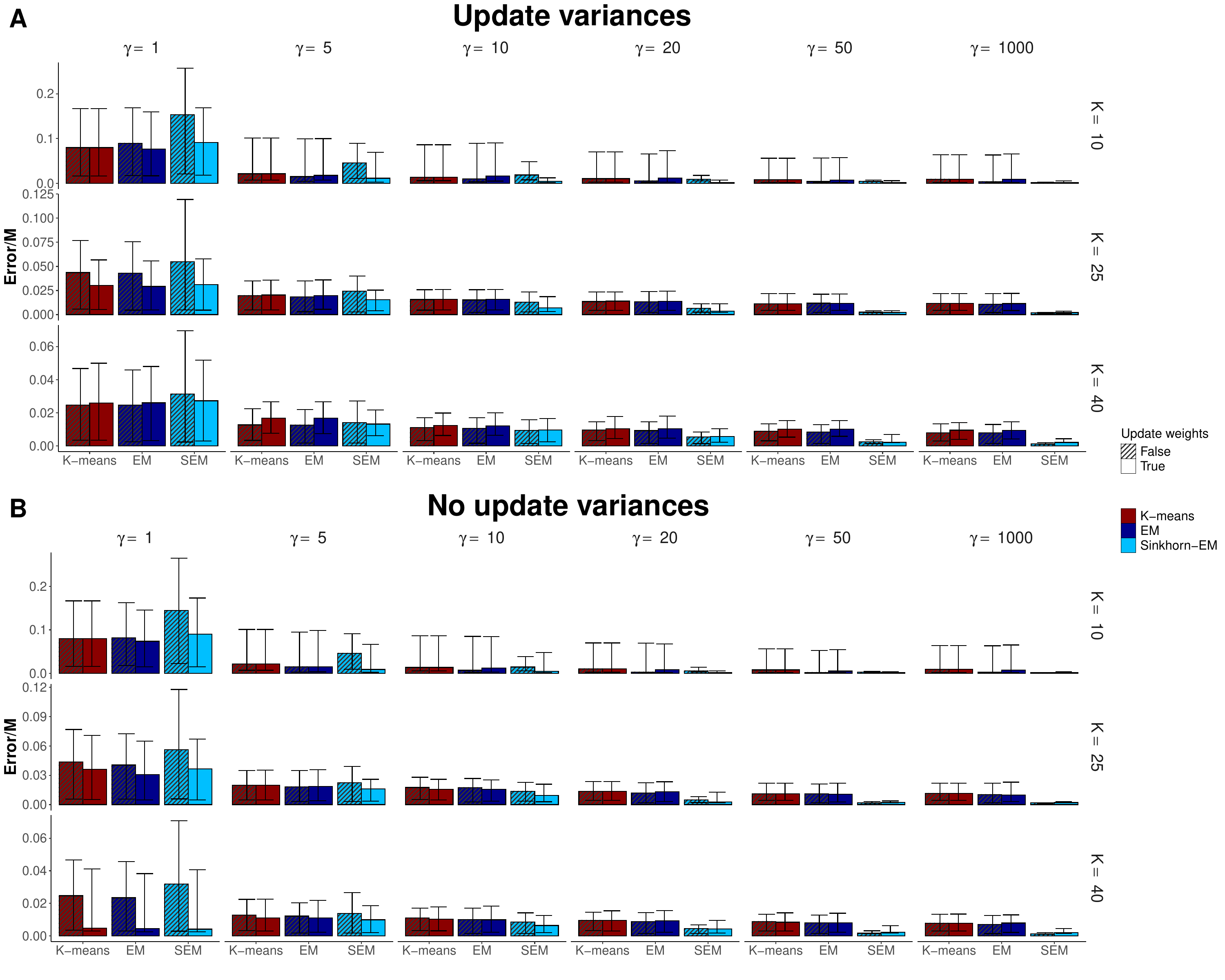}
  \end{center}
  \caption{Same as Fig. \ref{fig:errorweights1} but errors shown as a function of $M$ (different rows) }
  \label{fig:errorweights3}
\end{figure}
\subsection{Time comparisons}
As we discussed in the main text, the computational complexities of each EM and SEM iteration are $O(nK)$ and $\tilde{O}(nK)$, respectively. However, the latter conceils polynomial dependency in the norm $||C||_\infty$ of the cost matrix $C_{i,j}=||Y_i-\theta_k||^2/\sigma^2$  (in the Gaussian case) and inverse polynomial dependence on the error tolerance for the Sinkhorn iterations \citep{Altschuler2017}. Here we empirically quantify the extra computational burden of SEM over simple EM due to these parameters. To do so, we measure execution time of both full EM and SEM runs. Results are shown in Fig. \ref{fig:timesynthetic}. We observe that SEM is typical 10-100 times slower than EM in each experimental condition, with larger number of iteration leading to longer execution times (recall, as stopping criterion, we either stop if the marginals are $\varepsilon$ close to the prescribed ones or if the maximum number of iterations is exceeded). We also note that the computational time increases as both $M$ and $D$ increase, reflecting and increased cost $||C||_\infty$ in both cases. We note that errors are stable as a function of the number of iterations, suggesting that, in practice, we may only need a modest number of maximum Sinkhorn iterations (e.g. 50) to achieve benefits. As argued in these experiments, the extra computational cost is paid off by the ability to substantially decrease estimation errors compared to EM.

 \begin{figure}[H]
  \begin{center}
    \includegraphics[width=1.0\textwidth]{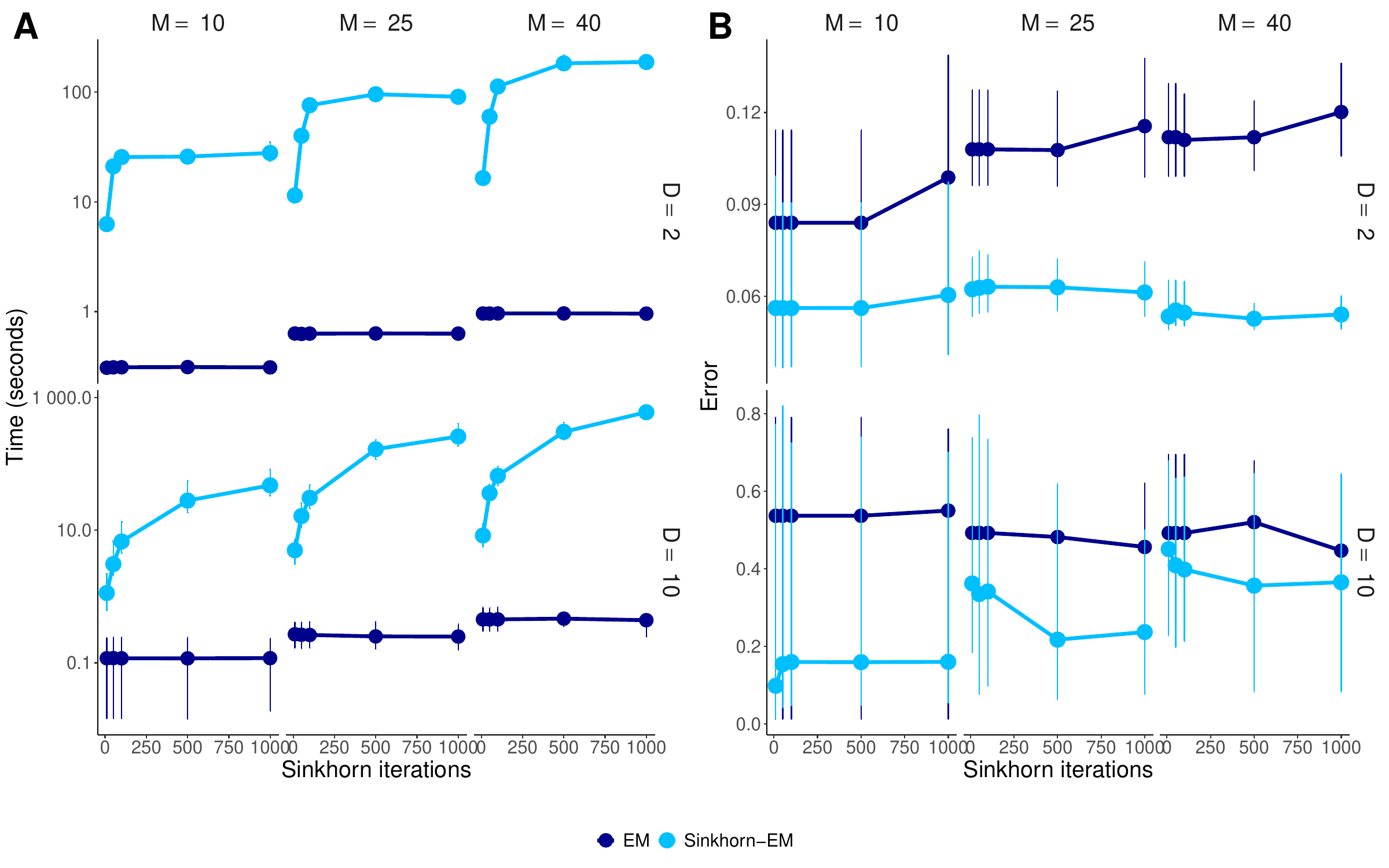}
  \end{center}
  \caption{Time (A) and error (B) comparisons for experiments on simulated data, as a function of number of Sinkhorn iterations, number of clusters $K$, and dimension $d$. Results for the unknown weights case in Section \ref{sub:exp2}.}
  \label{fig:timesynthetic}
\end{figure}

\subsection{Known weights and unknown $K$ (supplement to Section \ref{sec:experiments}) }
\label{sub:exp3}


We treat the unknown number of components case as a model selection problem over $K$. For each true $K$, we fit the model using several candidate models with $K_{model}\in\{K-5,K+5\}$ and estimate $\hat{K}$ as customary, via the Bayesian Information Criterion (BIC) \citep{kass1995reference} (not available for k-means). We quantify the estimation error as the difference $K-\hat{K}$ between the actual and inferred numbers of components. Compared to EM, SEM recovers the true number of components much more often. Moreover, in our experiment, SEM never overestimates the number of components ($\hat{K}>K$). In Appendix \ref{app:details} (Fig. \ref{fig:errorselectionext}) we provide detailed results for different choices of parameters.


\begin{figure}[H]
  \begin{center}
    \includegraphics[width=1.0\textwidth]{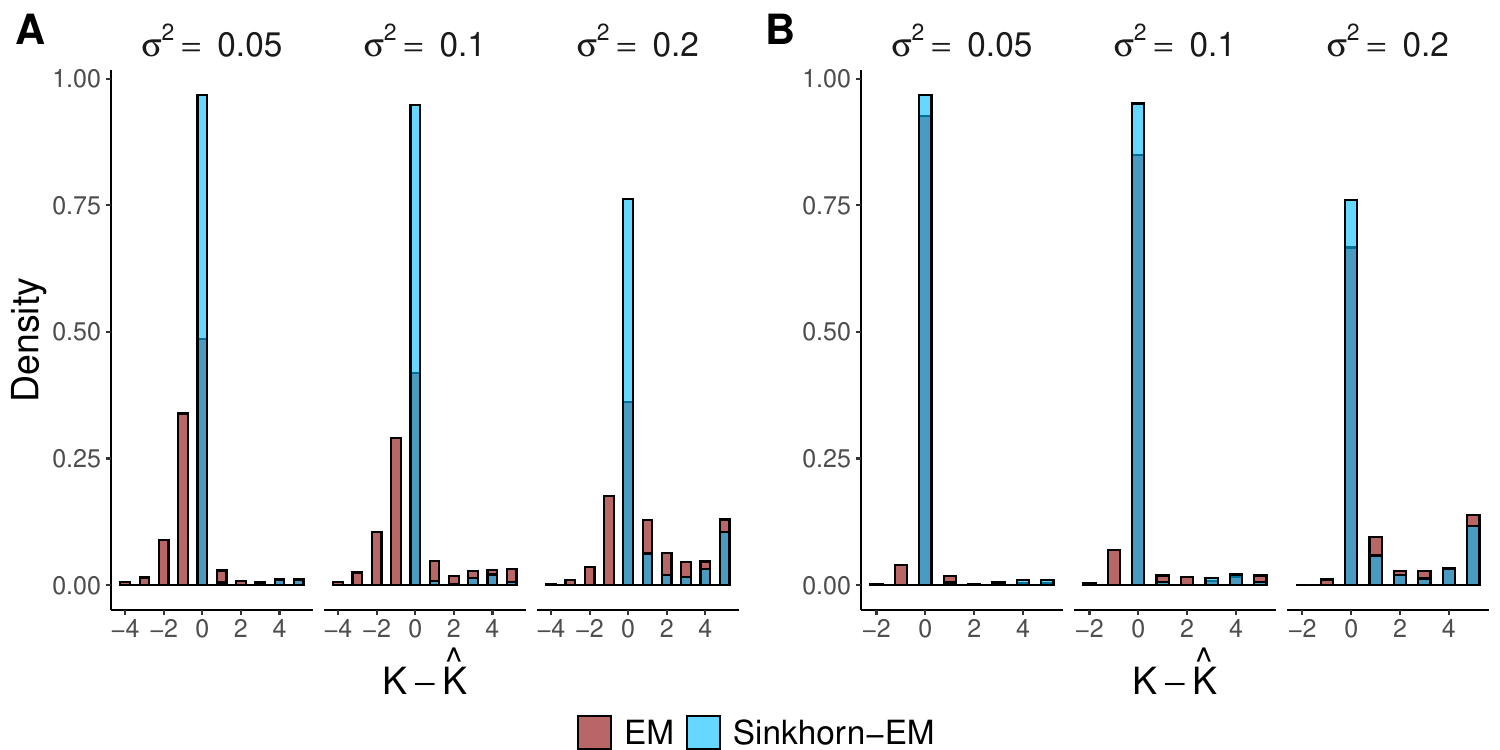}
  \end{center}
  \caption{Number of components estimation error histograms for the experiment in Section \ref{sub:exp3}. \textbf{A}: results for different seeds. \textbf{B}: results for the best seed}
  \label{fig:errorselection}
\end{figure}

\subsection{Experiments on Synthetic Data (supplement to Section \ref{sec:coclustering})}
\label{sub:coclust}
 We compare VEM, Sinkhorn-VEM and the competitive spectral co-clustering as implemented in \cite{pedregosa2011scikit} on simulated data sampled from the simple Gaussian generative model: $$Y_{i,j}|(z_{i,k}=1,w_{j,g}=1,\theta)\sim \mathcal{N}\left(\theta_{k,g},\sigma^2\right).$$
By performing thousands of experiments, we studied differences between estimated and actual $\theta_{k,g}$ for the three methods at different noise levels $\sigma^2$ and number of co-clusters $K^2$. Results are summarized in Fig. \ref{fig:biclustering}, Sinkhorn-VEM vastly outperforms both spectral and VEM co-clustering. Details appear in the Appendix \ref{app:coclust}.
\begin{figure}[H]
  \begin{center}
    \includegraphics[width=1.0\textwidth]{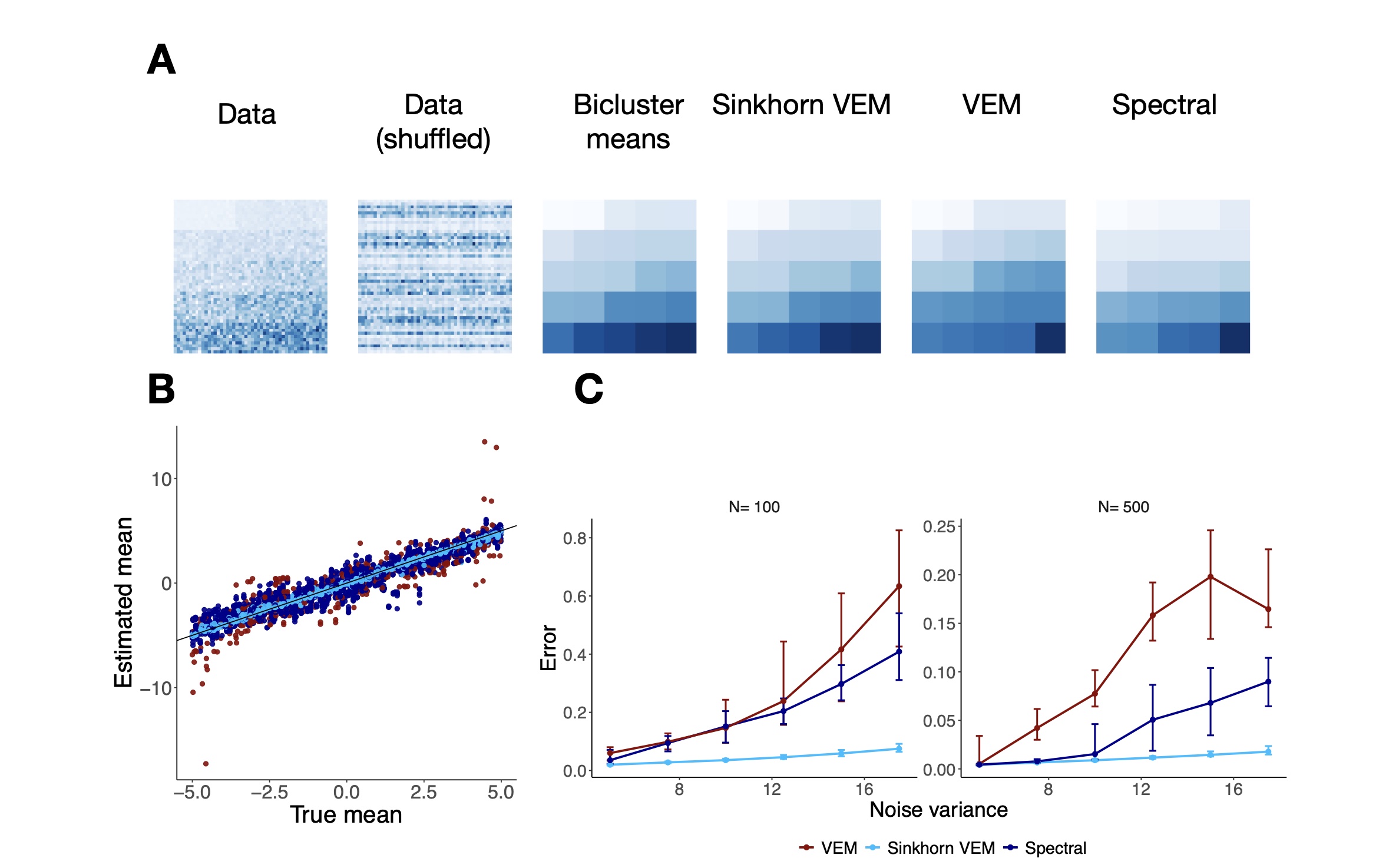}
  \end{center}
  \caption{Comparison of three Co-clustering methods based on simulated data. \textbf{A}: Setup for each experiment. The first inset contains data; each pixel is a noise-corrupted version of the corresponding co-cluster (each sub-square, also shown in the third subplot). Rows and columns of the data matrix are displayed in increasing order in the co-cluster for visualization purposes. The second inset is the same data matrix in the true shuffled order. The third inset shows the (ordered) co-cluster means. The last three columns are co-clusters recovered by each of the three methods. \textbf{B}: Scatterplot of true vs estimated co-cluster of means for three methods over thousands of experimental repetitions. \textbf{C}: Dependence of error on noise parameter $\sigma^2$ and sample sizes $N=M=100,500$ (so that the data matrix $Y$ had dimension $N\times N$). Sinkhorn-VEM consistently produces the most accurate co-clustering estimates.}
  \label{fig:biclustering}
\end{figure}

\section{Experimental details}
\label{app:details}
\subsection{Synthetic experiment details}\label{sub:details}
\sloppy For the four synthetic data experiments of Section \ref{sec:experiments} (main text) (i),(ii),(iii), and (iv), we considered several possible dimension sizes $d=2,5,10,20$, variances $\sigma^2=0.001,0.005,0.01,0.05,0.1,0.2,0.5$, and number of components $K=10,20,30,40$. \citep{brown2021bridges,boerner2023access}. Each true mean $\theta_k$ was sampled from a uniform $U(-1,1)$ distribution. In cases where variances were not spherical, diagonal entries were sampled from a $U(0.5\sigma^2,1.5\sigma^2)$ distribution (and i.i.d across different clusters). For each run of EM and Sinkhorn-EM, we performed at most 100 iterations with a termination criterion of $\varepsilon=10^{-3}$ change in the overall $L_1$ differences between consecutive parameter values (adding mean and variance differences). For each Sinkhorn-EM iteration, we performed at most 1000 iterations with a termination criterion of $\varepsilon=10^{-3}$ for $L_1$ differences between consecutive entropic optimal transport potentials. For k-means, we used the scikit-learn implementation with k-means++ initialization. This initialization was also used on the same seeds for Sinkhorn-EM and EM. 
Figs. \ref{fig:supperrorari1}, \ref{fig:200}, and \ref{fig:1000} supplement the results of the main text regarding experiment (i). In Fig. \ref{fig:supperrorari1} we compare the performance of three algorithms of each individual experiment, both at the level of individual seed or at the level of best seed. Both EM and Sinkhorn-EM outperform k-means and Sinkhorn-EM outperforms EM. In Figs. \ref{fig:200} we show performance of each algorithm as a function of $d$, $\sigma^2$ and $K$ for the two sample sizes we considered, $N=500$ and $N=1000$. These figures point to the same pattern described in the main text.
Figs. \ref{fig:difvar}, \ref{fig:unk1}, and \ref{fig:unk2} are analogs of \ref{fig:supperrorari1} in the main text for experiments (ii) (known unequal diagonal variances), (iii) (unknown spherical variance), and (iv) (unknown unequal variances). Results align with the same pattern, favoring Sinkhorn-EM.

For the unknown weights experiments of Section \ref{sub:exp2}, we considered parameters $d=2$, $K=10,20,30,40$, $N=1000$, $\sigma^2=0.005,0.001,0.01$, and $\gamma=1,10,20,50,100,1000$, and all other setups as in the experiments in Section \ref{sec:experiments} (main text). We jointly optimized over weights $\alpha$ and means $\theta$ by coordinate descent: for optimization over $\theta$ we performed Sinkhorn-EM and for optimization over $\alpha$ we used the mirror descent method described in \ref{sec:alpha}, with $\eta=1$. We iterated these two sub-routines until convergence.
For the model selection experiments of Section \ref{sub:exp3}, we considered several choices of $d$ and $K$ as shown in Fig. \ref{fig:errorselectionext}, which gives a more detailed account than Fig. \ref{fig:errorselection}.

\begin{figure}[H]
  \begin{center}
    \includegraphics[width=1.0\textwidth]{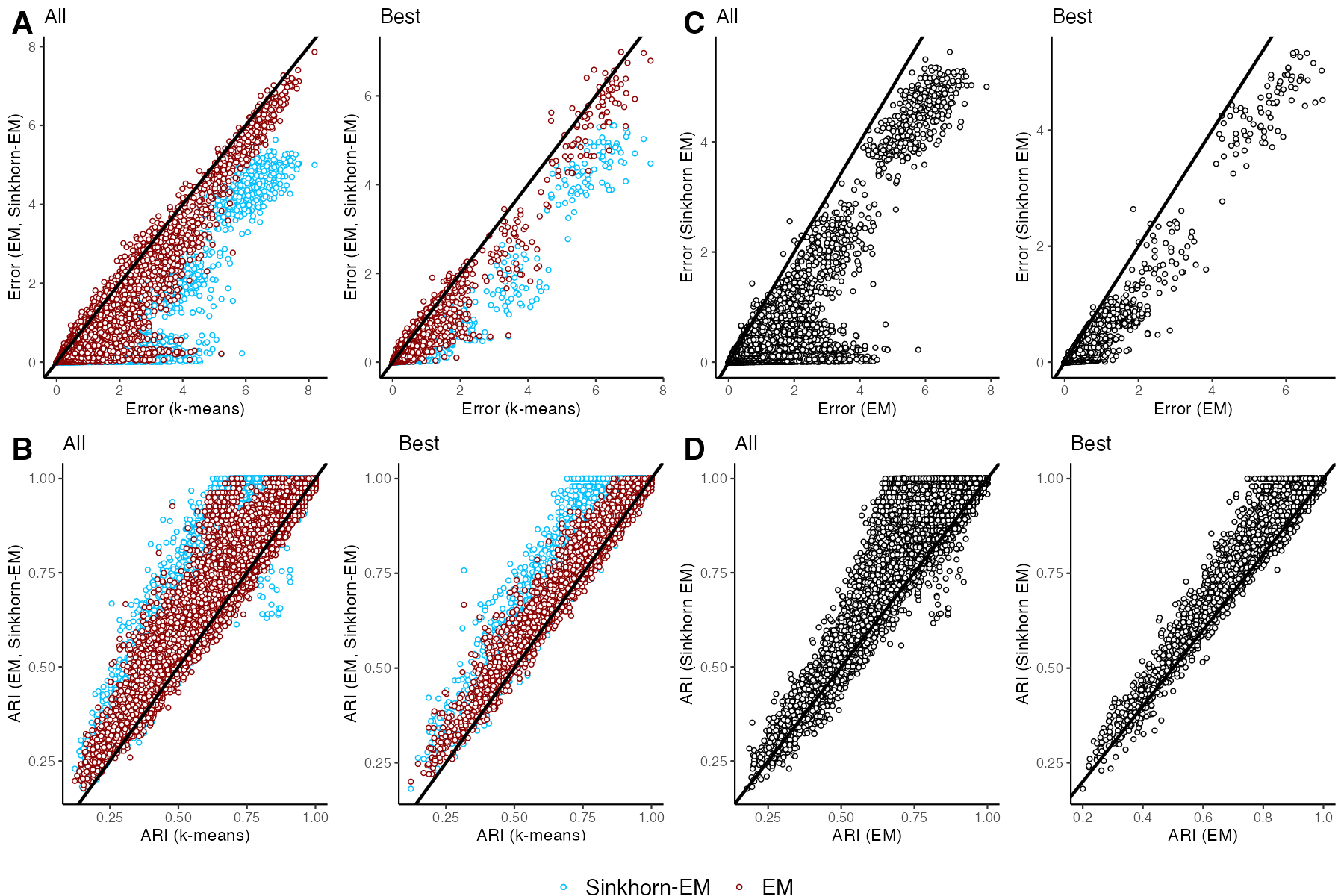}
  \end{center}
  \caption{Performance comparison of three methods in experiment (i) of Section \ref{sec:experiments} (main text). \textbf{A} Error scatterplots comparing the performance of k-means with EM or Sinkhorn-EM (colors). \textbf{B} same as \textbf{A} but with ARI score. \textbf{C,D} same as \textbf{A,B} but comparing EM with Sinkhorn-EM}
  \label{fig:supperrorari1}
\end{figure}

\begin{figure}[H]
  \begin{center}
    \includegraphics[width=1.0\textwidth]{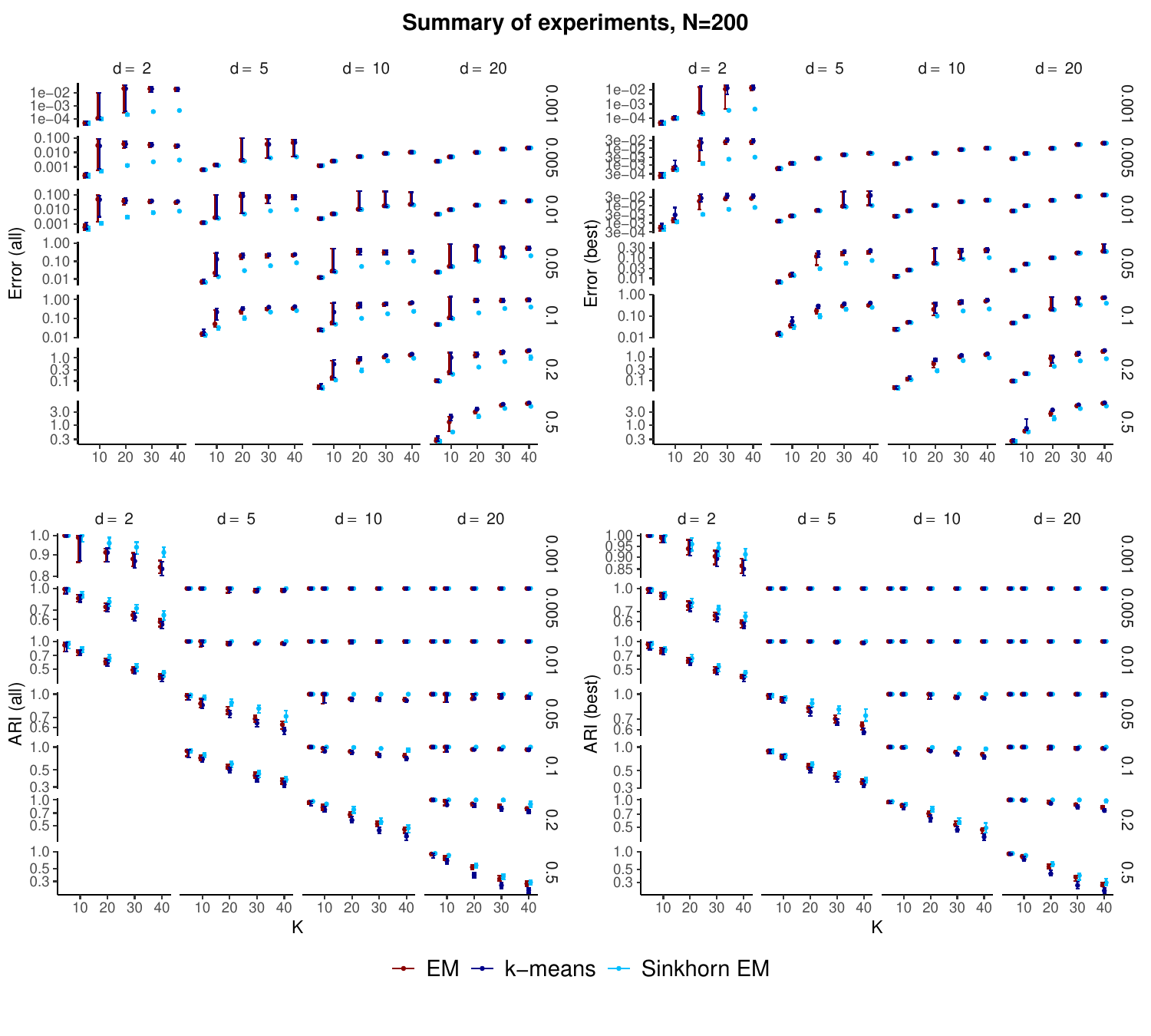}
  \end{center}
  \caption{Performance of three algorithms when the sample size is $N=200$. Each column represents a different dimension $d$ and each row a different noise variance $\sigma^2$.}
  \label{fig:200}
\end{figure}

\begin{figure}[H]
  \begin{center}
    \includegraphics[width=1.0\textwidth]{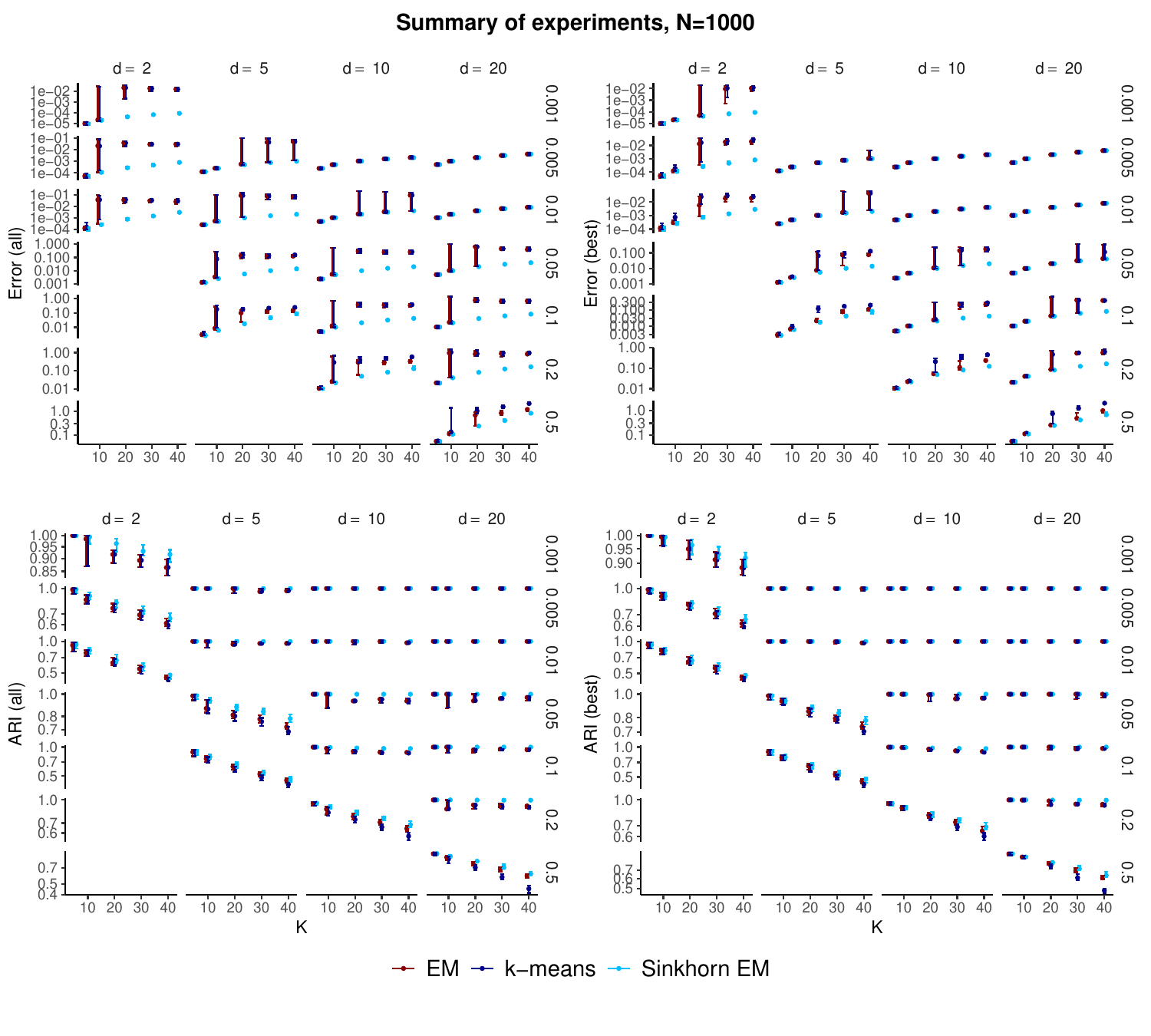}
  \end{center}
  \caption{Performance of three algorithms when the sample size is $N=1000$. Each column represents a different dimension $d$ and each row a different noise variance $\sigma^2$.}
  \label{fig:1000}
\end{figure}

\begin{figure}[H]
  \begin{center}
    \includegraphics[width=1.0\textwidth]{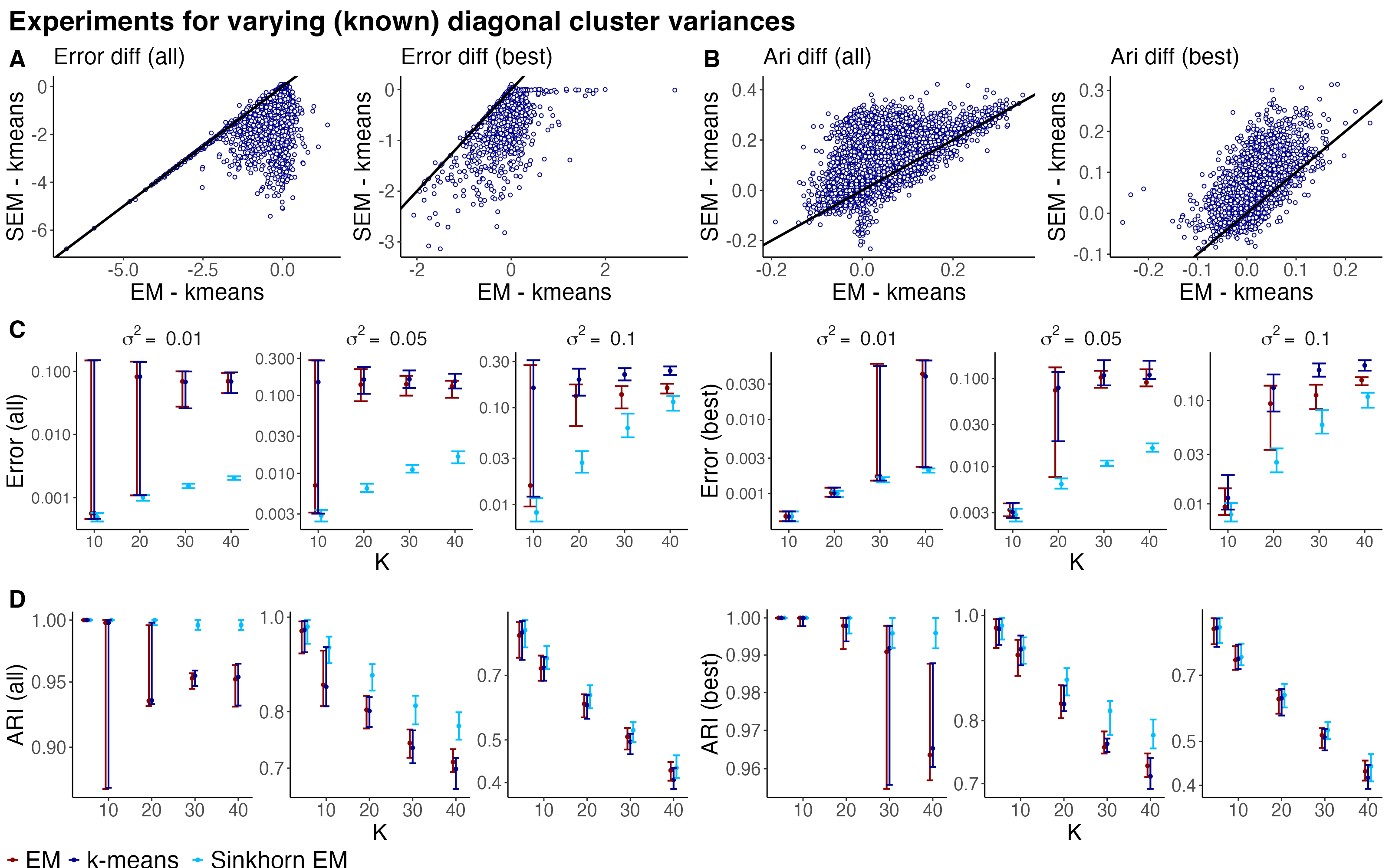}
  \end{center}
  \caption{Results for the experiment (ii) in Section \ref{sec:experiments}: each cluster is characterized by a different diagonal covariance matrix $\sigma^2_kI_d$, and these variances are assumed known (i.e., not estimated in the M-step)}
  \label{fig:difvar}
\end{figure}

\begin{figure}[H]
  \begin{center}
    \includegraphics[width=1.0\textwidth]{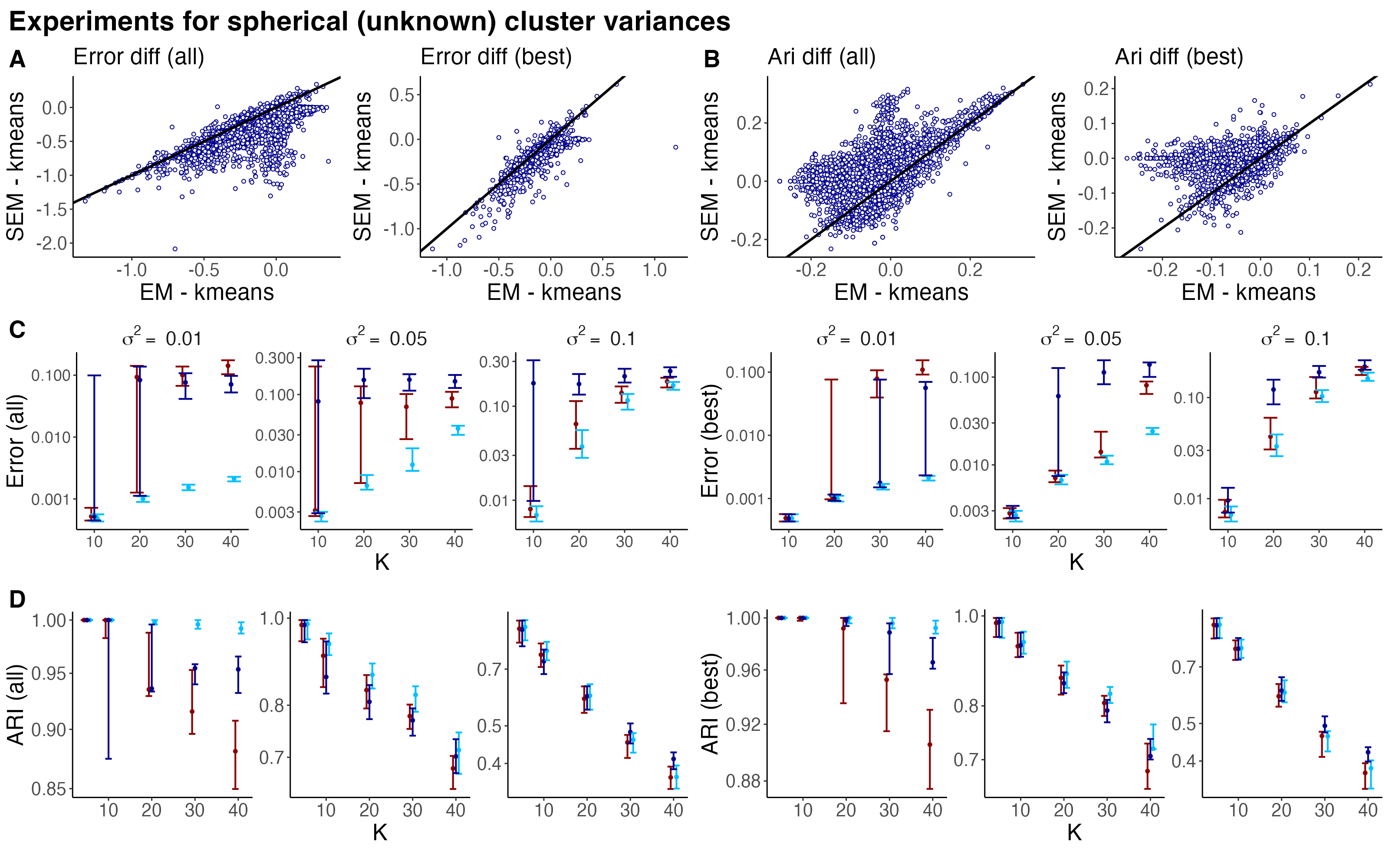}
  \end{center}
  \caption{Results for the experiment (iii) in Section \ref{sec:experiments}: each cluster has the same covariance $\sigma^2I_d$, but $\sigma^2$ is not known and hence must be estimated in the $M$-step.}
  \label{fig:unk1}
\end{figure}

\begin{figure}[H]
  \begin{center}
    \includegraphics[width=1.0\textwidth]{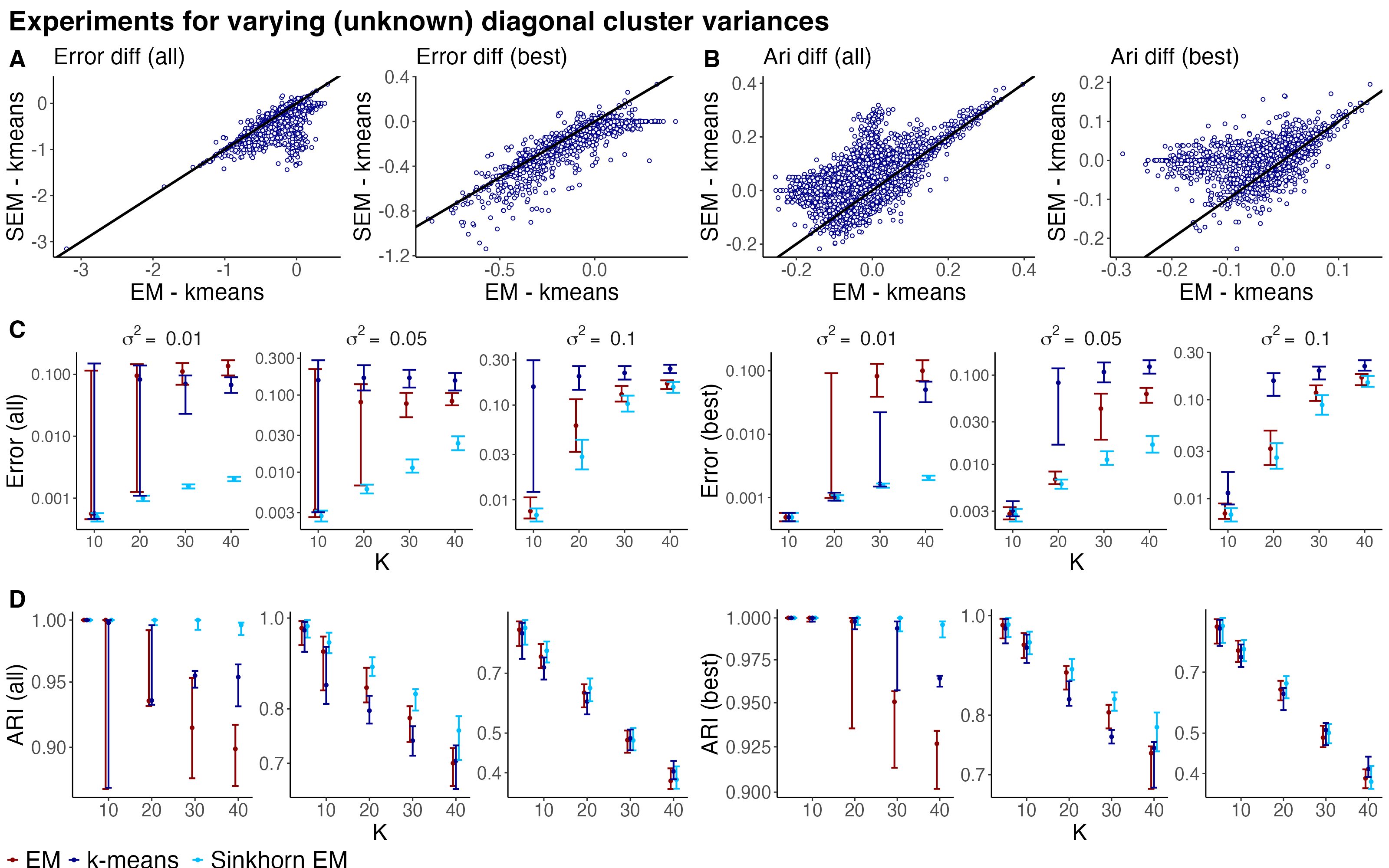}
  \end{center}
  \caption{Results for the experiment (iv) in Section \ref{sec:experiments} (main text): each cluster is characterized by a different diagonal covariance matrix and these variances are not assumed known and hence estimated in the $M$-step.}
  \label{fig:unk2}
\end{figure}

\begin{figure}[H]
  \begin{center}
    \includegraphics[width=1.0\textwidth]{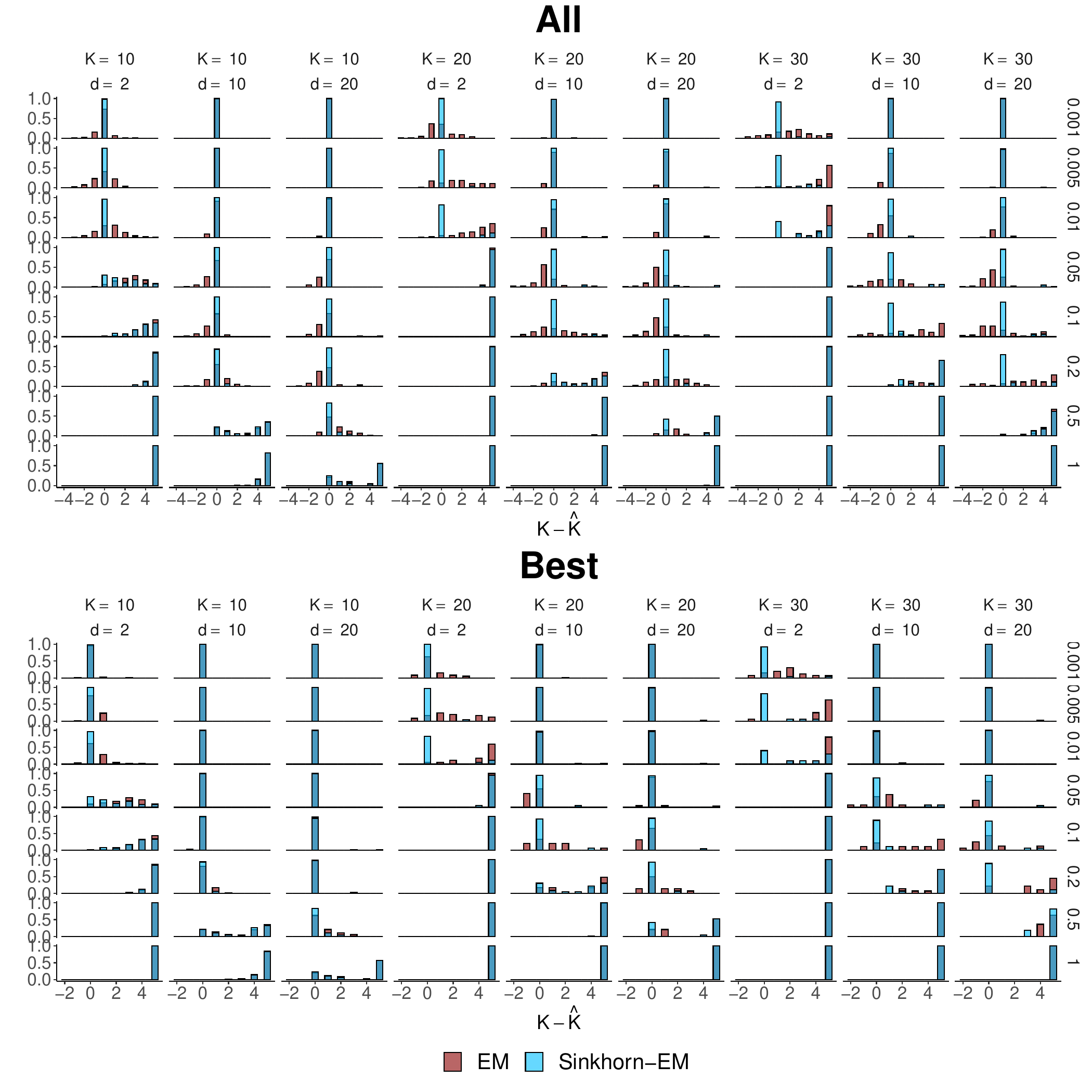}
  \end{center}
  \caption{Number of components estimation error histograms for the experiment in Section \ref{sub:exp3}. Each row represents a different variance $\sigma^2$. \textbf{A}: results for different seeds. \textbf{B}: results for the best seed. }
  \label{fig:errorselectionext}
\end{figure}
\subsection{C.elegans experiment details}
\label{app:eleg}
Results displayed in Fig. \ref{fig:celegans} in Section \ref{sec:celegans} correspond to the following semi-synthetic setup: we considered a vector of true neural locations and colors as described in \cite{nejatbakhsh2020probabilistic}. We considered clustering segmentation tasks in the head (195 neurons) and tail (45 neurons). In each experiment, we sampled 35 neurons at random either from the head or the tail. We sampled covariance matrices whose diagonal entries were drawn from a log-normal distribution with location parameter 1 and scale parameter 0.1. We considered GMM with 35 samples centered at sampled neural locations prescribed covariances and uniform weights, and based our inferences on  $N=5000$ samples from that mixture. In this experiment, we assumed the weights were fixed and well-specified, but we allowed the (diagonal) covariance matrices to be fitted. Results displayed in Fig. \ref{fig:celegans} are summaries of $200\times 2$ (head and tail) experiments; for each experiment, we retained the best of 10 seeds. As in the experiments of Section \ref{sec:experiments} (main text), we compared k-means, EM, and Sinkhorn-EM algorithms, in all cases using the k-means++ initialization method. Examples of segmentation performance in Fig. \ref{fig:celegans}B were plotted with the code released with \cite{nejatbakhsh2020probabilistic}. Finally, in Fig. \ref{fig:timelegans} we include a execution time comparison for the experiment described above. While EM typically takes less than a second, execution times for SEM are in the order of tens of seconds.
\begin{figure}[H]
  \begin{center}
    \includegraphics[width=0.5\textwidth]{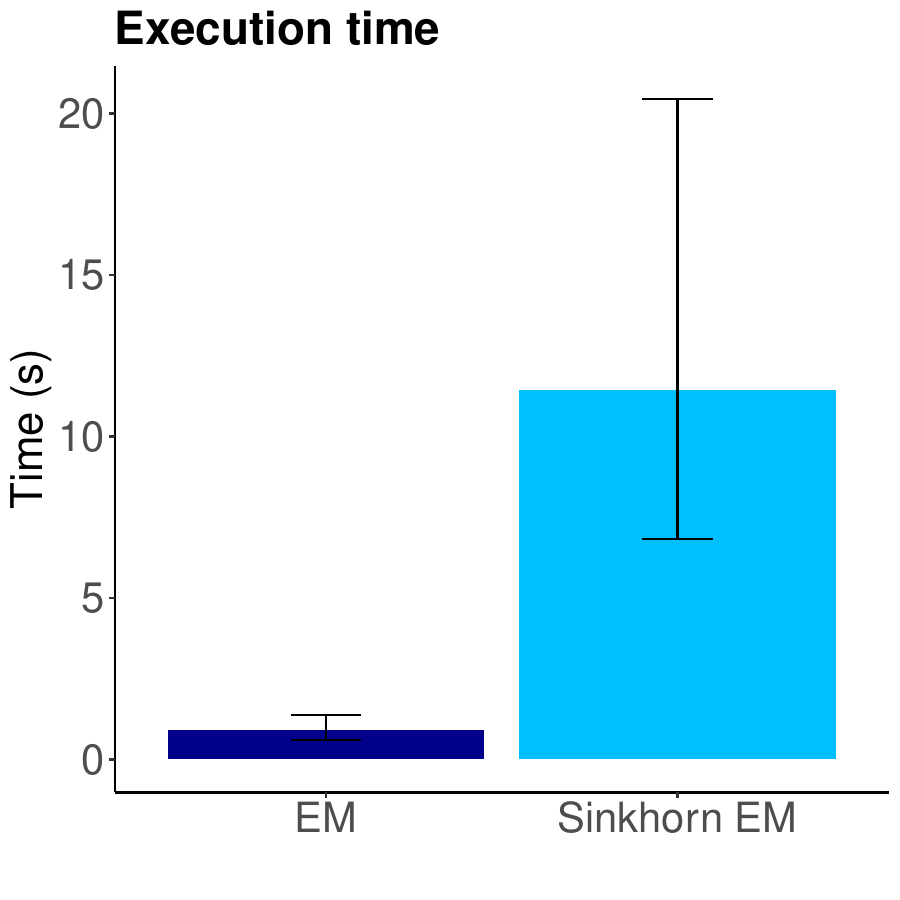}
  \end{center}
  \caption{Execution time comparisons for the C.elegans segmentation task.}
  \label{fig:timelegans}
\end{figure}

\subsection{Co-clustering experiment details}
\label{app:coclust}
\subsubsection{Algorithmic details:}
Suppose we want to perform co-clustering on an $N\times M$ matrix using $K\times G$ co-clusters. We consider row and column responsibility vectors $z\in\mathbb{R}^{N\times K}$  and $w\in\mathbb{R}^{M\times G}$. We briefly describe the VEM and SVEM algorithms. We implement the VEM algorithm as described in Algorithm 5.3 in Chapter 5 of \cite{govaert2013co}, which we reproduce in Algorithm \ref{alg:1}.

\begin{algorithm}[H]
  \caption{Variational EM (VEM) for co-clustering in model \eqref{eq:coclust}}
  \label{alg:1}
  \begin{enumerate}
    \item
    \textbf{Input}: Data matrix $Y$, number of co-clusters $K$ and $G$. Initial values for $z$ and $w$ 
    \item \textbf{Initialization}: compute
    $$\pi_k=\frac{\sum_k z_{i,k}}{N},\rho_g =\frac{\sum_{j}w_{j,g}}{M},\theta_{k,g} = \frac{\sum_{i,j}z_{i,k}Y_{i,j} w_{j,g}}{\sum_i z_{i,k}\sum_j w_{j,g}},\sigma^2_{k,g}=\frac{\sum_{i,j}z_{i,k}Y^2_{i,j} w_{j,g}}{\sum_i z_{i,k}\sum_j w_{j,g}}-\theta_{k,g}$$
    \item
    \textbf{Outer loop: repeat until convergence in $\theta$,$\sigma^2$,$\pi,\rho$)}
     \begin{itemize}
     \item[3.1] \textbf{Define}  $$Y^w_{i,g}=\frac{\sum_j w_{j,g}Y_{i,j}}{\sum_j w_{j,g}},u_{i,g}^w=\frac{\sum_j w_{j,g}Y^2_{i,j}}{\sum_j w_{j,g}}$$

      \item[3.2] \textbf{Inner loop in row ($k$) coordinates until convergence of $\pi,\mu,\sigma^2$}
      \begin{itemize}
      \item[3.2.1] \textbf{Update} $$z_{i,k}\propto \pi_k \exp\left(-\frac{1}{2}\sum_{g} \left(\sum_j w_{j,g}\right)\left(\log \sigma^2_{k,g}+\frac{u_{i,g}^w-2\mu_{k,g}Y^w_{i,g}+\mu_{k,g}^2}{\sigma^2_{k,g}}\right)\right)$$
      \item[3.2.2] \textbf{Update} $$\pi_k=\frac{\sum_{i}z_{i,k}}{N},\mu_{k,g}=\frac{\sum_{i} z_{i,k}Y^w_{i,g}}{\sum_{i}z_{i,k}}, \sigma^2_{k,g}=\frac{\sum_{i} z_{i,k}u^w_{i,g}}{\sum_{i}z_{i,k}}-\mu^2_{k,g}$$
      \end{itemize}
      \item[3.3] \textbf{Define}  $$Y^z_{j,k}=\frac{\sum_i z_{i,k}Y_{i,j}}{\sum_i z_{i,k}},v_{j,k}^z=\frac{\sum_j z_{i,k}Y^2_{i,j}}{\sum_i z_{i,k}}$$
      \item[3.4] \textbf{Inner loop in column ($g$) coordinates until convergence of $\rho,\mu,\sigma^2$}
      \begin{itemize}
      \item[3.4.1] \textbf{Update} $$w_{j,g}\propto \rho_g \exp\left(-\frac{1}{2}\sum_{k} \left(\sum_i z_{i,k}\right)\left(\log \sigma^2_{k,g}+\frac{v_{j,k}^z-2\mu_{k,g}Y^z_{j,k}+\mu_{k,g}^2}{\sigma^2_{k,g}}\right)\right)$$
      \item[3.4.2] \textbf{Update}  $$\rho_g=\frac{\sum_{j}w_{j,g}}{M},\mu_{k,g}=\frac{\sum_{j} w_{j,g}Y^z_{j,k}}{\sum_{j}w_{j,g}}, \sigma^2_{k,g}=\frac{\sum_{j} w_{j,g}v^z_{j,k}}{\sum_{j}w_{j,g}}-\mu^2_{k,g}$$
      \end{itemize}
      \end{itemize}

  \item
  \textbf{Output} $\pi,\rho,\sigma^2,\mu$
  \end{enumerate}
\end{algorithm}
In the case that some parameters (e.g., the weights $\pi,\rho$ or variances $\sigma^2$) are known, then the algorithm can be simplified by simply skipping the steps that implement updates of those parameters. 
The SVEM algorithm differs slightly from the above routine. In the case where $\pi$ and $\rho$ are fixed then we only need to replace steps 3.2.1 and 3.4.1 in Algorithm \ref{alg:1} by the computation of an entropic optimal transport plan between $\pi$ (resp, $\rho$) and a uniform measure with weights $1/N$ (resp, $1/M$) and cost function given by what is inside of the exponential term in 3.2.1 (resp, 3.4.1). This computation can be interpreted as transport towards a weighted version of $Y$, and we omit the details for simplicity. The computed optimal transport plan immediately yields the responsibilities $z$ (resp. $w$) as described in the main text. 

In the case where $\rho$ and $\pi$ are also parameters, we considered a variation over the above scheme where, on a small proportion of times, we carried out usual VEM updates in the inner loops (so that we can update weights) and in the rest of the times we performed the SVEM updates described in the above paragraphs. In our applications, we performed one VEM update every 6 SVEM updates. We could have alternatively considered weight updates as described in \ref{sec:alpha}, but we did not deem this necessary.

\subsubsection{Synthetic experiments:}
\sloppy For the synthetic data experiments in Fig. \ref{fig:biclustering} we considered $K=5,G=5$ or $K=10,G=10$ co-clusters, noise variances $\sigma^2=1.0,2.5$ $,5.0,7.5,10.0,12.5,15.0,17.5$ and sample sizes $N=M=100,500$ (i.e. the data matrix has $N$ rows and $N$ columns). Each co-cluster mean $\theta_{k,g}$ was sampled from a uniform distribution in the interval $[-5,5]$. For each parameter configuration, we considered a total of 40 experiments, and in each experiment, a total of 5 random seeds. To sample the data matrix $Y$, we divided the $N\times M$ entries into subsquares of size $N/K\times N/G$ each. We sampled entries $Y_{i,j}$ as a noisy version of the corresponding co-cluster mean, i.e. $Y_{i,j}=\mu_{k,g}+\mathcal{N}\left(0,\sigma^2\right)$ if where $(k,g)$ are indexes for the subsequent where indexes $(i,j)$ belong to. In this experiment we kept weights $\rho,\pi$ and variances $\sigma^2$ and only optimized over mean parameters $\theta$

As initial values for VEM and Sinkhorn-VEM, we used the solution provided by the spectral method with random initialization as implemented in Scikit-learn \cite{pedregosa2011scikit}. 
\subsubsection{Spatial Transcriptomic experiment details} We used the Prefrontal Dorsolateral Cortex spatial transcriptomic data from \citep{maynard2021transcriptome} available in the R package \cite{spatialibd}. We considered a subset from the original dataset, sample id 151673: we focused on the expression of the top 1000 genes (out of 33538) from the sample. We focused on the region with X coordinates between 139 and 400 (the original ranges were 139 and 498) and Y coordinates between -521 and -292 (the original ranges were -521 and -109). This subregion contains the majority of measurements anatomically associated with Layers 5, 6, and white matter (WM), and we only consider these layers for clustering purposes. In total, we clustered the resulting $N=1473$ spatial measurements out of the original 3639.

We compared the performance of VEM (Algorithm \ref{alg:1}), SVEM, and the spectral algorithm \cite{pedregosa2011scikit} with `scale' initialization. In all cases, we used $K=3$ row (spatial) clusters and $G=8$ column (gene) clusters. As initialization for VEM and SVE,M we used randomly sampled binary assignment matrices $z,w$. For VEM and SVEM algorithms, we optimized over mean $\theta$, variance $\sigma$, and weights $\rho,\pi$. In the case of SVEM, weight updates alternated with 

Since solutions from SVEM seemed to be stable to changes in random seed, to estimate variation in performance, we considered a number of 100 experiments where each time, we subsampled the original setup with a subsampling rate of 0.5. Bar plots in Fig. \ref{fig:spatialtrans}B are averages over these 100 repetitions.

\end{appendices}
\putbib
\end{bibunit}
\end{document}